\documentclass[aps,letterpaper,nofootinbib,showpacs,preprintnumbers,amsmath,amssymb]{revtex4}
\usepackage{hyperref}
\usepackage{graphicx}
\usepackage{amsmath}
\usepackage{amsthm,amssymb,amstext,amscd,amsfonts,amsbsy,amsxtra,latexsym}
\usepackage{wrapfig}
\usepackage{color}
\usepackage{pdflscape}
\usepackage{lscape}
\usepackage{tabularx}
\usepackage{lipsum}
\usepackage{rotating} 
\usepackage{pifont}
\usepackage{bbding}
\usepackage{upgreek}
\usepackage{mathrsfs}
\usepackage{wasysym}
\usepackage{verbatim}
\usepackage[paper=a4paper,left=35mm,right=35mm,top=25mm,bottom=25mm]{geometry}

\newcounter{mnotecount}[section]

\renewcommand{\themnotecount}{\thesection.\arabic{mnotecount}}

\newcommand{\mnote}[1]
{\protect{\stepcounter{mnotecount}}$^{\mbox{\footnotesize $%
\!\!\!\!\!\!\,\bullet$\themnotecount}}$ \marginpar{
\raggedright\tiny\em $\!\!\!\!\!\!\,\bullet$\themnotecount: #1} }

\newcommand{\lb}{\label}
\newcommand{\be}{\begin{equation}}
\newcommand{\ee}{\end{equation}}
\newcommand{\ben}{\begin{eqnarray*}}
\newcommand{\een}{\end{eqnarray*}}
\newcommand{\bea}{\begin{eqnarray}}
\newcommand{\eea}{\end{eqnarray}}
\newcommand{\md}{{\mathrm{d}}}
\newcommand{\beq}{\begin{equation}}
\newcommand{\eeq}{\end{equation}}
\newcommand{\beqn}{\begin{equation}\nonumber}
\newcommand{\bean}{\begin{eqnarray}\nonumber}
\DeclareMathAlphabet{\mathpzc}{OT1}{pzc}{m}{it}

\def\cE{{\mathcal E}}
\def\cR{{\mathcal R}}
\def\cH{{\mathcal H}}

\def\cB{{\mathcal B}}
\def\cN{{\mathcal N}}
\def\cO{{\mathcal O}}

\def\cQ{{\mathcal Q}}

\def\cD{{\mathcal D}}
\def\cL{{\mathcal L}}
\def\cM{{\mathcal M}}

\def\cC{{\mathcal C}}

\def\bbN{{\mathbb N}}

\def\cI{{\mathcal I}}

\def\cO{{\mathcal O}}

\def\bbS{{\mathbb S}}
\def\dV{{\mbox{dVol}}}
\def\schere{\text{\ding{34}}}
\newtheorem{thm}{Theorem}[section]

\newtheorem{lem}[thm]{Lemma}
\newtheorem{prop}[thm]{Proposition}
\newtheorem{cor}[thm]{Corollary}

\newtheorem{rmk}[thm]{Remark}
\renewenvironment{proof}[1][Proof]{\begin{trivlist}
\item[\hskip \labelsep {\bfseries #1:}]}{\qed\end{trivlist}}

\newcommand{\nabb}{\mbox{$\nabla \mkern-13mu /$\,}}
\newcommand{\Dell}{\mbox{$\Delta \mkern-13mu /$\,}}
\newcommand{\gin}{\mbox{$g \mkern-8mu /$\,}}

\def\XXint#1#2#3{{\setbox0=\hbox{$#1{#2#3}{\int}$ }
\vcenter{\hbox{$#2#3$ }}\kern-.6\wd0}}

\def\thesection{\arabic{section}}

\makeatletter
\def\p@subsection{}
\def\p@subsubsection{}
\def\p@paragraph{}
\makeatother

\begin{document}

\begin{center}

{\bf {\Large 
BOUNDEDNESS OF MASSLESS SCALAR WAVES ON KERR INTERIOR BACKGROUNDS}}\\

\bigskip
\bigskip
Anne T. Franzen\footnote{e-mail address: anne.franzen@tecnico.ulisboa.pt}

\bigskip
\bigskip
{\it Center for Mathematical Analysis, Geometry and Dynamical Systems,}\\
{\it Mathematics Department, Instituto Superior T\'ecnico,}\\ 
{\it Universidade de Lisboa, Portugal}

\end{center}
\medskip

\centerline{ABSTRACT}

\noindent
We consider solutions of the massless scalar wave equation $\Box_g\psi=0$, without symmetry, on fixed subextremal Kerr backgrounds $(\cM, g)$.
It follows from previous analyses in the Kerr exterior that for solutions $\psi$ arising from
sufficiently regular data on a two ended Cauchy hypersurface, the solution and its derivatives decay suitably fast along
the event horizon $\cH^+$.
Using the derived decay rate, we show that $\psi$ is in fact uniformly bounded, $|\psi|\leq C$, in the black hole interior
up to and including the bifurcate Cauchy horizon $\cC\cH^+$, to which $\psi$ in fact extends continuously.
In analogy to our previous paper, \cite{anne},
on boundedness of solutions to the massless scalar wave equation on fixed subextremal Reissner--Nordstr\"om backgrounds,
the analysis depends on weighted energy estimates, commutation by angular momentum operators and application of Sobolev embedding. 
In contrast to the Reissner--Nordstr\"om case the commutation leads to additional error terms that have to be controlled.
\medskip

\tableofcontents

\section{Introduction}
\lb{intro}
The Kerr spacetime $(\cM,g)$ is a fundamental axialsymmetric 2-parameter family of solutions to the vacuum Einstein field equations. A brief introduction to the spacetime is given in \cite{haw_ellis}, and for a more detailed discussion see \cite{oneill}. The Penrose diagram of the subextremal case, \mbox{$M>|a|\neq 0$}, with $a$ the angular momentum per unit mass and $M$ the mass of the black hole, is shown in Figure \ref{kerrfigure}. 
The problem of analyzing the solution of the scalar wave equation 
\bea
\lb{wave_psi}
\Box_g \psi=0
\eea 
on Kerr backgrounds is intimately related to the stability properties of the spacetime itself and to the celebrated Strong Cosmic Censorship Conjecture, see \cite{christo_sing} and \cite{penrose} for a presentation of the conjecture. 
{\begin{figure}[ht]
\centering
\includegraphics[width=0.5\textwidth]{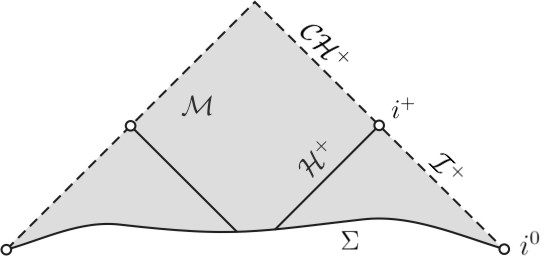}
\caption[]{Penrose diagram of the maximal future development of a Cauchy hypersurface $\Sigma$ in Kerr spacetime $(\cM,g)$. What is depicted is in fact the range of a double null coordinate system which is global in the interior and will be discussed further in Section \ref{doublenull}.}
\label{kerrfigure}\end{figure}}

The analysis of \eqref{wave_psi} for the exterior Kerr region $J^-(\cI^+)$, in the {\it full subextremal range $|a|<M$}, has already been accomplished in \cite{m_kerr}. The purpose of the present work is to extend the investigation to the interior of the black hole, up to and including the Cauchy horizon $\cC\cH^+$.
The mathematical structure and notation of this paper closely follows our work on Reissner--Nordstr\"om backgrounds, see \cite{anne}.

As already mentioned, the analysis of \eqref{wave_psi} is motivated by the Strong Cosmic Censorship Conjecture. In order to investigate its validity we need to understand stability and instability properties of black hole interiors. 
A brief overview on this topic as well as a mathematical formulation of the conjecture were already given in Section 1.2 of \cite{anne}. References discussing the instability behavior in similar settings are \cite{luk_oh, luk_jan}. Also refer to \cite{simpson} for a numerical analysis and \cite{poisson,amos,brady} for early discussions of heuristic models.

\subsection{The main result}
\lb{mainresult}
The main result of this paper can be stated as follows.
\begin{thm}
\lb{main}
On subextremal Kerr spacetime $(\cM,g)$, with mass $M$ and angular momentum per unit mass $a$ and $M>|a|\neq 0$, let $\psi$ be 
a solution of the wave equation $\Box_g \psi=0$ arising from
sufficiently regular localized Cauchy data on a two-ended asymptotically flat Cauchy surface $\Sigma$. Then
\bea
\lb{maineq}
|\psi|\leq C
\eea
globally in the black hole interior, in particular up to and including the Cauchy horizon
$\cC\cH^+$, to which $\psi$ extends in fact continuously.
\end{thm}
The constant $C$ is explicitly computable in terms of parameters $a$ and $M$ and a suitable norm on initial data.
In order to prove the above theorem, we will first derive weighted energy boundedness, as expressed in the following theorem.
\begin{thm}
\lb{energythm}
On subextremal Kerr spacetime $(\cM,g)$, with mass $M$ and angular momentum per unit mass $a$ and $M>|a|\neq 0$, let $\psi$ be 
a solution of the wave equation $\Box_g \psi=0$ arising from
sufficiently regular localized Cauchy data on a two-ended asymptotically flat Cauchy surface $\Sigma$. Then, in the black hole interior we have
\bea
\lb{energy1}
\int\limits^{\infty}_{v_{fix}}\int\limits_{\bbS^2_{u,v}}\left[ v^p (\partial_v \psi+b^{\tilde{\phi}} \partial_{\tilde{\phi}} \psi)^2(u,v, \theta^{\star}, \tilde{\phi}) + \Omega^2|\nabb \psi|^2(u,v, \theta^{\star}, \tilde{\phi}) \right]\md \sigma_{\mathbb S}^2 \md v&\leq& E,\\
 \quad \mbox{for $v_{fix} \geq 1$, $u > -\infty$}&&\\
 \lb{energy2}
\int\limits^{\infty}_{u_{fix}}\int\limits_{\bbS^2_{u,v}}\left[ u^p (\partial_u \psi)^2 (u,v, \theta^{\star}, \tilde{\phi})+ \Omega^2|\nabb \psi|^2(u,v, \theta^{\star}, \tilde{\phi}) \right]\md \sigma_{\mathbb S}^2\md u&\leq& E, \\
\quad \mbox{for $u_{fix} \geq 1$, $v > -\infty$},&&
\eea
where $p>1$ is an appropriately chosen constant, and $(u,v, \theta^{\star}, \tilde{\phi})$ are Eddington--Finkelstein normalized double null coordinates defined in Section \ref{doublenull}. The functions $\Omega^2(u,v, \theta^{\star})$ and $b^{\tilde{\phi}}(u,v, \theta^{\star})$ are defined by \eqref{R_delta_sigma} and \eqref{b_phi}, respectively, together with the metric \eqref{kerrmetric}. The expression 
\bea
\lb{sigmadef}
\md \sigma_{\mathbb S}^2=\sin\theta\md \theta^{\star} \md \tilde{\phi}
\eea
represents the volume element on the two-spheres ${\mathbb S}^2$, which due to the dependence of $\theta$ on $(u,v,\theta^{\star})$, stated in \eqref{dependence}, do not represent round spheres. Moreover, ${\mathbb S}^2_{u,v}$ are the two-spheres, which are defined by the intersection of the level sets of $u$ and $v$, we will discuss them more in Section \ref{doublenull} when we introduce the double null coordinates\footnote{We will later introduce a function $L$ which quantifies the difference between the volume elements of these spheres, see Section \ref{doublenull} and also \eqref{suv_vol}. Since $L$ is bounded we are able to omit it in the above statements.}. Further, we denote
\bea
\lb{nabb}
|\nabb \psi|^2=(\gin^{-1})^{\theta_A\theta_B}(\partial_{\theta_A} \psi \partial_{\theta_B} \psi).
\eea
\end{thm}
{\em Remark.} Statement \eqref{energy1} for $u = -\infty$ and statement \eqref{energy2} for $v = -\infty$ hold by previous work on the exterior, see \cite{m_kerr, george}.

Having obtained Theorem \ref{energythm}, commutation by angular momentum operators, using the wave equation itself and estimating error terms of the form \eqref{E} leads to a higher order version of the above theorem. The presence of error terms is a difference from the spherically symmetric Reissner--Nordstr\"om case, see \cite{anne}. This will be explained in Section \ref{vectorfieldmethod}. The pointwise boundedness of Theorem \ref{main} will then follow from the higher order energy theorem and applying Sobolev embedding.

\subsection{A first look at the analysis}
\lb{firstlook}
In the analysis of this paper we are going to use global double null coordinates $(u,v,\theta^{\star},\tilde{\phi})$, which were derived by Pretorius and Israel \cite{pretorius} and will be discussed in Section \ref{doublenull}.
To prove Theorem \ref{main} and \ref{energythm} we first consider a characteristic rectangle $\Xi$ within the black hole interior, whose future {\it right} boundary coincides with the Cauchy horizon $\cC\cH^+$ in the vicinity of $i^+$, cf.~ Figure \ref{alle_bnr} a), and whose past {\it right} boundary coincides with the event horizon $\cH^+$. 
{\begin{figure}[ht]
\centering
\includegraphics[width=0.7\textwidth]{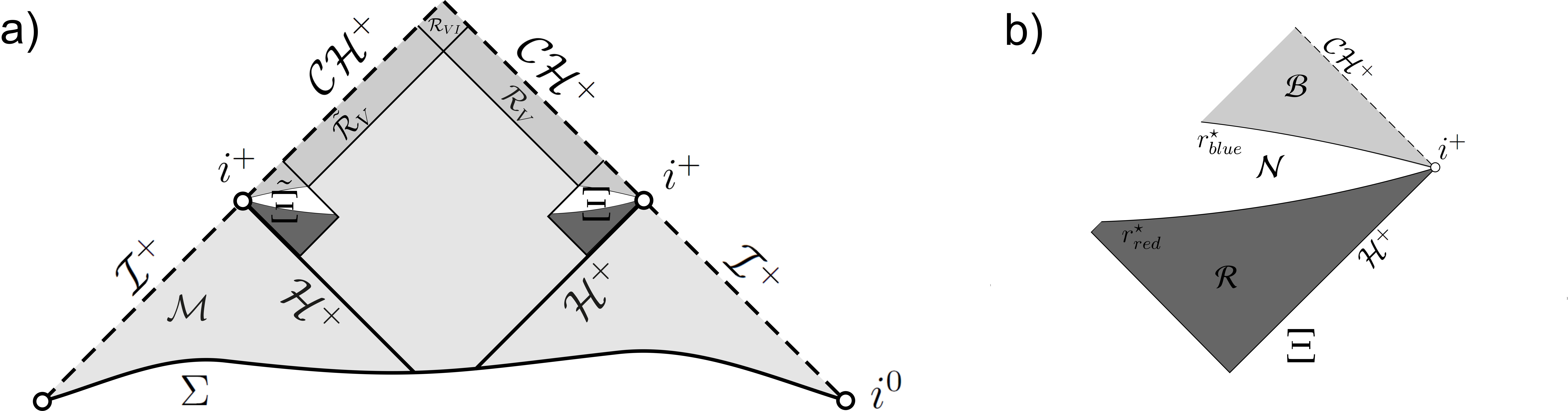}
\caption[]{a) Penrose diagram of Kerr spacetime depicting the regions considered in the proof. b) Characteristic rectangle $\Xi$, with redshift $\cR$, noshift $\cN$ and blueshift regions $\cB$.}
\label{alle_bnr}\end{figure}}

The crux of the entire proof is establishing boundedness of weighted higher order energy norms in $\Xi$. 
Once that is done, analogous results hold for a characteristic rectangle $\tilde{\Xi}$ to the {\it left}, depicted in Figure \ref{alle_bnr} a). Hereafter, boundedness of the energy can be propagated to regions $\cR_V$, $\tilde{\cR}_V$ and $\cR_{VI}$ as depicted, proving Theorem \ref{energythm}. Commutation by angular momentum operators, controlling the error terms and application of Sobolev embedding then yields Theorem \ref{main}. For brevity and simplicity we will in fact immediately carry out the commutation and control of error terms in each region. This is to say we directly derive higher order energy estimates for each region respectively.

We will now return to the discussion of $\Xi$, since that is the most involved part of the proof. 
Starting from an upper decay bound for $|\psi|$ and its derivatives on the event horizon $\cH^+$, we will
prove Theorem \ref{energythm} and its higher order version (and hence Theorem \ref{main}) restricted to $\Xi$. This upper bound along $\cH^+$ can be deduced from the work of Dafermos et al., cf.~ \cite{m_kerr} and \cite{george}. We will state the precise required result from previous work in Section \ref{horizon_estimates}.

In region $\Xi$ the proof involves distinguishing redshift \mbox{$\cR=\left\{-\infty < r^{\star} \leq r^{\star}_{red}\right\}$}, noshift \mbox{$\cN=\left\{r^{\star}_{red}\leq r^{\star}\leq r^{\star}_{blue}\right\}$} and blueshift \mbox{$\cB=\left\{r^{\star}_{blue}\leq r^{\star}< \infty\right\}$} regions, as shown in Figure \ref{alle_bnr} b).

These regions have appeared in the analysis of the wave equation on Reissner--Nordstr\"om backgrounds, cf.~ \cite{anne} and references therein.
The analysis is analogous to the previous result, but requires slightly different vector field multipliers and control of error terms. The reader familiar with \cite{anne} will recognize the structure of the previous proof, which we have maintained here for better readability.
A more detailed discussion of the separation into $\cR$, $\cN$ and $\cB$ regions is given in Section \ref{bnrsection} and Section \ref{statement}.
In region $\cB$, which is adjacent to $\cC\cH^+$, the weighted vector field 
\ben
S=|u|^p\partial_u+v^p\partial_v+v^p b^{\tilde{\phi}}\partial_{\tilde{\phi}} 
\een
in Eddington--Finkelstein-like coordinates $(u,v)$ is used for the analysis.
Note the parameter $p>1$ appearing in Theorem \ref{energythm}. The weights in this vector field, associated to region $\cB$, will allow us to prove uniform boundedness despite the blueshift instability.

\subsection{Outline of the paper}
The remaining part of the paper is organized as follows. 

Section \ref{preliminaries} contains all mathematical tools required in order to follow the analysis. In particular, we recall the vector field method in Section \ref{vectorfieldmethod}. Moreover, we specify the double-null, Eddington--Finkelstein-like coordinates, in which we will carry out the analysis, in Section \ref{kerr}, and introduce more notation in Section \ref{nota}. In Section \ref{setup} we give a brief review of the horizon estimates, that follow for the evolution of smooth compactly supported initial data on a Cauchy hypersurface by \cite{m_kerr}. We also give data on a null hypersurface transverse to the event horizon. Both of these statements together then allow us to propagate the energy estimate further inside, as shown in Section \ref{energy_interior}. 

Section \ref{energy_interior} contains the crux of the proof, in which we derive estimates inside a characteristic rectangle $\Xi$, as we already explained in Section \ref{firstlook}. The higher order propagation inside $\Xi$ is derived in Section \ref{red_region} to \ref{blue_future}, implying control over all error terms. Hereafter, we derive pointwise estimates inside $\Xi$ as well as higher order energy estimates in the remaining regions $\cR_V$, $\tilde{\cR}_V$ and $\cR_{VI}$, see Section \ref{nineten} to \ref{bifurcate}.

In Section \ref{globalenergy} we obtain the higher order weighted energy statement for the entire interior and derive pointwise boundedness by using Sobolev embedding on the spheres. We eventually give an outlook in Section \ref{outlook}, in which we put the present work into context and discuss related results.

\section{Preliminaries}
\lb{preliminaries}
\subsection{Vector field method and energy currents}
\lb{vectorfieldmethod}
In the following section we will briefly review the vector field method which we are going to use as an essential tool throughout this work.
For an introduction and the history of this method, see \cite{sergiu} by Klainerman.

The wave equation \eqref{wave_psi} can be derived from the following matter field Lagrangian 
\be
\lb{lagrangian}
\cL(\psi)=\int_{\cM} g^{\mu\nu} \partial_{\mu} \psi \partial_{\nu} \psi \dV,
\ee
on a spacetime manifold $(\cM, g)$.
A symmetric stress-energy tensor can be identified from the variation of the action as
\be
\lb{energymomentum}
T_{\mu\nu}(\psi)=\partial_\mu\psi\partial_\nu\psi-\frac12g_{\mu\nu}g^{\alpha\beta}\partial_\alpha\psi
\partial_\beta\psi.
\ee
Since $\psi$ is a solution to \eqref{wave_psi} energy-momentum conservation is implied, i.e.
\begin{equation}
\label{divfree}
\nabla^\mu T_{\mu\nu}=0.
\end{equation}

In this work we will often be interested in the inhomogenous wave equation
\be
\lb{inhomoKG}
\Box_g \tilde{\psi}=\Box_g Y{\psi}=F(Y{\psi}),
\ee
where $Y=Y^{\theta_C}\partial_{\theta_C}$ is a non-Killing commutation vector field, that will be explained more in Section \ref{angular}.
The motivation for this is that in order to obtain pointwise estimates we will have to commute with vector fields that are not Killing.

In the remaining section we will define one-currents as well as scalar currents depending on the geometry, which constitute a robust tool for obtaining $L^2$ estimates.
By contracting the energy-momentum tensor with a vector field multiplier $V$, we define the current
\be
\lb{J}
J_\mu^V(\tilde{\psi})\doteq T_{\mu\nu}(\tilde{\psi}) V^\nu.
\ee
If the vector field $V$ is timelike, then the one-form $J_\mu^V$ can be interpreted as an energy flux.

We will frequently apply the divergence theorem, often referred to as Stokes' theorem, to the above defined energy fluxes \eqref{J}. 
Consider for example a spacetime region $\mathcal{B}$ which is bound by two homologous hypersurfaces, $\Sigma_{\tau}$ and $\Sigma_0$, then we obtain
\begin{equation}
\label{divthe}
\int_{\Sigma_{\tau}} J^V_\mu (\tilde{\psi})n^\mu_{\Sigma_{\tau}} \dV_{\Sigma_{\tau}}
+\int_{\mathcal{B}} \nabla^\mu J_\mu(\tilde{\psi}) \dV=
\int_{\Sigma_0} J^V_\mu(\tilde{\psi}) n^\mu_{\Sigma_0} \dV_{\Sigma_0}.
\end{equation}
The vector $n^{\mu}_{\Sigma}$ denotes the normal to the subscript hypersurface $\Sigma$ oriented according to Lorentzian geometry convention. Further, $\dV$ denotes the volume element over the entire spacetime region and $\dV_{\Sigma}$ the volume elements on $\Sigma$, respectively.
Using the divergence theorem, we will often estimate the spacetime integral by the boundary terms; or future boundary terms from the sum of past boundary terms and the spacetime integral.
In the above example the spacetime integral refers to the second term of the left hand side of \eqref{divthe}, the future boundary to the first term and the past boundary to the term on the right hand side of the equation.
A proof of the divergence theorem in general can be found in \cite{frankel}; another useful reference is \cite{taylor}.

In view of the spacetime integral, we are interested in the divergence of the current \eqref{J} which reads as
\be
\lb{divergence}
\nabla^{\mu}J_{\mu}=\nabla^{\mu}(T_{\mu\nu}V^{\nu})=T_{\mu\nu}(\nabla^{\mu} V^{\nu})+(\nabla^{\mu}T_{\mu\nu})V^{\nu},
\ee
and suggests defining the following two scalar currents
\be
\lb{K}
K^V(\tilde{\psi})\doteq T(\tilde{\psi})(\nabla V)=(\pi^V)^{\mu\nu}T_{\mu\nu}(\tilde{\psi}),
\ee
where $(\pi^V)^{\mu\nu} \doteq \frac{1}{2} (\cL_V g)^{\mu\nu}$ is the so called deformation tensor
along $V$,
and
\be
\lb{E}
\cE^V(\tilde{\psi})\doteq (\nabla^{\mu}T_{\mu\nu})V^{\nu}=(\Box_g \tilde{\psi})V(\tilde{\psi}).
\ee
We will sometimes refer to $K^V$ as ``bulk term'' and to $\cE^V$ as ``error term''.
Thus
\be
\nabla^{\mu}J_{\mu}=K^V+\cE^V.
\ee
From \eqref{K} we see that $K^V$ is zero in case that our multiplier is a Killing vector field and $\cE^V$ is zero if $\tilde{\psi}$ is a solution to the homogeneous wave equation. Recall \eqref{inhomoKG} and suppose for arbitrary $\tilde{\psi}$ we write $\tilde{\psi}=Y\psi$, with $\psi$ being a solution to \eqref{wave_psi} and $Y$ an arbitrary vector field, then we see that $\cE^V$ is zero if $[\Box_g, Y]=0$, which is always the case when $Y$ is a Killing vector field. Since in this analysis we are interested in commuting with non-Killing vector fields the error terms will not vanish, but have to be controlled.

\subsection{The Kerr solution}
\lb{kerr}
In the following we will briefly recall Kerr spacetimes\footnote{The reader unfamiliar with this solution may for example consult \cite{haw_ellis} for a brief review, or the more detailed \cite{oneill}.} which are a family of solutions to the
Einstein vacuum field equations 
\bea
\lb{EF2}
R_{\mu\nu}=0,
\eea
where $R_{\mu \nu}$ is the Ricci tensor.
The Kerr solution represents an isolated rotating black hole in an asymptotically flat spacetime and was discovered in 1963, \cite{kerr}.

For our proofs of Theorem \ref{main} and \ref{energythm}, it will be favorable to carry out the analysis in double null coordinates. Therefore, in the following sections we will first recall the more widely known Boyer--Lindquist coordinates, to then derive the coordinate transformation leading to the double null form of the metric, first considered by Pretorius and Israel in \cite{pretorius}. 

\subsubsection{The metric, ambient differential structure and Killing vector fields}
\lb{ambient}
To set the semantic convention, whenever we refer to the Kerr solution $(\cM,g)$ we mean the maximal domain of dependence \mbox{$\cD (\Sigma)=\cM$} of complete two-ended asymptotically flat data $\Sigma$ for the parameter range $0<|a|<M$. 
The manifold $\cM$ can be expressed by \mbox{$\cM=\cQ\times \bbS^2_{u,v}$}, where $\bbS^2_{u,v}$ is defined as the intersection of the level sets of $\tilde{u}$ and $\tilde{v}$, for a representation of $\cQ$ see Figure \ref{diff_structure}. The submanifold \mbox{$\cM|_{II}=\cQ|_{II}\times \bbS^2_{u,v}$}, which is the region of interest, admits global double null coordinates $(u,v,\theta^{\star}, \tilde{\phi})$, to be defined in Section \ref{doublenull}, so that we have
{\begin{figure}[ht]
\centering
\includegraphics[width=0.35\textwidth]{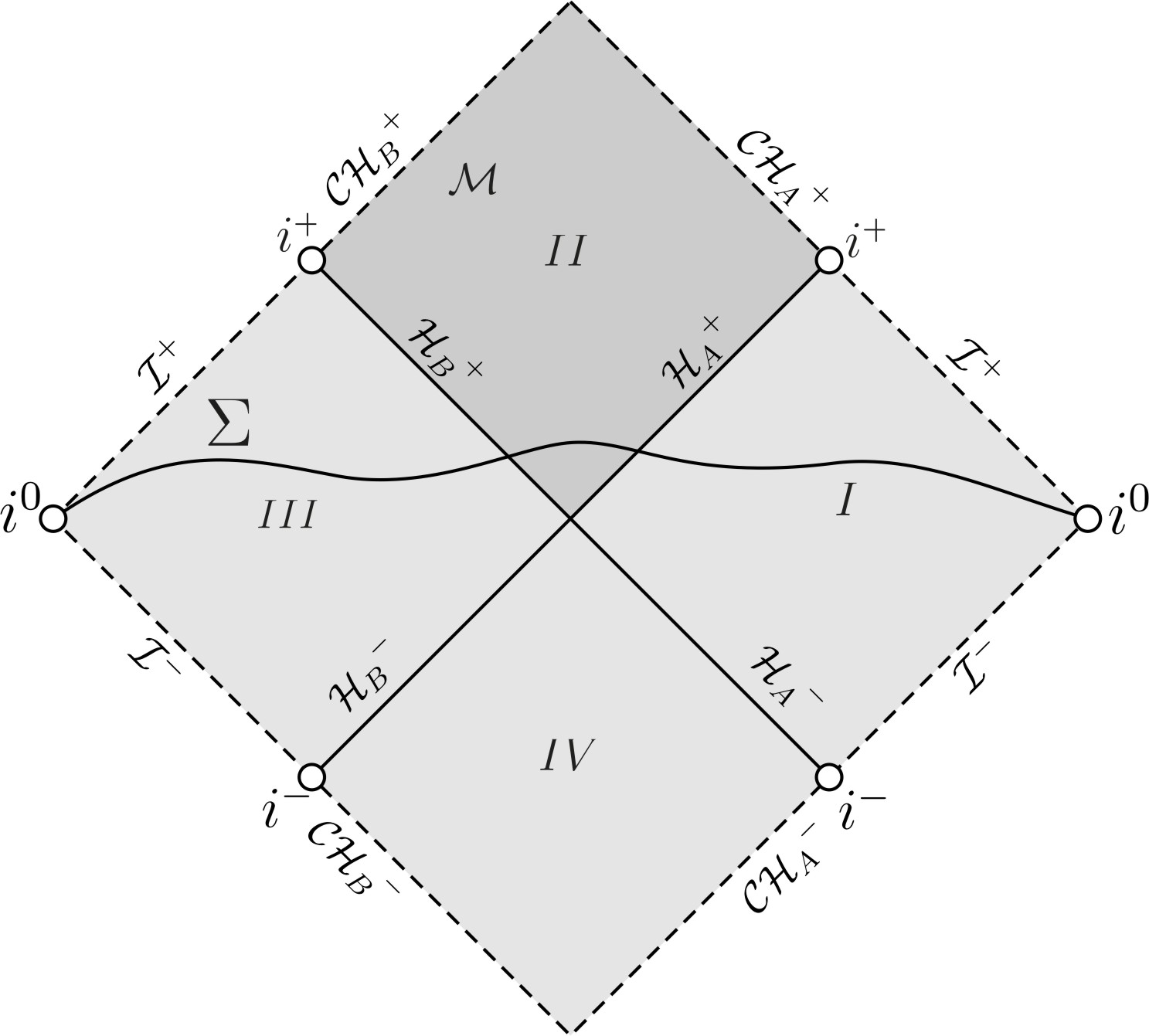}
\caption[]{Penrose diagram of maximal domain of dependence of Kerr spacetime with the region of interest shown darker shaded.}
\label{diff_structure}\end{figure}}
\mbox{$\cQ|_{II}=[-\infty, \infty) \times [-\infty, \infty)$} with  \mbox{$u,v \in (-\infty, \infty)$} and thus
\bea
\lb{MII}
\cM|_{II}=[-\infty, \infty) \times [-\infty, \infty)\times \bbS^2_{u,v}.
\eea
We can also formally parametrize the event horizon, the past boundary of $\cQ|_{II}$, namely $\cQ|_{II}\backslash\mathring{\cQ}|_{II}$ by
\bea
\cH^+=\cH_A^+\cup\cH_B^+
=\left\{-\infty\right\}\times[-\infty, \infty)\cup
[-\infty, \infty) \times \left\{ -\infty \right\},
\eea
and the future boundary, $\overline{\cQ|_{II}}\backslash\cQ|_{II}$, which is a Cauchy horizon by 
\bea
\cC\cH^+=\cC\cH_A^+\cup\cC\cH_B^+=
\left\{\infty\right\}\times[-\infty, \infty)\cup
[-\infty, \infty) \times \left\{ \infty \right\},
\eea
where $\mathring{\cQ}|_{II}$ is the interior and $\overline{\cQ|_{II}}$ is the closure of $\cQ|_{II}$. The fact that $\cC\cH^+$ is not contained in $\cQ|_{II}$ is indicated in Figure \ref{range} by the dashed line. The significance of these horizons are discussed further in the following.
{\begin{figure}[ht]
\centering
\includegraphics[width=0.2\textwidth]{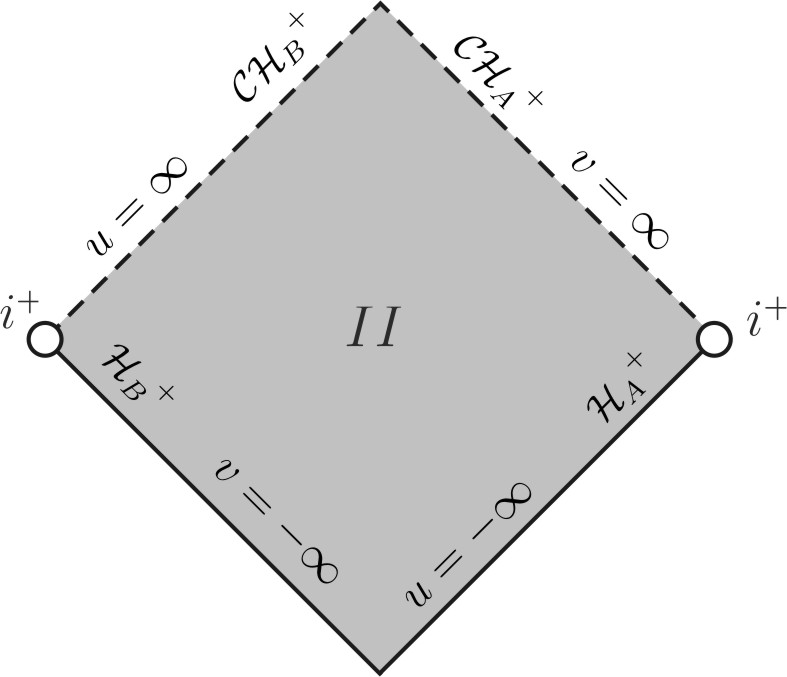}
\caption[]{Penrose diagram of region $II$, in which the double null coordinates $(u,v,\theta^{\star}, \tilde{\phi})$, to be constructed in Section \ref{doublenull}, are global.}
\label{range}\end{figure}}

Pretorius and Israel \cite{pretorius}
have already shown the coordinate transformation from the familiar Boyer--Lindquist coordinates to double null coordinates, see also Dafermos et al.~ \cite{m_rueck}. For self containment we will briefly repeat it in Section \ref{doublenull}. 

In order to do that, first recall the Kerr metric in Boyer--Lindquist coordinates $(t, r, \theta, \phi)$:
\bea
\lb{bl_metric}
g_{\mbox{\tiny{B-L}}}=\frac{\Sigma}{\Delta}\md r^2+\Sigma \md \theta^2 +R^2\sin^2\theta \md \phi^2-\frac{4Mar\sin^2\theta}{\Sigma}\md \phi \md t-\left(1-\frac{2Mr}{\Sigma}\right)\md t^2,
\eea
with 
\bea
\lb{sigma_usw}
\Delta=r^2-2Mr+a^2, \quad R^2=r^2+a^2+\frac{2Ma^2r\sin^2\theta}{\Sigma}, \quad \Sigma=r^2+a^2\cos^2\theta,
\eea
where $M$ is the mass and $a$ is the angular momentum per unit mass of the black hole. 
Moreover, note the useful identity
\bea
\lb{R_sigma}
\Sigma R^2=(r^2+a^2)^2-\Delta a^2\sin^2\theta.
\eea
Further, $\theta$ is the colatitude with 
\bea
\lb{theta_range}
0< \theta < \pi
\eea
and $\phi$ is the longitude with 
\bea
0\leq \phi < 2\pi,
\eea
so that this coordinate system covers the entire sphere except for the north and south pole.
Like in Reissner--Nordstr\"om spacetime we have an event horizon $\cH^+$ and a Cauchy horizon $\cC\cH^+$ bounding the interior region $II$. The Cauchy horizon is not part of the maximal domain of dependence which is indicated by the dashed line in Figure \ref{diff_structure}. Both horizons must be stationary null hypersurfaces. Since the normal to any stationary hypersurface is proportional to $\nabla r$ and lightlike character implies the condition \mbox{$g(\nabla r, \nabla r)=g^{\mu \nu}\partial_{\mu}r\partial_{\nu}r=g^{rr}=0$}, we obtain two positive roots 
\bea
\lb{r_pm}
r_{\pm}=M\pm\sqrt{M^2-a^2}
\eea
of $\Delta$. Relating this to our yet to be defined $(u,v)$ coordinates we can designate the position of the event horizon by
\bea
r(-\infty,v)=r|_{\cH_A^+}=r_+, \qquad r(u,-\infty)=r|_{\cH_B^+}=r_+,
\eea
and the location of the Cauchy horizon by
\bea
r(u,\infty,)=r|_{\cC\cH_A^+}=r_-, \qquad r(\infty, v)=r|_{\cC\cH_B^+}=r_-.
\eea
Despite merely having axisymmetry, the surface gravities can be shown to be independent of $\theta$, and thus constant along the horizons\footnote{This feature is predicted from the {\it zeroth law} of black hole thermodynamics, see \cite{toolkit}.}, namely
\bea
\lb{kappapm}
\kappa_{\pm}=\frac{r_{\pm}-M}{r_{\pm}^2+a^2}=\pm\frac{\sqrt{M^2-a^2}}{r_{\pm}^2+a^2},
\eea
where $\kappa_+$ is the surface gravity at $\cH^{\pm}$ and $\kappa_-$ is the surface gravity at $\cC\cH^{\pm}$, see \cite{toolkit} for further discussion.

Kerr spacetimes exhibit two Killing vector fields $\partial_t$ and $\partial_{\phi}$. The linear combination 
\bea
\lb{T_H}
T_{\cH^+}=\partial_t + \frac{a}{r_+^2 +a^2} \partial_{\phi},
\eea
here given in Boyer--Lindquist coordinates,
will play a particular role in the analysis inside the redshift region, see Section \ref{red_region}.
Using \eqref{r_pm} we can now also express region $\cQ_{II}$, the projected region of interest, in terms of a range in $r$ by 
\bea
\cQ_{II}=\left\{r_-<r\leq r_+\right\}.
\eea
We would now like to briefly clarify the causal properties of the Killing vector fields $\partial_t$, $T_{\cH^+}$ and the vector field $\partial_r$. In particular, in the interior we have
\begin{equation}
 \begin{cases}
    \partial_t & \mbox{spacelike in} \quad \mathring{\cQ}|_{II},\\
    T_{\cH^+} & \mbox{spacelike in} \quad \mathring{\cQ}|_{II},\quad \mbox{lightlike at}\quad \cH^+,\\
		 \partial_r & \mbox{timelike in}\quad \mathring{\cQ}|_{II},
  \end{cases}
\end{equation}
where $\mathring{\cQ}|_{II}=\cQ|_{II}\setminus \partial \cQ|_{II}$ is the interior of $\cQ|_{II}$.
Apart from the two Killing vector fields, there is also a Killing tensor, the so called Carter tensor which will not be used in our analysis.

\subsubsection{Defining suitable null coordinates}
\lb{r_star_section}
In our previous paper on Reissner--Nordstr\"om interiors, we chose
$t=v-u$ and $r^{\star}=u+v$,
where $r^{\star}$ was the Regge-Wheeler tortoise coordinate. We will see that the Kerr analog has to have a $\theta$ dependence, i.~e.~ $r^{\star}(r,\theta)$. Moreover, we will have to define a coordinate
\bea
\lb{thstar}
\theta^{\star}(r, \theta),
\eea
where $\theta^{\star}$ also has $r$ and $\theta$ dependence and $(t,r^{\star}, \theta^{\star}, \phi)$ forms a coordinate system in $(\cM,g)$.
All the sign differences compared to \cite{m_rueck} and \cite{pretorius} appearing in this section, arise from choices, which are appropriate for the analysis in the {\it interior} region as opposed to the {\it exterior} region of the black hole and are in agreement with the interior analysis of \cite{m_luk}. 

As explained in \cite{pretorius} and in more detail in Appendix A of \cite{m_luk}, we first define $\theta^{\star}\in [\theta, \frac{\pi}{2})$ implicitly in the following: For 
$r\in[r_-,r_+], \theta \in (0, \frac{\pi}{2})$, 
\bea
\lb{F_imp_def}
F(r, \theta, \theta^{\star})=\int\limits_{\theta^{\star}}^{\theta} \frac{\md \theta'}{P(\theta', \theta^{\star})}+\int\limits_r^{r_+}\frac{\md r'}{Q(r', \theta^{\star})}=0,
\eea
with 
\bea
\lb{PandQ}
P^2(\theta,\theta^{\star})=a^2(\sin^2\theta^{\star}-\sin^2\theta) \qquad \mbox{and} \qquad Q^2(r,\theta^{\star})=(r^2+a^2)^2-a^2\sin^2\theta^{\star}\Delta,
\eea
and
\bea
\lb{positivity_of_PQ}
P\geq 0,\qquad Q>0.
\eea
For the range $\theta \in (\frac{\pi}{2}, \pi)$, the coordinate $\theta^{\star}$ is defined by 
\bea
\lb{oth_def}
\theta^{\star}(r, \theta)=\pi-\theta^{\star}(r, \pi-\theta).
\eea
And further, we fix
\bea
\lb{moredef}
\theta^{\star}(r, 0)=0, \qquad \theta^{\star}(r, \frac{\pi}{2})=\frac{\pi}{2},  \qquad \theta^{\star}(r, {\pi})={\pi},
\eea
for $r\in[r_-,r_+]$.
It was shown in Appendix A of \cite{m_luk} that equations \eqref{F_imp_def}, \eqref{oth_def} and \eqref{moredef} define a unique $\theta^{\star}$.
Further, we define $\rho$ by
\bea
\rho(r,\theta,\theta^{\star})&=&\int\limits_{\theta^{\star}}^{\theta}{P(\theta', \theta^{\star})}{\md \theta'}+\int\limits_{r_+}^r\frac{Q(r', \theta^{\star})}{\Delta}{\md r'}
+\rho({r_+-\epsilon},0,0)\nonumber\\
&=&\int\limits_{\theta^{\star}}^{\theta}{P(\theta', \theta^{\star})}{\md \theta'}-\int\limits_r^{r_+}\frac{Q(r', \theta^{\star})}{\Delta}{\md r'}+\rho({r_+},0,0)\nonumber\\
&=&\int\limits_{\theta^{\star}}^{\theta}{P(\theta', \theta^{\star})}{\md \theta'}+\int\limits_r^{r_+}\frac{{r'}^2+a^2-Q(r', \theta^{\star})}{\Delta}{\md r'}+\left[\int\limits_{r_+}^r\frac{{r'}^2+a^2}{\Delta}{\md r'}+\rho({r_+},0,0)\right]\nonumber\\
&=&\int\limits_{\theta^{\star}}^{\theta}{P(\theta', \theta^{\star})}{\md \theta'}+\int\limits_r^{r_+}\frac{{r'}^2+a^2-Q(r', \theta^{\star})}{\Delta}{\md r'}+\int\limits_c^{r}\frac{{r'}^2+a^2}{\Delta}{\md r'}.
\eea
The arrangement in the radial dependence was made to ensure convergence of the definite integral despite $\Delta$ approaching zero at the horizons. 
Because of \eqref{thstar} we can now define
\bea
r^{\star}(r, \theta)=\rho(\theta^{\star}(r, \theta); r,\theta),
\eea
where $r^{\star}\in(-\infty,\infty)$ as r ranges between $r_+$ and $r_-$ and for the partial derivatives we obtain
\bea
\lb{partials}
\frac{\partial r^{\star}}{\partial r}=\frac{Q}{\Delta}, \qquad \frac{\partial r^{\star}}{\partial \theta}=P.
\eea
It is easy to see from the eikonal equation, that $\partial_{r^{\star}}$ is a future directed timelike vector field in the interior and $r^{\star}=const$ is a spacelike hypersurface.
As already pointed out in \cite{pretorius}, with this construction of $r^{\star}$ the variables 
\bea
v=\frac{t+{r}^{\star}(r, \theta)}{2}, \qquad u=\frac{r^{\star}(r,\theta)-t}{2} 
\eea
are lightlike and satisfy
\bea
\lb{eikonal}
g^{\alpha \beta}(\partial_{\alpha} v)(\partial_{\beta} v)=\frac{1}{4\Sigma}\left[\Delta(\partial_r {r}^{\star})^2+(\partial_{\theta} {r}^{\star})^2-\frac{(r^2+a^2)^2}{\Delta}+a^2\sin^2\theta \right]=0,
\eea
together with the analog for the $u$ coordinate.

Now we are ready to summarize all differentials, which we are going to use in Section \ref{doublenull} to derive the Kerr metric in double null coordinates. 
From a straight forward calculation we get
\bea
\lb{differentials}
\md r&=&\frac{Q\Delta}{\Sigma R^2}(\md v+\md u)+\frac{G QP^2\Delta}{\Sigma R^2}\md \theta^{\star},\\
\lb{differentials2}
\md \theta &=& \frac{P\Delta}{\Sigma R^2}(\md v+\md u)-\frac{G PQ^2}{\Sigma R^2}\md \theta^{\star},
\eea
with
\bea
\lb{mudef}
G(r^{\star},\theta^{\star})\equiv\partial_{\theta^{\star}} F,
\eea 
where $r$ and $\theta$ are fixed for the partial derivative $\partial_{\theta^{\star}}$. 
The exact expression was given in Proposition A.2 of Appendix A of \cite{m_luk}.

\subsubsection{Eddington--Finkelstein normalized double null coordinates}
\lb{doublenull}
In the following we will carry out the coordinate transformation \mbox{$(r,t, \phi, \theta)\rightarrow (u,v, \tilde{\phi}, \theta^{\star})$}, since double null coordinates will turn out more convenient for our interior analysis. 
Using \eqref{differentials}, \eqref{differentials2} and \mbox{$\md t=\md v-\md u$} in \eqref{bl_metric} we obtain
\bea
g_{\mbox{\tiny{Kerr}}}=4\frac{\Delta}{R^2}\md u \md v+\frac{L^2}{R^2}\md {\theta^{\star}}^2+R^2\sin^2 \theta(\md \phi-\omega_B(\md v-\md u))^2,
\eea
where we have defined 
\bea
\lb{Ldef}
L(u,v,\theta^{\star})\equiv -G PQ,
\eea
\bea
\lb{omega_B}
\omega_B(u,v,\theta^{\star})=\frac{2Mar}{\Sigma R^2}.
\eea
Note that the quantities $Q$ and $-GP$, given in \eqref{PandQ} and \eqref{mudef} are positive and bounded from above and below 
in the interior region. In particular, 
\bea
\lb{P_estimate}
P &\lesssim& |a|{\sin(2\theta^{\star})},
\eea
since $ \sqrt{\sin^2\theta^{\star}-\sin^2\theta} \lesssim \sin(2\theta^{\star})$, as shown in Appendix A of \cite{m_luk}.
Also notice that
\bea
\lb{R_Sigma}
\Sigma R^2=P^2\Delta +Q^2= (r^2+a^2)^2-a^2\Delta\sin^2\theta,
\eea
so the quantity is bounded in the interior region.
Further, refer to Proposition A.3 of \cite{m_luk} where the following proposition was proven.
\begin{prop} 
	\lb{Fprop}
	(M.~ Dafermos and J.~ Luk) Suppose $F(\theta^{\star}; r, \theta)=0$. Then
\bea
-{	c_{M,a}}\leq \left.
P(\theta^{\star};\theta) \frac{\partial}{ \partial \theta^{\star}}\right|_{r,\theta \mbox{\tiny{fixed}}} F(\theta^{\star};r,\theta) 
\leq -\frac{	\sin(2\theta)}{\sin(2\theta^{\star})}< -C_{M,a},
\eea
for some constants $c_{M,a}, C_{M,a}>0$ depending on $M$ and $a$. 
\end{prop}
The proposition together with \eqref{Ldef} and the boundedness above and below of $Q$ then lead to the boundedness 
\bea\lb{boundedL}
0< c < L(u,v,\theta^{\star})\leq C.
\eea
Let us now use the coordinate transformation\footnote{The constant multiplying $v$ in transformation \eqref{phi_trafo} is added in hindsight. It will render the parameter $b^{\tilde{\phi}}$, see \eqref{b_phi}, zero at $\cC\cH^+$.}
\bea
\lb{phi_trafo}
\phi&=&\tilde{\phi}+h(u,v,\theta^{\star})+\left.\frac{4Mar}{\Sigma R^2}\right|_{r=r_-}\cdot v,\\
\md \phi&=&\md \tilde{\phi}+\partial_{\theta^{\star}} h \md \theta^{\star}+ \partial_{u} h \md u +\left(\partial_{v} h+\left.\frac{4Mar}{\Sigma R^2}\right|_{r=r_-}\right) \md v,
\eea
with $h(u,v,\theta^{\star})$ a function for which we require that $\partial_u h=\partial_v h=-\omega_B$ and with initial data $h(r^{\star}=0; \theta^{\star})=0$. This leads to the following more convenient form of the metric in Eddington--Finkelstein normalized double null coordinates which were already related to the well known Boyer--Lindquist coordinates in \cite{pretorius} and \cite{m_rueck}: 
\bea
\lb{kerrmetric}
g=-2\Omega^2(u,v, \theta^{\star})(\md u \otimes\md v+\md v \otimes\md u)+ \gin_{\theta_C\theta_D}(\md \theta_C-b^{\theta_C}\md v)\otimes (\md \theta_D-b^{\theta_D}\md v),
\eea                                                                        
with $C,D=1,2$ and $\theta_1$, $\theta_2$ denoting $\theta^{\star}$ and ${\tilde{\phi}}$ respectively.
Further, 
\bea
&\gin_{\theta_C\theta_D}(u,v,\theta^{\star}, {\tilde{\phi}}) \quad &\mbox{is a Riemannian metric on $\bbS^2_{u,v}$},\nonumber\\
&b^{\theta_C}(u,v,\theta^{\star}, {\tilde{\phi}}) \quad &\mbox{is a vector field taking values in the tangent space of $\bbS^2_{u,v}$}.
\eea
For the inverse metric we then obtain
\bea
\lb{inverse}
g^{-1}&=&-\frac{1}{2\Omega^2(u,v,\theta^{\star})}(\partial_u \otimes \partial_v+ \partial_v \otimes \partial_u)-\frac{1}{2\Omega^2(u,v, \theta^{\star})}b^{\theta_C}(\partial_u \otimes \partial_{\theta_C}+ \partial_{\theta_C} \otimes \partial_u)\nonumber\\
&&+(\gin^{-1})^{\theta_C\theta_D}(\partial_{\theta_C}\otimes \partial_{\theta_D}).
\eea
Further, for the metric coefficients we obtain
\bea 
\lb{R_delta_sigma}
&&\Omega^2=-\frac{\Delta}{R^2},\\
\lb{b_phi}
&& b^{\theta^{\star}}=0, \quad b^{\tilde{\phi}}=2\omega_B(u,v,\theta^{\star})-\left.\frac{4Mar}{\Sigma R^2}\right|_{r=r_-}\stackrel{\eqref{omega_B}}{=}\frac{4Mar}{\Sigma R^2}-\left.\frac{4Mar}{\Sigma R^2}\right|_{r=r_-}, \\
\lb{gmpp}
&& \gin_{\tilde{\phi}\tilde{\phi}}=R^2\sin^2\theta,  \\
&& \gin_{\theta^{\star} \theta^{\star}}=\frac{L^2}{R^2}+\left(\frac{\partial h}{\partial \theta^{\star}}\right)^2R^2\sin^2\theta,  \\
\lb{gmlp}
&& \gin_{\theta^{\star} \tilde{\phi}}=\left(\frac{\partial h}{\partial \theta^{\star}}\right)R^2\sin^2\theta, \\
\lb{gpp}
&& (\gin^{\tilde{\phi}\tilde{\phi}})^{-1}=\frac1{R^2\sin^2\theta}+\left(\frac{\partial h}{\partial \theta^{\star}}\right)^2\frac{R^2}{L^2},  \\
&& (\gin^{\theta^{\star} \theta^{\star}})^{-1}=\frac{R^2}{L^2}, \\
\lb{glp}
&& (\gin^{\theta^{\star} \tilde{\phi}})^{-1}=-\left(\frac{\partial h}{\partial \theta^{\star}}\right)\frac{R^2}{L^2}.
\eea 
Recall again that the coordinate $\tilde{\phi}$ was chosen such that $b^{\tilde{\phi}}$ vanishes at $\cC\cH^+$.
Moreover, by \eqref{differentials}, from which we can read off $\partial_{\theta^{\star}} r$, together with \eqref{R_delta_sigma}, we can see that $\partial_{\theta^{\star}} b^{\tilde{\phi}}$ is proportional to $\Omega^2$ and thus also vanishes at $\cC\cH^+$, compare with Section \ref{bounded}. This is shown in more detail in Appendix A of \cite{m_luk}, where also smoothness at the poles for the angular variables $(\theta^{\star}, \tilde{\phi})$ is shown.

Further, note that we have the following dependencies
\bea
\lb{dependence}
&&r=r(u, v, \theta^{\star}), \qquad \theta=\theta(u, v, \theta^{\star}), \\
&&L=L(u, v, \theta^{\star}),  \qquad h=h(u, v, \theta^{\star}), \qquad R=R(u, v, \theta^{\star}), \mbox{ cf.~ \eqref{sigma_usw}}.
\eea
These dependencies are crucial since together with \eqref{delta_r} and \eqref{delta_theta} we will see that derivatives of these quantities with respect to $u$ and $v$ will also decay like $\Omega^2$. We will explain this more in Section \ref{bounded}.

\subsubsection{The volume elements}
\lb{volume_element}
From the expression for the metric we derive the volume elements for the entire spacetime
\bea
\dV&=&2\Omega^2 L \sin\theta\md \theta^{\star} \md \tilde{\phi} \md u \md v ,
\eea
as well as for the 3-dimensional hypersurfaces along constant $r^{\star}$ values:
\bea
\dV_{{r^{\star}=const}}&=&\sqrt{2\Omega^2} L \sin\theta\md \theta^{\star} \md \tilde{\phi} \md r^{\star},
\eea
were we have used
\bea
\sqrt{-g}=2\Omega^2\sqrt{g_{\theta^{\star} \theta^{\star}}g_{\tilde{\phi}\tilde{\phi}}-g_{\theta^{\star} \tilde{\phi}}^2}=2\Omega^2 L\sin \theta.
\eea
Note that the induced volume element of our spheres ${\mathbb S^2_{u,v}}$,
which are not round spheres, but are defined as the intersections of the level sets of $u$ and $v$, is given by
\bea
\lb{suv_vol}
\md \sigma_{\mathbb S^2_{u,v}}&=&\sqrt{-\gin}\md \theta^{\star} \md \tilde{\phi}=L \sin\theta\md \theta^{\star} \md \tilde{\phi}=L\md \sigma_{\mathbb S^2},
\eea
see \eqref{sigmadef}.
Therefore, we see that $L(u,v, \theta^{\star})$ plays the role of a radial coordinate which also varies in $\theta^{\star}$.
Moreover, the normal vectors on constant $u, v$ hypersurfaces and their related volume elements appearing in the statements of our theorems are given by:
\bea
\lb{nv}
n^{\mu}_{{v=const}}&=\partial_u, \qquad\qquad\qquad &\dV_{{v=const}}=2\Omega^2 L \sin\theta\md \theta^{\star} \md \tilde{\phi} \md u ,\\
\lb{nu}
n^{\mu}_{{u=const}}&=\partial_v +b^{\tilde{\phi}}\partial_{\tilde{\phi}}, \qquad &\dV_{{u=const}}=2\Omega^2 L  \sin\theta\md \theta^{\star} \md \tilde{\phi} \md v.
\eea
\subsubsection{Boundedness of all required coefficients}
\lb{bounded}
In order to carry out the analysis in the interior we will require
the following relations of the parameters introduced in the previous sections.
\bea
\lb{r_rstar}
\frac{\partial r}{\partial r^{\star}}&=&\frac{\Delta Q}{\Sigma R^2}, \qquad \frac{\partial r}{\partial {\theta^{\star}}}=\frac{\Delta P}{\Sigma R^2}\\
\lb{theta_rstar}
\frac{\partial \theta}{\partial r^{\star}}&=&\frac{\Delta P}{\Sigma R^2},\qquad \frac{\partial \theta}{\partial {\theta^{\star}}}=-\frac{GP Q^2}{\Sigma R^2}.
\eea
These relations are valid since $\frac{\partial \theta^{\star}}{\partial r}$ and $\frac{\partial \theta^{\star}}{\partial \theta}$ are finite, which was proven in Appendix A.3 of \cite{m_luk}. Further, recall \eqref{P_estimate} and the statement above of it together with \eqref{R_Sigma} which imply
\bea
\lb{partial_eq1}
&&\left|\frac{\partial r}{\partial r^{\star}}\right|\lesssim \left|\Delta\right|, \qquad \left|\frac{\partial r}{\partial {\theta^{\star}}}\right|\lesssim \left|\Delta\right| \sin (2 \theta^{\star}), \\
\lb{partial_eq2}
&& \left|\frac{\partial \theta}{\partial r^{\star}}\right|\lesssim \left|\Delta\right|\sin (2 \theta^{\star}), \qquad \left|\frac{\partial \theta}{\partial {\theta^{\star}}}\right|\lesssim 1,
\eea
as also stated in \cite{m_luk}.
Note that the sign of $Q$ was chosen to be positive so the sign of $\frac{\partial r}{\partial r^{\star}}$ is negative. Moreover, for the partial derivatives of $(r, \theta)$ with respect to $(u,v)$ we have
\bea
\lb{delta_r}
\partial_{\zeta} r=\frac{\partial r^{\star}}{\partial \zeta}\frac{\partial r}{\partial r^{\star}}\stackrel{\eqref{r_rstar}}{=}\frac{\Delta Q}{\Sigma R^2},\\
\lb{delta_theta}
\partial_{\zeta} \theta=\frac{\partial r^{\star}}{\partial \zeta}\frac{\partial \theta}{\partial r^{\star}}\stackrel{\eqref{theta_rstar}}{=}\frac{\Delta P}{\Sigma R^2},
\eea
with $\partial_{\zeta} r$ and $\partial_{\zeta} \theta$ negative, and $\zeta=u,v$.
Note that the entire analysis will be discussed in $r^{\star}$ coordinates not in $r$ coordinates, where in the region $r_-<r<r_+$ all $r^{\star}=const$ hypersurfaces are spacelike hypersurfaces and $r^{\star}\in(-\infty, \infty)$.

Further, to close the proof of our main theorem we will need to show boundedness of our metric coefficients and their derivatives up to second order.
For this we repeat Proposition A.6 and Remark A.11 of Appendix A of \cite{m_luk}.
\begin{prop}
	\lb{higher_o_boundedness}
(M. Dafermos and J. Luk) For every fixed $k\leq 2$, $r$ and $\theta$ are $C^k$ functions of $r^{\star}$ and $\theta^{\star}$. Moreover, the following bounds hold with implicit constants depending on $k$ (in addition to $M$ and $a$):
\bea
&&\sum\limits_{1\leq k_1\leq k-1}\left|\left(\frac{1}{\sin(2\theta^{\star})}\frac{\partial}{\partial \theta^{\star}}\right)^{k_1}\left(\frac{\partial \theta}{\partial \theta^{\star}}\right)\right|\lesssim 1,\\
&&\sum\limits_{1\leq k_1\leq k}\left|\left(\frac{\partial \theta}{\partial r^{\star}}\right)^{k_1}\theta\right|\lesssim |\Delta| \sin(2\theta^{\star}),\\
&&\sum\limits_{1\leq k_1+k_2\leq k-1; k_2 \geq 1}\left|\left(\frac{1}{\sin(2\theta^{\star})}\frac{\partial}{\partial \theta^{\star}}\right)^{k_1}\left(\frac{\partial }{\partial r^{\star}}\right)^{k_2}\left(\frac{\partial }{\partial \theta^{\star}}\right)\right|\lesssim |\Delta|,\\
&&\sum\limits_{1\leq k_1+k_2\leq k}\left|\left(\frac{\partial }{\partial r^{\star}}\right)^{k_1}\left(\frac{1}{\sin(2\theta^{\star})}\frac{\partial}{\partial \theta^{\star}}\right)^{k_2}r\right|\lesssim |\Delta|.
\eea
\end{prop}
Recall the functions \mbox{$b^{\tilde{\phi}}(u,v, \theta{\star})=2\omega_B-\left.\frac{4Mar}{\Sigma R^2}\right|_{r=r_-}$} of \eqref{b_phi} and \mbox{$\frac{\md h(u,v, \theta{\star})}{\md r^{\star}}=-\omega_B$} appering in \eqref{phi_trafo}, with \eqref{omega_B}. 
With the help of \eqref{partial_eq1} we can now state
\bea
\lb{b_bound}
\left|\frac{\partial b^{\tilde{\phi}}(u,v, \theta{\star})}{\partial r^{\star}} \right| \lesssim |\Delta|.
\eea
We further make use of the following.
\begin{rmk} (M. Dafermos and J. Luk) Notice that using 
\bea
\frac{\partial}{\partial \theta^{\star}}\left(\frac{4Mar}{\Sigma R^2}\right)=\frac{\partial r}{\partial \theta^{\star}}\frac{\partial}{\partial r}\left(\frac{4Mar}{\Sigma R^2}\right)+\frac{\partial \theta}{\partial \theta^{\star}}\frac{\partial}{\partial \theta}\left(\frac{4Mar}{\Sigma R^2}\right)
\eea
together with
\bea
\left|\frac{\partial}{\partial \theta}\left(\frac{4Mar}{\Sigma R^2}\right)\right|=\left|\frac{4Ma^3r\sin(2\theta)\Delta}{((r^2+a^2)^2-a^2\sin^2\theta \Delta)^2}\right|\lesssim |\Delta||\sin(2\theta^{\star})|
\eea
and estimates \eqref{partial_eq1}-\eqref{partial_eq2}, we have
\bea
\left|\frac{\partial}{\partial \theta^{\star}}\left(\frac{4Mar}{\Sigma R^2}\right)\right|\lesssim |\Delta||\sin(2\theta^{\star})|.
\eea
As a consequence, the expression $\frac{\partial h}{\partial \theta^{\star}}$ that appears in the metric component, compare \eqref{b_phi}, \eqref{gmpp} and \eqref{gmlp}, satisfies the following bound
\bea
\left|\frac{\partial h}{\partial \theta^{\star}}\right| \lesssim |\sin(2\theta^{\star})|.
\eea
\end{rmk}
We moreover make use of Proposition A.12 of Appendix A of \cite{m_luk}, which we are not repeating here.
Further, we obtain
\bea
\lb{L_bounded}
\left| 2\frac{\partial_{\delta} L\sin\theta}{L\sin\theta}\right| =\left| \frac{\partial_{\delta} (L^2\sin^2\theta)}{L^2\sin^2\theta}\right| = \left| 2\frac{\partial_{\delta} L}{L}+2\frac{\cos \theta}{\sin\theta}\partial_{\delta} \theta\right| \lesssim  |\Delta|,
\eea
with $\delta= u,v$ and where we used \eqref{partial_eq1}, \eqref{partial_eq2} and Proposition A.2 of \cite{m_luk}. This will turn out useful since we have to control this term in Section \ref{blue_future}. Moreover, we have
\bea
\lb{det_regel}
(\gin^{-1})^{\theta_C\theta_D}\partial_{\eta} \gin_{\theta_C\theta_D}=\frac{\partial_{\eta}\mbox{det}\gin}{\mbox{det}\gin}= \frac{\partial_{\eta}( L
^2\sin^2\theta)}{L^2 \sin^2 \theta},
\eea
is a term appearing in our bulk term, with $\eta=u,v, \theta^{\star}, \tilde{\phi}$, see \eqref{Kplug_edd-f} and also in the wave equation, see \eqref{wave_eq_eddf}. Note here, that $\frac{\partial_{\theta^{\star}} |L\sin\theta|}{|L\sin\theta|}$ is not bounded, but the product of this term and the volume element is.

\subsubsection{Angular momentum operators}
\lb{angular}
In this section we define the angular momentum operators by the vector fields
\bea
\lb{angular_comm}
Y_i=Y^{\theta_C}_i\partial_{\theta_C}= Y_i^{\theta^{\star}}\partial_{\theta^{\star}}+Y_i^{\tilde{\phi}}\partial_{\tilde{\phi}}, 
\eea
where $Y_i^{\theta_C}=Y_i^{\theta_C}(\theta^{\star}, \tilde{\phi})$, $C=1,2$ and $\theta_1=\theta^{\star}$, $\theta_2=\tilde{\phi}$ and $i=1,2,3$.
Further, we fix $Y_3^{\theta^{\star}}=0$ and $Y_3^{\tilde{\phi}}=1$, so that
\bea
\lb{angular_comm3}
Y_3=\partial_{\tilde{\phi}}, 
\eea
which is Killing for Kerr spacetimes. 
The $Y_i$ vector fields represent the standard generators (with respect to the $(\theta^{\star}, {\tilde{\phi}})$ coordinates) of the Lie algebra $so(3)$ acting on the spheres $\bbS^2_{u,v}$.  Since only $Y_3$ is Killing in Kerr spacetime, only the abelian subalgebra \mbox{$u(1) = \left\langle \partial_{\tilde{\phi}}\right\rangle  \subset so(3)$} acts by isometries. 
As we can see from \eqref{kerrmetric}, the two null vectors are orthogonal to $Y_i$, namely $\langle\partial_u, \partial_{\theta_C}\rangle=0$ and \mbox{$\langle\partial_v +b^{\theta_D} \partial_{\theta_D}, \partial_{\theta_C}\rangle=0$}, with $D=1,2$.
More generally we will use $Y^k$, with the upper index \mbox{$k \in \left\{0,1,2\right\}$} to refer to the summation of the required number of commutations, expressed by
\bea
\lb{Yk}
J_{\mu}^X(Y^k \psi)n^{\mu}&\doteq& \sum\limits_{i_1=1}^3 \cdot \cdot\cdot \sum\limits_{i_k=1}^3 J_{\mu}^X( Y_{i_1} \cdot \cdot\cdot (Y_{i_k} \psi))n^{\mu},
\eea
where $i_j=1,2,3$.
To make it more explicit, we can write
\bea
\lb{Y0}
J_{\mu}^X(Y^0 \psi)n^{\mu}&\doteq& J_{\mu}^X(\psi)n^{\mu},\\
\lb{Y1}
J_{\mu}^X(Y^1 \psi)n^{\mu}&\doteq&  \sum\limits_{i_1=1}^3 J_{\mu}^X( Y_{i_1} \psi )n^{\mu}=J_{\mu}^X( Y_{1} \psi )n^{\mu}+J_{\mu}^X( Y_{2} \psi )n^{\mu}+J_{\mu}^X( Y_{3} \psi )n^{\mu},\\
\lb{Y2}
J_{\mu}^X(Y^2 \psi)n^{\mu}&\doteq&\sum\limits_{i_1=1}^3 \sum\limits_{i_2=1}^3 J_{\mu}^X(Y_{i_1} (Y_{i_2}\psi))n^{\mu}\nonumber \\
&=& \quad J_{\mu}^X( Y_{1}(Y_{1} \psi) )n^{\mu}+J_{\mu}^X(  Y_{1}(Y_{2} \psi ))n^{\mu}+J_{\mu}^X(  Y_{1}(Y_{3} \psi ))n^{\mu}\nonumber \\
&& +J_{\mu}^X( Y_{2}(Y_{1} \psi ))n^{\mu}+J_{\mu}^X(  Y_{2}(Y_{2} \psi ))n^{\mu}+J_{\mu}^X(  Y_{2}(Y_{3} \psi ))n^{\mu}\nonumber \\
&& +J_{\mu}^X( Y_{3}(Y_{1} \psi ))n^{\mu}+J_{\mu}^X(  Y_{3}(Y_{2} \psi) )n^{\mu}+J_{\mu}^X(  Y_{3}( Y_{3} \psi ))n^{\mu}.
\eea 
Similar to \eqref{Yk}, one can define summed expressions for $K^X(Y^k \psi)$ and $\cE^X(Y^k \psi)$. In our proof we will also estimate absolute values of these quantities, which we express by
\bea
\lb{YkK}
\left|K^X(Y^k \psi)\right|&\doteq& \sum\limits_{i_1=1}^3 \cdot \cdot\cdot \sum\limits_{i_k=1}^3\left|K^X( Y_{i_1} \cdot \cdot\cdot (Y_{i_k} \psi))\right|,
\eea
and \bea
\lb{YkE}
\left|\cE^X(Y^k \psi)\right|&\doteq& \sum\limits_{i_1=1}^3 \cdot \cdot\cdot \sum\limits_{i_k=1}^3\left|\cE^X( Y_{i_1} \cdot \cdot\cdot (Y_{i_k} \psi))\right|.
\eea

\subsubsection{The redshift, noshift and blueshift region}
\lb{bnrsection}
As we have already mentioned in the introduction, in the interior we can distinguish 
\bea
\lb{defregions_r}
\mbox{redshift}\quad  &&\mbox{$\cR=\left\{-\infty < r^{\star} \leq r^{\star}_{red}\right\}$},\\
\lb{defregions_n}
\mbox{noshift}\quad  &&\mbox{$\cN=\left\{r^{\star}_{red}\leq r^{\star}\leq r^{\star}_{blue}\right\}$},\\
\lb{defregions_b}
\mbox{and} \quad \mbox{blueshift}\quad  &&\mbox{$\cB=\left\{r^{\star}_{blue}\leq r^{\star}< \infty\right\}$}
\eea
subregions, as shown in Figure \ref{IIbnr}.
{\begin{figure}[ht]
\centering
\includegraphics[width=0.3\textwidth]{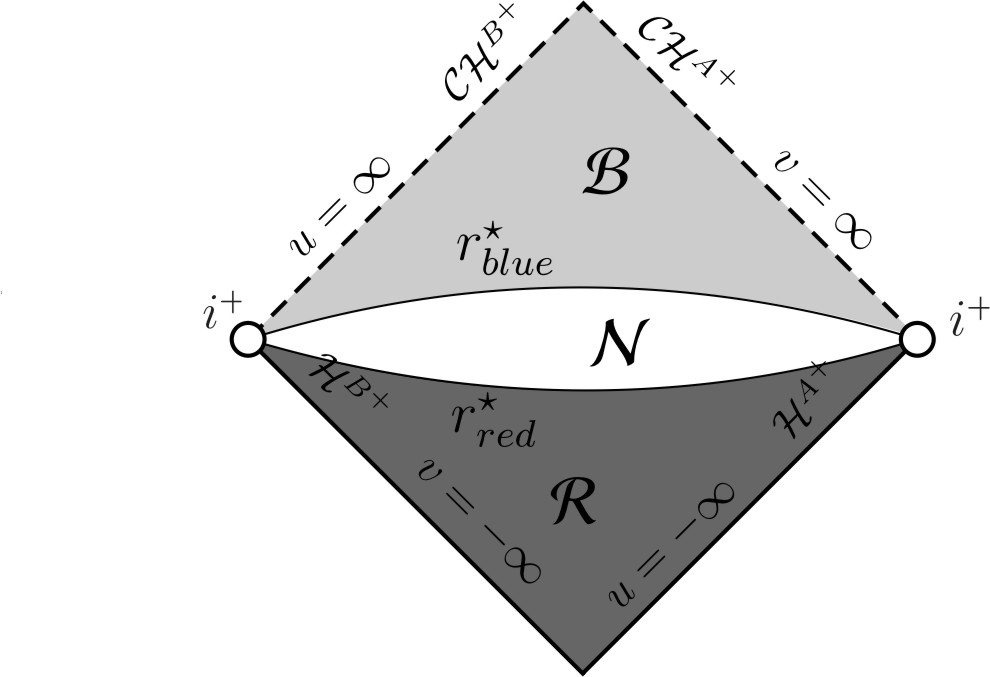}
\caption[]{Penrose diagram of the interior with distinction into redshift $\cR$, noshift $\cN$ and blueshift $\cB$ regions.}
\label{IIbnr}\end{figure}}

In the redshift region $\cR$, we make use of the fact that the surface gravity $\kappa_+$ of the event horizon is positive, see \eqref{kappapm}. 
Let $\varphi_t$ and $\varphi_{{\phi}}$ 
denote the 1-parameter groups of diffeomorphisms generated by the Killing
fields $\partial_t$ and $\partial_{\phi}$ (in Boyer-Lindquist coordinates), respectively. 
Then, the redshift region is characterized by the fact that there exists a $\varphi_t$ and $\varphi_{{\phi}}$ invariant vector field $N$, 
such that its associated current $J^N_{\mu}n^{\mu}_{v=const}$ on a $v=const$ hypersurface can be controlled by the related bulk term $K^N$, cf.~ Proposition \ref{mi}. 
This property of the bulk term $K^N$ holds sufficiently close to $\cH^+$. In particular, we shall define 
$r^{\star}_{red}$ negative and with big enough absolute value such that Proposition \ref{mi} is applicable. (Furthermore, note that the quantity $\frac{\partial_{\zeta}\Omega}{\Omega}$ is always positive in $\cR$.)

In $\cN$, defined in \eqref{defregions_n}, we exploit the fact that $J^{U}$, $K^{U}$ and $\cE^{U}$ are invariant under translations along $\partial_t$, where $t$ is the Boyer--Lindquist coordinate. For that reason we can uniformly control the bulk by the current along a constant $r^{\star}$ hypersurface. This will be explained further in Section \ref{no_region}.

The blueshift region $\cB$ is characterized by the fact that the bulk term $K^{S_0}$, associated to the vector field $S_0$, to be defined in \eqref{S0field}, is positive. We define 
$r^{\star}_{blue}$ big enough such that the quantities 
\bea
\lb{omegabruch}
\frac{\partial_{\zeta} \Omega}{\Omega}=\frac12 \left[\frac{\partial_{\zeta}\Delta}{\Delta}-\frac{\partial_{\zeta}(R^2)}{R^2}\right],
\eea
with $\zeta=u,v$, carry a negative sign\footnote{Readers familiar with \cite{anne} might notice that, for $|a|\ll M$, the quantity in \eqref{omegabruch} defining the sign approaches $M-\frac{a^2}{r}$ which is analogous to the charged case with $a$ in place of the charge $e$.}.
In particular, 
for $\tilde{\epsilon}$ sufficiently small the following lower bound holds in $\cB$
\bea
\lb{lowerboundu}
0<{\beta} \leq  -2\frac{\partial_u \Omega}{\Omega},\qquad
\lb{lowerboundv}
0<{\beta} \leq -2\frac{\partial_v \Omega}{\Omega},
\eea
with $\beta$ a positive constant.

\subsection{Notation}
\lb{nota}
It will be useful to determine in- and outgoing null hypersurfaces from intersection of out- and ingoing null hypersurfaces with
the spacelike hypersurfaces $r^{\star}=r^{\star}_{red}$, $r^{\star}=r^{\star}_{blue}$ and in addition the hypersurface $\gamma$ which will be defined in Section \ref{gamma_curve}.
In order to do that, recall the beginning of Section \ref{ambient} and end of Section \ref{doublenull}, were we have expressed the submanifold \mbox{$\cM|_{II}=\cQ|_{II}\times \bbS^2_{u,v}$} by \eqref{MII}. 
Further, we define \mbox{$\cM|_{II}=\pi^{-1}(\cQ|_{II}$}), where $\pi$ is the projection \mbox{$\pi:\cM|_{II}\rightarrow \cQ|_{II}$}.
Now we are able to define functions that are evaluated by the projection $\pi$ of two particular intersecting hypersurfaces.
For example given the hypersurface $r^{\star}=r^{\star}_{red}$ and the hypersurface $u=\tilde{u}$ we define the $v$ value at which these two hypersurfaces intersect by a function $v_{red}(\tilde{u})$ evaluated for $\tilde{u}$. Let us therefore introduce the following notation:
\bea
\lb{notation_neu}
v_{red}(\tilde{u}) \qquad &\mbox{is determined by}& \qquad r^{\star}(v_{red}(\tilde{u}), \tilde{u})=r^{\star}_{red},\nonumber \\
v_{\gamma}(\tilde{u}) \qquad &\mbox{is determined by}& \qquad (v_{\gamma}(\tilde{u}), \tilde{u}) \in \gamma, \nonumber\\
v_{blue}(\tilde{u}) \qquad &\mbox{is determined by}& \qquad  r^{\star}(v_{blue}(\tilde{u}), \tilde{u})=r_{blue},\nonumber \\
u_{red}(\tilde{v}) \qquad &\mbox{is determined by}& \qquad r^{\star}(u_{red}(\tilde{v}), \tilde{v})=r^{\star}_{red},\nonumber\\
u_{\gamma}(\tilde{v}) \qquad &\mbox{is determined by}& \qquad (u_{\gamma}(\tilde{v}), \tilde{v}) \in \gamma,\nonumber\\
u_{blue}(\tilde{v}) \qquad &\mbox{is determined by}& \qquad r^{\star}(u_{blue}(\tilde{v}), \tilde{v})=r^{\star}_{blue}.
\eea
For a better understanding the reader may refer to Figure \ref{integralbild3}.
{\begin{figure}[!ht]
\centering
\includegraphics[width=0.75\textwidth]{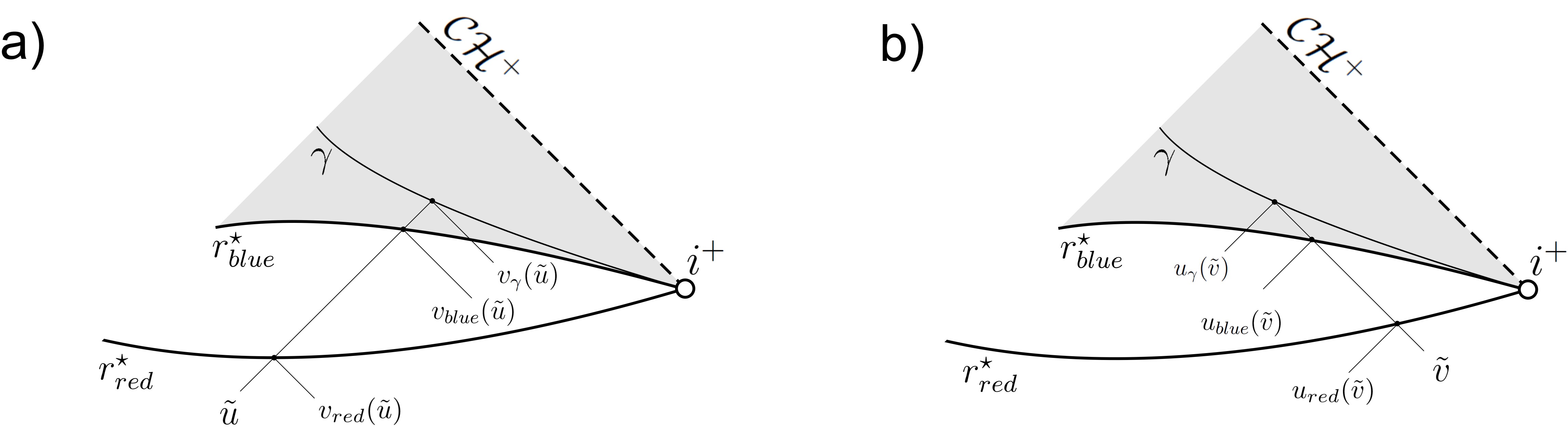}
\caption[]{Sketch of blueshift region $\cB$, with quantities depicted a) dependent on $\tilde{u}$ and b) dependent on $\tilde{v}$.}
\label{integralbild3}\end{figure}}

Note that the above functions are well defined since $r^{\star}=r^{\star}_{red}$, $r^{\star}=r^{\star}_{blue}$ and $\gamma$ are spacelike hypersurfaces along which we have \mbox{$u_{\gamma}(v)\rightarrow -\infty$} as \mbox{$v\rightarrow \infty$}. 

\section{The setup}
\lb{setup}

\subsection{Horizon estimates and Cauchy stability}
\lb{horizon_estimates}

Our starting point will be decay bounds for $\psi$, proven by Dafermos, Rodnianski and Shlapentokh-Rothman \cite{m_kerr}. In particular they show decay for $\psi$ and its derivatives in the black hole {\it exterior} up to and including the event horizon for the full range of all subextremal Kerr solutions. More specifically, the following theorem was stated in estimate $(21)$ of Theorem 3.1 and $(26)$ of Theorem 3.2 of \cite{m_kerr}, where the coordinate $\tau$ is comparable to our coordinate $v$. Also see the improved decay rate in \cite{george} derived by Moschidis.
\begin{thm}
\lb{anfang}
Let $\psi$ be a solution of the wave equation \eqref{wave_psi} on a subextremal Kerr background $(\cM,g)$, with mass $M$ and angular momentum per unit mass $a$ and $M>|a|\neq 0$, arising from smooth compactly supported initial data on an arbitrary Cauchy hypersurface $\Sigma$, cf.~ Figure \ref{kerrfigure}. Then, there exists $\delta>0$ such that
\bea
\lb{thpur1}
\int\limits^{v+1}_v\int\limits_{\bbS^2_{u,v}}\left[(\partial_v Y^k\psi+ b^{\tilde{\phi}}\partial_{\tilde{\phi}} Y^k\psi)^2(-\infty, v) +\Omega^2|\nabb Y^k\psi|^2(-\infty, v)\right] L\md \sigma_{\mathbb S^2}\md v\leq C v^{-2-2\delta},
\eea
on ${\cH_A}^+$, for all $v$, with $Y^k$ as in \eqref{angular_comm}, all $k \in \left\{0,1,2\right\}$ and some positive constants $C$ depending on the initial data. 
\end{thm}
The expression ${\bbS^2_{u,v}}$ denotes the spheres obtained in Eddington-Finkelstein-like coordinates which were derived in Section \ref{doublenull}.
The assumption of smoothness and compact support in Theorem \ref{anfang} can be weakened. See also related results \cite{anderson, tataru, tataru2, luk2, m_I_kerr, m_lec, m_bound} for the $|a|\ll M$ case and \cite{bernard, yakov} for mode stability.

Moreover, trivially from Cauchy stability, boundedness of the energy along a null segment transverse to $\cH^+$ can be derived. See Section \ref{firstlook} and recall that we have chosen the $v=1$ hypersurface to be the past boundary of the characteristic rectangle $\Xi$. More generally we can state the following proposition.
\begin{prop}
\lb{initialdataprop}
Let \mbox{$u_{\diamond}, v_{\diamond} \in (-\infty, \infty)$}.
Under the assumption of Theorem \ref{anfang}, the energy at advanced Eddington--Finkelstein-like coordinate \mbox{$\left\{v=v_{\diamond}\right\}\cap\left\{{-\infty}\leq u \leq {u_{\diamond}}\right\}$}
is bounded from the initial data 
\bea
\lb{proppur1}
\int\limits_{-\infty}^{u_{\diamond}}\int\limits_{\bbS^2_{u,v}}\left[ (\partial_u Y^k\psi)^2(u,v_{\diamond})+{\Omega^2}|\nabb  Y^k\psi|^2(u,v_{\diamond})\right] L\md \sigma_{\mathbb S^2}\md u&\leq& D(u_{\diamond},v_{\diamond}),
\eea
with $Y^k$ as in \eqref{angular_comm} and for all $k \in \left\{0,1,2\right\}$. Further, 
\bea
\lb{aufeins1}
\sup_{-\infty\leq u \leq {u_{\diamond}}}\int\limits_{\bbS^2_{u,v}} (Y^k\psi)^2(u,v_{\diamond})\md \sigma_{\mathbb S^2}&\leq&  D(u_{\diamond},v_{\diamond}),
\eea
with $D(u_{\diamond},v_{\diamond})$ positive constants depending on the initial data on $\Sigma$.
\end{prop}
\begin{proof}
This follows immediately from local energy estimates in a compact spacetime region. Note the $\Omega^{2}$ weights which arise since $u$ is not regular at $\cH^+_A$.
\end{proof}

\subsection{Statement of the theorem and outline of the proof in the neighborhood of $i^+$}
\lb{statement}
Before we can prove Theorem \ref{main}, we will first show the following. 
\begin{thm}
\lb{dashier}
On subextremal Kerr spacetime with $M>|a|\neq 0$, let $\psi$ be as in Theorem \ref{anfang}, then
\ben
|\psi|\leq C
\een
in the black hole interior up to
$\cC\cH^+$ in a ``small neighbourhood'' of timelike infinity $i^+$,
that is in \mbox{$(-\infty, u_{\schere}] \times [1, \infty)$} for some $u_{\schere}>-\infty$.
\end{thm}
{\em Remark.} We will see that $C$ depends only on the initial data. 

First we consider the {\it right} side of the spacetime interior. In particular we consider a characteristic rectangle $\Xi$ which extends from $\cH_A^+$, as shown in Figure \ref{Kerr_character}.
{\begin{figure}[ht]
\centering
\includegraphics[width=0.8\textwidth]{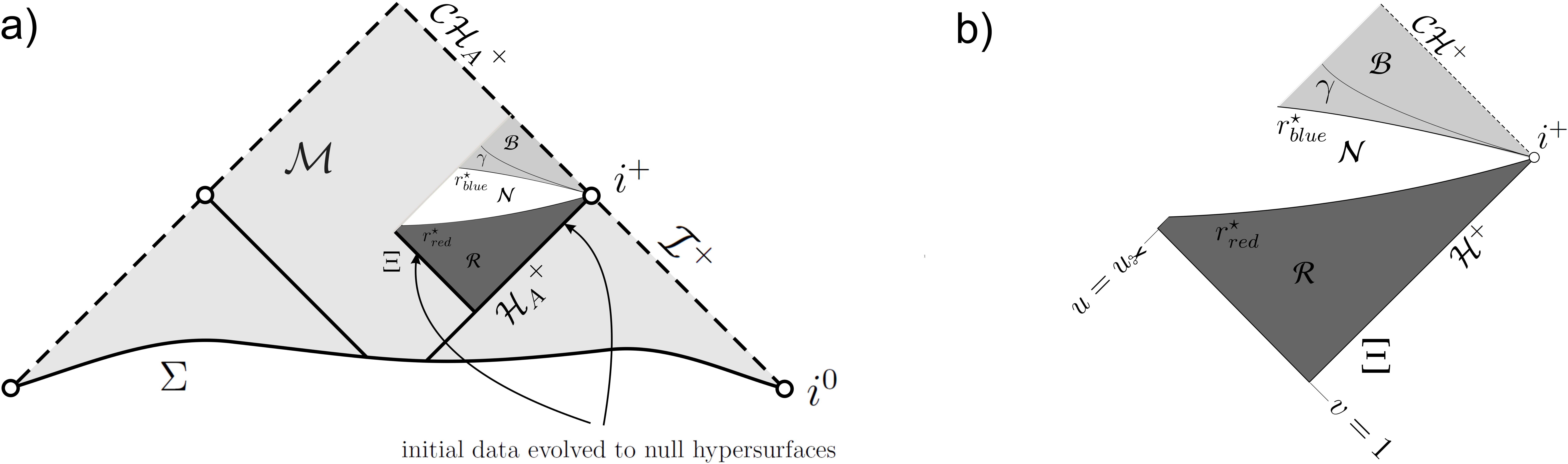}
\caption[]{a) Depiction of characteristic rectangle $\Xi$ within $\cQ|_{II}$ of Kerr spacetime, b) region $\Xi$ zoomed in.}
\label{Kerr_character}\end{figure}} 
We pick the characteristic rectangle to be defined by \mbox{$\Xi=\left\{(-\infty\leq u\leq u_{\schere}), (1\leq v <\infty)\right\}$}, where $u_{\schere}$ is sufficiently close to $-\infty$ for reasons that will become clear later on, cf.~ Lemma \ref{K_spher} and Proposition \ref{kastle}.
Bounds on the data along the event horizon $\cH_A^+$ were stated in Theorem \ref{anfang}, and bounds on the data of the transverse null segment were given in Proposition \ref{initialdataprop}. Both are consequences of propagation of smooth initial data from a Cauchy hypersurface as explained in Section \ref{setup}.
Theorem \ref{dashier} will be proven by first deriving energy estimates in the distinguished regions, denoted by redshift $\cR$, the noshift $\cN$ and the blueshift $\cB$ region, with the properties as explained in Section \ref{bnrsection}, 
cf.~ Figure \ref{Kerr_character} b). According to the properties of these regions we will have to choose different vector field mutlipliers in each of them. 

In the redshift region $\cR$ we will make use of the redshift vector field $N$, already introduced in \cite{m_lec}. We will elaborate more on this in Section \ref{red_region}.
Proposition \ref{mi} proves the positivity of the bulk $K^N$ which leads to boundedness of the energy flux generated by the current $J^N_{\mu}n^{\mu}_{v=const}$. Here we will have to construct $N$, as explained in Section \ref{red_region}, such that it can also compensate for the appearing error terms after commuting twice with the angular momentum operators.
Applying the divergence theorem, decay up to $r^{\star}=r^{\star}_{red}$ will be proven.

In the noshift region $\cN$, we can simply appeal to the fact that the future directed timelike vector field 
\ben
U=\partial_u+\partial_v +b^{\tilde{\phi}}\partial_{\tilde{\phi}},
\een
is invariant under the flow of the spacelike Killing vector field $\partial_t$.
It is for that reason that the bulk terms $K^{U}$ as well as the error terms $\cE^{U}$ can be uniformly controlled by the energy flux $J_{\mu}^{U}n^{\mu}_{{r^{\star}}=\bar{r}^{\star}}$ through the ${r^{\star}}=\bar{r}^{\star}$ hypersurface.
Decay up to $r^{\star}=r^{\star}_{blue}$ will be proven by making use of this, together with the uniform boundedness of the $v$ length of the region that we cut in $\cN$.

In order to gain control in the blueshift region $\cB$, we will partition it by the hypersurface $\gamma$ admitting logarithmic distance in $v$ from $r^{\star}=r^{\star}_{blue}$, cf.~ Section \ref{gamma_curve}. We will then separately consider the region to the past of $\gamma$, \mbox{$J^-(\gamma)\cap \cB$}, and the region to the future of $\gamma$, \mbox{$J^+(\gamma)\cap \cB$}.
The region to the future of $\gamma$ is characterized by good decay bounds on $\Omega^2$.
This property will be crucial to derive estimates for bulk and error terms in the future of $\gamma$.

In \mbox{$J^-(\gamma)\cap \cB$} we use a vector field 
\ben
{S_0}=f^q\partial_{r^{\star}}=f^q(\partial_u+\partial_v+ b^{\tilde{\phi}} \partial_{\tilde{\phi}}),
\een
where $q$ is sufficiently large, and $f$ is a function with certain properties defined in Section \ref{functionf}. Choosing the parameter $q$ big enough renders the terms associated to the spacetime integral positive, which is the ``good'' sign, when using the divergence theorem to derive upper bounds.

In order to complete the proof, we consider finally the region \mbox{$J^+(\gamma)\cap \cB$} and propagate the decay further from the hypersurface $\gamma$, up to the Cauchy horizon in a neighborhood of $i^+$. For this, we introduce a new timelike vector field ${S}$ defined by
\bea
\lb{N1}
{S}=|u|^p\partial_u+v^p\partial_v+ v^p b^{\tilde{\phi}} \partial_{\tilde{\phi}},
\eea
for an arbitrary $p$ such that
\bea
\lb{waspist}
1<p\leq 1+2\delta,
\eea
where $\delta$ is as in Theorem \ref{anfang}.
We use pointwise estimates on $\Omega^2$ in $J^+(\gamma)$ as a crucial step, see Section \ref{finiteness}. This is due to the proportionality of the occurring quantities to the function $\Omega^2$, see Section \ref{bounded}.

Putting everything together, in view of the geometry and the weights of $S$, we finally obtain for all $v_*\geq 1$
\bea
\lb{fluxscetch}
 \int\limits_1^{v_*} \int\limits_{\bbS^2_{u,v}} v^p(\partial_v Y^k\psi)^2 L \sin\theta\md \theta^{\star} \md \tilde{\phi}\md v\lesssim \mbox{Data},
\eea
for the weighted flux, with $Y^k$ as in Section \ref{angular}. 
Using the above, the uniform boundedness for $\psi$ stated in Theorem \ref{dashier} can then be derived from an argument that can be sketched as follows. 

Let us first see how we obtain an integrated bound on the spheres.
By the fundamental theorem of calculus and the Cauchy--Schwarz inequality we get
\bea
\lb{vertausch}
&&\int\limits_{\bbS^2_{u,v}} (Y^k\psi)^2(u, v_*, \theta^{\star}, \tilde{\phi}) L \sin\theta\md \theta^{\star} \md \tilde{\phi}\nonumber\\
&\lesssim& \int\limits_{\bbS^2_{u,v}}\left(\int\limits_1^{v_*} v^p(\partial_v Y^k\psi)^2\md v\right)\left(\int\limits_1^{v_*} v^{-p}\md v\right)L \sin\theta\md \theta^{\star} \md \tilde{\phi}+\mbox{data}\nonumber\\
&\lesssim& \left(\int\limits_1^{v_*}\int\limits_{\bbS^2_{u,v}}v^p(\partial_v Y^k\psi)^2L \sin\theta\md \theta^{\star} \md \tilde{\phi}\md v\right)\left(\int\limits_1^{v_*} v^{-p}\md v\right)+\mbox{data},
\eea
where the bound \eqref{boundedL} allowed us to pull the $(L\sin \theta)$-factor of the volume form inside the integral, so that the first factor of the first term is controlled by \eqref{fluxscetch}.
Therefore, we further get
\bea
\lb{fundcauchy11}
\int\limits_{\bbS^2_{u,v}} (Y^k\psi)^2 L\sin\theta\md \theta^{\star} \md \tilde{\phi}\stackrel{\eqref{fluxscetch}}{\lesssim}\mbox{Data}\int\limits_1^{v_*} v^{-p}\md v+\mbox{data}
\lesssim \mbox{Data}+\mbox{data},
\eea 
where we have used \mbox{$\int\limits^{\infty}_{1} v^{-p}\md v< \infty$} which followed from the first inequality of \eqref{waspist}.

Obtaining a pointwise statement from the above will be achieved by using \eqref{fundcauchy11} and applying Sobolev embedding on the spheres $\bbS^2_{u,v}$ (which are not round spheres), thus leads to the desired bounds, see Section \ref{uni_bounded}. This will close the proof of Theorem \ref{dashier}.
Having only one spacelike symmetry and having to control the appearing error terms is one of the main difficulties of this paper in comparison to the Reissner--Nordstr\"om analog \cite{anne}.

\section{Energy estimates in the interior}
\lb{energy_interior}

\subsection{Propagation through the redshift region $\cR$ to $r^{\star}=r^{\star}_{red}$}
\lb{red_region}
The following proposition was shown in \cite{m_lec}, see also \cite{volker} for a detailed proof.
\begin{prop}
\lb{mi}
(M. Dafermos and I. Rodnianski) For $r^{\star}_{red}$ sufficiently close to $-\infty$ there exists a $\varphi_{t}$ and $\varphi_{\phi}$-invariant smooth future directed timelike vector field $N$ on \mbox{$\left\{-\infty< r^{\star} \leq r^{\star}_{red} \right\}\cap \left\{v\geq 1\right\}$} and a positive constant $b_0$ such that 
\bea
\lb{energy_control}
b_0J_{\mu}^N(\psi) n^{\mu}_{{v}} \leq K^N(\psi) ,
\eea
for all solutions $\psi$ of \eqref{wave_psi}.
\end{prop}
\begin{proof}
As was already pointed out in \cite{m_kerr} and also \cite{luk_jan}, in Kerr spacetime, by $\varphi_{\tau}$ and $\varphi_{\phi}$ we denote diffeomorphisms generated by the Killing fields $\partial_t$ and $\partial_{\phi}$, respectively. These contain the diffeomorphisms generated by $T_{\cH^+}$, defined in \eqref{T_H2}. As mentioned in Section \ref{ambient}, $T_{\cH^+}$ is the vector field turning null on $\cH^+$ which is required for the construction of the proof. To see more details of the proof, and general structure of the vector field $N$, refer to Appendix \ref{redshift_app}.
\end{proof}
The decay bound along $r^{\star}=r^{\star}_{red}$ can now be stated in the following proposition.
\begin{prop}
\label{rtildedecay}
Let $\psi$ be as in Theorem \ref{anfang}, and $Y^k$ as in \eqref{angular_comm} with \eqref{Yk},  \eqref{YkK}, \eqref{YkE} and all $k \in \left\{0,1,2\right\}$.
Then, for all $\tilde{r^{\star}} \in (-\infty,r^{\star}_{red}]$, with $r^{\star}_{red}$ as in Proposition \ref{mi} and for all $v_*>1$,
\ben
\lb{decay_*}
\int\limits_ {\left\lbrace  v_* \leq v \leq v_*+1\right\rbrace } J_{\mu}^N(Y^k \psi) n^{\mu}_{r^{\star}=\tilde{r^{\star}}} \dV_{r^{\star}=\tilde{r^{\star}}} \leq C v_*^{-2-2\delta} , 
\een
with $C$ depending on $C_{0}$ of Theorem \ref{anfang} and $D_{0}(u_{\diamond}, 1)$ of Proposition \ref{initialdataprop}, where $u_{\diamond}$ is defined by $r^{\star}_{red}=r(u_{\diamond},1)$.
\end{prop}
{\em Remark 1.} The decay in Proposition \ref{rtildedecay} matches the decay on $\cH^+$ of Theorem \ref{anfang}.

{\em Remark 2.} 
$n^{\mu}_{r^{\star}=r^{\star}_{red}}$ denotes the normal to the $r^{\star}=r^{\star}_{red}$ hypersurface oriented according to the Lorentzian geometry convention. $\dV$ denotes the volume element over the entire spacetime region and $\dV_{r^{\star}=r^{\star}_{red}}$ denotes the volume element on the $r^{\star}=r^{\star}_{red}$ hypersurface. Similarly for all other subscripts.\footnote{Refer to the end of Section \ref{doublenull} for further discussion of the volume elements.}

\begin{proof}
Applying the divergence theorem, the decay rate from Theorem \ref{anfang} and Proposition \ref{initialdataprop} as well as Propositon \ref{mi} and a bootstrap argument, the statement of the Propositon \ref{rtildedecay} follows for $k=0$. For more details see Proposition 4.2 of \cite{anne}. To consider higher orders we need the following lemma.
\begin{lem}
\lb{mi_k}
For $r^{\star}_{red}$ sufficiently close to $-\infty$ there exists a $\varphi_{t}$ and $\varphi_{\phi}$-invariant smooth future directed timelike vector field $N$ on \mbox{$\left\{-\infty< r^{\star} \leq r^{\star}_{red} \right\}\cap \left\{v\geq 1\right\}$} and positive constants $b_k$, with all $k \in \left\{0,1,2\right\}$ and $Y^k$ as in \eqref{angular_comm} with \eqref{Yk},  \eqref{YkK}, \eqref{YkE}, such that 
\bea
\lb{energy_control2_1}
b_1J_{\mu}^N(Y\psi) n^{\mu}_{{v}}+b_0J_{\mu}^N(\psi) n^{\mu}_{{v}} \leq K^N(Y\psi)+ \cE^N(Y\psi)+K^N(\psi),
\eea
and
\bea
\lb{energy_control2_2}
b_2J_{\mu}^N(Y^2\psi) n^{\mu}_{{v}}+b_1J_{\mu}^N(Y\psi) n^{\mu}_{{v}}+b_0J_{\mu}^N(\psi) n^{\mu}_{{v}} \nonumber\\
\leq K^N(Y^2\psi)+ \cE^N(Y^2\psi)+ K^N(Y\psi)+ \cE^N(Y\psi)+K^N(\psi),
\eea
for all solutions $\psi$ of \eqref{wave_psi}.
\end{lem}
\begin{proof}
In Appendix \ref{redshift_app} we have discussed how to prove
Proposition \ref{mi}. Now recall expressions \eqref{errorexpression} and \eqref{comm} in order to investigate the control over the error terms.
From equations \eqref{pi-Y-termevv} to \eqref{pi-Y-termeAB} of Appendix \ref{error_app}.\ref{S_error}, we can see that like all first derivative terms, all terms multiplying the second derivatives will also be bounded. Note, that in order to control the $(\partial_u\partial_v \psi)$ term, we need to use \eqref{uv_term}. Therefore, with the choice $N^u, N^v$ positive and $\partial_u N^v$, $\partial_u N^u$ negative and with large enough absolute value, so that the contribution from the error terms is compensated, we can prove the above lemma for $k \in \left\{0,1,2\right\}$.  
\end{proof}
To prove Propositon \ref{rtildedecay} for $k=1,2$, we can now use the divergence theorem and the above lemma. 
\end{proof}
\begin{cor}
Let $\psi$ be as in Theorem \ref{anfang}, and $Y^k$ as in \eqref{angular_comm} with \eqref{Yk},  \eqref{YkK}, \eqref{YkE} and all $k \in \left\{0,1,2\right\}$, and for $r^{\star}_{red}$ as in Proposition \ref{mi}. Then, for all $v_*\geq 1$, $v_*+1\leq v_{red}(\tilde{u})$ and for all $\tilde{u}$ such that \mbox{$r^{\star}(\tilde{u},v_*+1) \in (-\infty, r^{\star}_{red}]$},  we have
\lb{cor5.2}
\bea
\int\limits_ {\left\lbrace  v_* \leq v \leq v_*+1\right\rbrace } J^N_\mu(Y^k\psi)n^\mu_{u=\tilde{u}}dVol_{u=\tilde{u}} \le C v_*^{-2-2\delta},
\eea
with $C$ depending on $C_{0}$ of Theorem \ref{anfang} and $D_{0}(u_{\diamond}, 1)$ of Proposition \ref{initialdataprop}, where $u_{\diamond}$ is defined by $r^{\star}_{red}=r^{\star}(u_{\diamond},1)$ and $v_{red}(\tilde{u})$ as in \eqref{notation_neu}.
\end{cor}
\begin{proof}
The conclusion of the statement follows by applying again the divergence theorem 
and using the results of the proof of Proposition \ref{rtildedecay}. 
\end{proof}

\subsection{Propagation through the noshift region $\cN$ to $r^{\star}=r^{\star}_{blue}$}

\lb{no_region}
Without any loss in the decay rate we can propagate it further inside through the noshift region $\cN$ up to the $r^{\star}=r^{\star}_{blue}$ hypersurface. 
In order to do that, we will use the future directed timelike vector field 
\bea
\lb{partial_r}
U=\partial_u+\partial_v +b^{\tilde{\phi}}\partial_{\tilde{\phi}},
\eea
and state the following lemma.
\begin{lem}
\lb{Kr} 
Let $\psi$ be an arbitrary function. Then, for $r^{\star}_{blue}$ sufficiently large, the bulk term of the future directed timelike vector field $U$ can be estimated by
\bea
\lb{controll_bulk_partial_r}
 |K^{U}(\psi)| &\leq& B_0 J_{\mu}^{U}(\psi) n^{\mu}_{r^{\star}={\bar{r}^{\star}}},
\eea
in $\cN$, where $B_0$ is independent of $v_*$ but dependent on the choice of $r^{\star}_{blue}$.
\end{lem}
\begin{proof}
Validity of the estimate \eqref{controll_bulk_partial_r} can be seen without computation from the fact that currents with timelike vector field multiplier, such as $J_{\mu}^{U}(\psi) n^{\mu}_{r^{\star}={\bar{r}^{\star}}}$, contain all derivatives. 
The uniformity of $B_0$ is given by the fact that $K^{U}$ and $J^{U}$ are invariant under translations along spacelike $\partial_t$.
Therefore, we can just look at the maximal deformation on a compact \mbox{$\left\{t=const \right\} \cap \left\{r^{\star}_{blue} \leq r^{\star} \leq r^{\star}_{red}\right\}$} hypersurface and get an estimate for the deformation everywhere.
\end{proof}

\begin{prop}
\label{r_{red}}
Let $\psi$ be as in Theorem \ref{anfang}, $r^{\star}_{blue}$ such that the quantities \eqref{omegabruch} are negative and $r^{\star}_{red}$ as in Proposition \ref{mi}, and $Y^k$ as in \eqref{angular_comm} with \eqref{Yk},  \eqref{YkK}, \eqref{YkE} and all $k \in \left\{0,1,2\right\}$. Then, for all $v_*>1$ and $\tilde{r^{\star}} \in [r^{\star}_{blue},r^{\star}_{red})$, we have 
\bea
\lb{noprop}
\int\limits_ {\left\lbrace  v_* \leq v \leq {v_*}+1\right\rbrace } J_{\mu}^{U}(Y^k\psi) n^{\mu}_{r^{\star}=\tilde{r}^{\star}} \dV_{r^{\star}=\tilde{r}^{\star}}
\leq
C{v_*}^{-2-2\delta}, 
\eea
with $C$ depending on the initial data or more precisely depending on $C_{0}$ of Theorem \ref{anfang} and $D_{0}(u_{\diamond}, 1)$ of Proposition \ref{initialdataprop}, where $u_{\diamond}$ is defined by $r^{\star}_{red}=r^{\star}(u_{\diamond},1)$.
\end{prop}
\begin{proof}
We will first prove the statement for $k=0$.
Given $v_*$, we define regions $\cR_{II}$ and $\tilde{\cR}_{II}$ as in Figure \ref{r_{red}_decay}, where we use \eqref{notation_neu} and
\bea
\lb{notation_v_*}
v(\tilde{r^{\star}}, v_*) \quad &\mbox{is determined by}& \quad r^{\star}(u_{blue}(v_*), v(\tilde{r^{\star}}, v_*))=\tilde{r^{\star}}.
\eea
Thus the depicted regions are given by \mbox{$\cR_{II}\cup \tilde{\cR}_{II}= \cD^+(\left\lbrace v_1 \leq v \leq v_*+1\right\rbrace \cap  \left\lbrace r=r_{red}\right\rbrace )\cap \cN $}, where region \mbox{$\cR_{II}$} is given by \mbox{$\cR_{II}= \cD^+(\left\lbrace v_1 \leq v \leq v_*\right\rbrace \cap \left\lbrace r^{\star}=r^{\star}_{red}\right\rbrace)  $}.
{\begin{figure}[ht]
\centering
\includegraphics[width=0.4\textwidth]{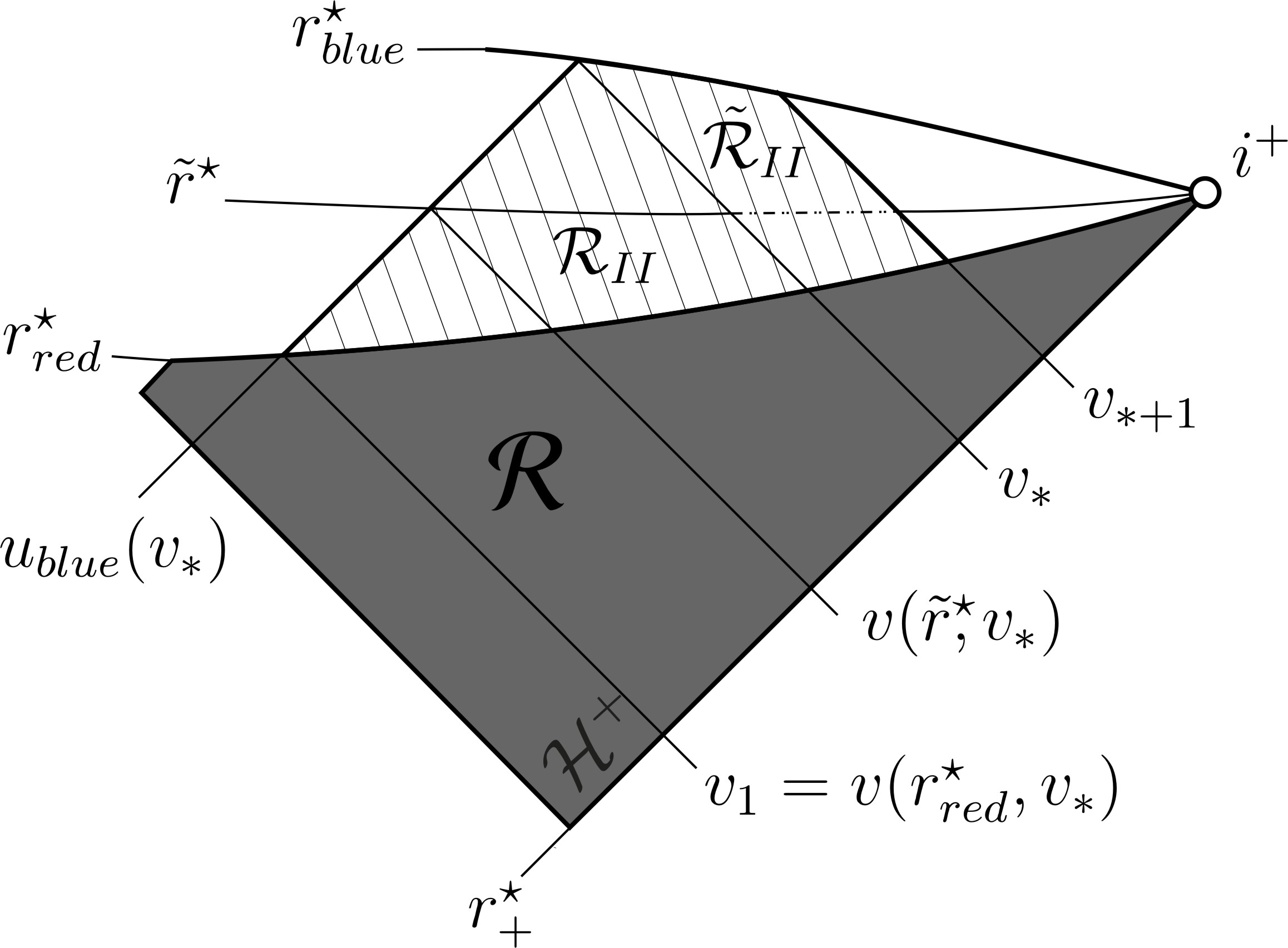}
\caption[Text der im Bilderverzeichnis auftaucht]{Region $\cR_{II}\cup \tilde{\cR}_{II}$ represented as the hatched area.}
\label{r_{red}_decay}\end{figure}}
Now we use the result of Proposition \ref{rtildedecay} and the divergence theorem in region $\cR_{II}\cup \tilde{\cR}_{II}$, to obtain decay on an arbitrary $r^{\star}=\tilde{r^{\star}}$ hypersurface, dash-dotted line, for $\tilde{r^{\star}} \in [r^{\star}_{blue},r^{\star}_{red})$. We achieve this by using Lemma \ref{Kr}, Gr\"onwall's inequality and comparability of $v_1$ and $v_*$. For more details refer to the proof of Proposition 4.5 of \cite{anne}. This proves \eqref{noprop} for $k=0$. In order to extend the proof for $k=1,2$ we will need the following lemma.
\begin{lem} 
\lb{higher_r}
Let $\psi$ be as in Theorem \ref{anfang}, and $Y^k$ as in \eqref{angular_comm} with \eqref{Yk},  \eqref{YkK}, \eqref{YkE} and all $k \in \left\{0,1,2\right\}$. Then, for $r^{\star}_{blue}$ sufficiently large, the bulk and error terms of the future directed timelike vector field $U$ can be estimated by
\bea
\lb{controll_bulk_partial_r1_1}
 |K^{U}(Y\psi)|+ |\cE^{U}(Y\psi)|+|K^{U}(Y\psi)| &\leq& B_1 J_{\mu}^{U}(Y\psi) n^{\mu}_{r^{\star}={\bar{r}^{\star}}}+B_0 J_{\mu}^{U}(\psi) n^{\mu}_{r^{\star}={\bar{r}^{\star}}},
\eea
and
\bea
\lb{controll_bulk_partial_r1_2}
 |K^{U}(Y^2\psi)|+ |\cE^{U}(Y^2\psi)|+|K^{U}(Y\psi)|+ |\cE^{U}(Y\psi)|+|K^{U}(Y\psi)|\nonumber\\
\leq B_2 J_{\mu}^{U}(Y^2\psi) n^{\mu}_{r^{\star}={\bar{r}^{\star}}}+B_1 J_{\mu}^{U}(Y\psi) n^{\mu}_{r^{\star}={\bar{r}^{\star}}}+B_0 J_{\mu}^{U}(\psi) n^{\mu}_{r^{\star}={\bar{r}^{\star}}},
\eea
in $\cN$,where $B_k$ with $k \in \left\{0,1,2\right\}$ are independent of $v_*$.
\end{lem}
\begin{proof}
The commutation with $Y^k$ does not change the $\partial_t$ invariance of the currents $K^{U}$ and $J_{\mu}^{U}$, so that we immediately get
\bea
\lb{controll_k_partial_r1}
 |K^{U}(Y^k\psi)| &\leq& \tilde{B}_k J_{\mu}^{U}(Y^k\psi) n^{\mu}_{r^{\star}={\bar{r}^{\star}}}.
\eea
Further, from equation \eqref{error_structure2} and \eqref{error_structure3} of Appendix \ref{error_app} \ref{error} we obtain
\bea
\lb{controll_e_partial_r1_1}
 |\cE^{U}(Y\psi)| &\leq& \tilde{\tilde{B}}_1 J_{\mu}^{U}(Y\psi) n^{\mu}_{r^{\star}={\bar{r}^{\star}}}+\tilde{\tilde{B}}_0 J_{\mu}^{U}(\psi) n^{\mu}_{r^{\star}={\bar{r}^{\star}}},
\eea
and
\bea
\lb{controll_e_partial_r1_2}
|\cE^{U}(Y^2\psi)|
\leq \tilde{\tilde{B}}_2 J_{\mu}^{U}(Y^2\psi) n^{\mu}_{r^{\star}={\bar{r}^{\star}}}+\tilde{\tilde{B}}_1 J_{\mu}^{U}(Y\psi) n^{\mu}_{r^{\star}={\bar{r}^{\star}}}+\tilde{\tilde{B}}_0 J_{\mu}^{U}(\psi) n^{\mu}_{r^{\star}={\bar{r}^{\star}}}.
\eea
Similar to the $K$-current the uniformity of $\tilde{\tilde{B}}_k$ follows from compactness of the region \mbox{$\cR_{II}\cup \tilde{\cR}_{II}$} and the $\partial_t$-invariance of $\cE^{U}$, see \cite{anne} for more details. 
The statement of Lemma \ref{higher_r} then follows by applying divergence theorem to all orders and summing the inequalities.
\end{proof}
The higher order decay rate of Proposition \ref{r_{red}} to the hypersurface $r^{\star}=r^{\star}_{blue}$ can now be shown by using Lemma  \ref{higher_r}, the divergence theorem and Gr\"onwall's inequality, as we have explained above and in more detail in the proof of Proposition 4.5 of \cite{anne}. 
\end{proof}
The above now also implies the following statement.
\begin{cor}
\lb{cor6.1}
Let $\psi$ be as in Theorem \ref{anfang}, $r^{\star}_{blue}$ sufficiently big as in Lemma \ref{higher_r}, $r^{\star}_{red}$ as in Proposition \ref{mi}, and $Y^k$ as in \eqref{angular_comm} with \eqref{Yk},  \eqref{YkK}, \eqref{YkE} and all $k \in \left\{0,1,2\right\}$.
Then, for all $v_*>1$ and all $\tilde{u}$ such that $r^{\star}(\tilde{u}, v_*) \in [r^{\star}_{blue}, r^{\star}_{red})$
\bea
\int\limits_ {\left\lbrace v_{red}(\tilde{u}) \leq v \leq v_{blue}(\tilde{u})\right\rbrace } J^{U}_\mu(Y^k\psi)n^\mu_{u=\tilde{u}}dVol_{u=\tilde{u}} \le C v_*^{-2-2\delta}, 
\eea
with $C$ depending on $C_{0}$ of Theorem \ref{anfang} and $D_{0}(u_{\diamond}, 1)$ of Proposition \ref{initialdataprop}, where $u_{\diamond}$ is defined by $r^{\star}_{red}=r^{\star}(u_{\diamond},1)$ and $v_{red}(\tilde{u})$, $v_{blue}(\tilde{u})$ are as in \eqref{notation_neu}.
\end{cor}
\begin{proof}
The conclusion of the statement follows by considering the divergence theorem for a triangular region \mbox{$J^-(x)\cap \cN$} with \mbox{$x=(\tilde{u}, v_{blue}(\tilde{u}))$}, $x \in J^-(r^{\star}=r^{\star}_{blue})$ and using the results of Proposition \ref{r_{red}}. Note that \mbox{$v_*\sim v_{blue}(\tilde{u})\sim v_{red}(\tilde{u})$}.
\end{proof}
{\em Remark.} Recall here, that $r^{\star}_{blue}$ was chosen such that the quantity $\frac{\partial_{\zeta}\Omega}{\Omega}$, see \eqref{omegabruch}, is negative in the future of the hypersurface $r^{\star}=r^{\star}_{blue}$.

By the previous proposition we have successfully propagated the energy estimate further inside the black hole, up to $r^{\star}=r^{\star}_{blue}$.

\subsection{Propagation through $\cB$ from $r^{\star}=r^{\star}_{blue}$ to the hypersurface $\gamma$}
\lb{blue_past}
In Section \ref{statement} we have already introduced the hypersurface $\gamma$, to which we would like to propagate the energy estimate in the following. In order to do that we will introduce its properties first, to then return to the estimate.

\subsubsection{The hypersurface $\gamma$}
\lb{gamma_curve}
The idea of the hypersurface $\gamma$, see Figures \ref{Kerr_character} and \ref{region3_neu}, was already entertained in \cite{m_interior} by Dafermos, and basically locates $\gamma$ a logarithmic distance in $v$ from a constant $r$ hypersurface living in the blueshift region. 

Let $\alpha$ be a fixed constant satisfying 
\bea
\lb{alpha}
\alpha>\frac{p+1}{\beta}, \qquad \alpha>\frac{2}{\beta},
\eea 
with $\beta$ as in \eqref{lowerboundu}.
(The significance of the bound \eqref{alpha} will become clear later in Section \ref{bounding_bulk_S}. In hindsight of \eqref{waspist} the first condition given in \eqref{alpha} implies the second.) Let us for convenience also assume that
\bea
\alpha&>&1,
\eea
and
\bea
\alpha(2-\log{2\alpha})>v_{blue}(u)+u+1,
\eea
see notation \eqref{notation_neu}.
Then, we can define a function $H(u,v)$ such that 
\bea
\lb{ableit}
\frac{\partial H}{\partial u}=1, \qquad \frac{\partial H}{\partial v}=1-\frac{\alpha}{v},
\eea
and 
so that the hypersurface $\gamma$ is the levelset 
\bea
\lb{gammadefine}
\gamma=\left\{H(u,v)=0\right\}\cap \{v> 2\alpha\},
\eea 
satisfying the relation
\bea
\lb{gamma}
v_{\gamma}(u)-v_{blue}(u)=\alpha \log v_{\gamma}(u).
\eea
The hypersurface $\gamma$ is spacelike and terminates at $i^+$. (In the notation \eqref{notation_neu}, \mbox{$u_{\gamma}(v)\rightarrow -\infty$} as \mbox{$v\rightarrow \infty$}.)
Note that by our choices $u<-1$ and $v>|u|$ in $\cD^+(\gamma)$.

As we shall see in Section \ref{finiteness} the above properties of $\gamma$ will allow us to derive pointwise estimates of $\Omega^2$ in $J^+(\gamma)\cap \cB$. We first turn however to the region $J^-(\gamma)\cap \cB$. 

\subsubsection{The vector field multiplier $S_0$ used in region $J^-(\gamma)\cap \cB$}
\lb{functionf}
Now we are ready to propagate the energy estimates further into the blueshift region $\cB$ up to the hypersurface $\gamma$.
We will in this part of the proof use the vector field 
\bea
\lb{S0field}
S_0=f^q(\partial_u+\partial_v+b^{\tilde{\phi}}\partial_{{\tilde{\phi}}}),
\eea
where the constant $q$ will finally be fixed in Lemma \ref{K_S_0_E}. Further, the function $f$ depends merely on $u$ and $v$ and has to satisfy
\bea
\lb{fcon1}
f(u,v)&\geq& 0, \\
\lb{ableit_neg}
\partial_{\zeta} f(u,v)&<&0,\\
\lb{fcon3}
c\leq \,\,\left|\frac{\partial_{\zeta} f(u,v)}{\Omega^2}\right|&\leq&C.
\eea
This construction is very similar to the construction of the vector field $S_0$ of \cite{anne}, where we used $r$ instead of $f$. We have introduced $f$ here to emphasize that we do not want the $\theta^{\star}$-dependence of $r$, which for Kerr spacetime leads to extra terms of indefinite sign in the bulk. Once we associate constant values to the angular variables we can derive the above properties from the function $r^{\star}(u,v,\theta^{\star})$. In order to do that we consider $\md r^{\star}$, using \eqref{partials} for ${\theta^{\star}} \to 0$ and obtain
\bea
\lim_{\theta^{\star} \to 0}\md {r^{\star}}=\frac{r^2+a^2}{\Delta}(u,v, \theta^{\star}) \md r.
\eea
Further, by \eqref{R_delta_sigma} we get
\bea
\lim_{\theta^{\star} \to 0}{\Omega}^2(u,v, \theta^{\star})=-\frac{\Delta}{r^2+a^2}(u,v, \theta^{\star}).
\eea
Therefore, the choice $\lim_{\theta^{\star} \to 0} f(u,v)=r(u,v, \theta^{\star})$ implies \mbox{$\partial_{\zeta} f= \frac{\partial f}{\partial r} \frac{ \partial r}{\partial r^{\star}} \frac{\partial  r^{\star}}{\partial \zeta}$} and with \eqref{delta_r} we see that \eqref{fcon1}-\eqref{fcon3} are satisfied.

\subsubsection{Positivity of the bulk term $K^{S_0}$} 
Let us now consider the bulk term and derive conditions for positivity so that the future energy flux can be estimated by the initial flux when carrying out the divergence theorem.
\begin{lem}
\label{K_S_0}
Let $\psi$ be an arbitrary function. Then, for the vector field $S_0$ as in \eqref{S0field} and a suitable choice of $q$, 
\ben
\lb{lemma_B_gamma_past}
K^{S_0}(\psi)&\geq& 0,
\een
in $\cB$.
\end{lem}
\begin{proof} 
Plugging \eqref{S0field} in \eqref{Kplug_edd-f} of Appendix \ref{Kcurrents} 
in Eddington--Finkelstein-like coordinates leads to the following expression for $K^{S_0}(\psi)$.
\bea
\lb{KS0_zuerst}
K^{S_0}(\psi)&=& -q \frac{f^{q-1}}{2\Omega^2}\partial_u f(\partial_u \psi)^2\nonumber\\
&&-q \frac{f^{q-1}}{2\Omega^2}\partial_v f(\partial_v \psi)^2\nonumber\\
&&+\left(-q \frac{f^{q-1}}{2\Omega^2}(\partial_v f+\partial_u f)+B_1\right)|\nabb \psi|^2\nonumber\\
&&+B_2(\partial_u \psi\partial_v \psi)\nonumber\\
&&+\left[-q \frac{b^{\tilde{\phi}}}{2\Omega^2}f^{q-1}(\partial_v f+\partial_u f)+ B_3\right](\partial_u \psi\partial_{\tilde{\phi}} \psi)\nonumber\\
&&+\left[-q \frac{b^{\tilde{\phi}}}{\Omega^2}f^{q-1}(\partial_u f)+ B_4\right](\partial_v \psi\partial_{\tilde{\phi}} \psi)\nonumber\\
&&+\left[-q \frac{|b^{\tilde{\phi}}|^2}{\Omega^2}f^{q-1}(\partial_u f)+B_5^{{\tilde{\phi}}{\tilde{\phi}}}\right](\partial_{\tilde{\phi}} \psi\partial_{\tilde{\phi}} \psi)\nonumber\\
&&+B_6^{\tilde{\phi}\theta^{\star}}(\partial_{\tilde{\phi}} \psi\partial_{\theta^{\star}} \psi)\nonumber\\
&&+B_7^{\theta^{\star} \theta^{\star}}(\partial_{\theta^{\star}} \psi\partial_{\theta^{\star}} \psi),
\eea
where all $B$ coefficients are bounded. 
Note that our free parameter $q$ exclusively multiplies terms which are positive, see \eqref{ableit_neg}. Moreover, the terms multiplied by $b^{\tilde{\phi}}$ are small close to $\cC\cH^+$ since we have chosen $r^{\star}_{blue}$ close enough to $\cC\cH^+$ and such that the quantity given in \eqref{omegabruch} is negative.
We can now apply the Cauchy--Schwarz inequality to \eqref{KS0_zuerst} to obtain the following simplified expression
\bea
\lb{KS0_even_simpler}
K^{S_0}(\psi)&\geq&\quad q P_1(\partial_u \psi)^2\nonumber\\
&&+q P_2(\partial_v \psi)^2\nonumber\\
&&+q {P}_3\left[(\partial_{\tilde{\phi}}\psi)^2+ (\partial_{\theta^{\star}}\psi)^2\right]\nonumber\\
&&+\tilde{B}_2(\partial_u \psi\partial_v \psi)\nonumber\\
&&+\tilde{B}_3(\partial_u \psi\partial_{\tilde{\phi}} \psi)\nonumber\\
&&+\tilde{B}_4(\partial_v \psi\partial_{\tilde{\phi}} \psi),
\eea
where the $P$ coefficients denote positive terms and the $\tilde{B}$ coefficients are again bounded terms. 
Using the Cauchy--Schwarz inequality again for the first three rows,
we obtain the conditions
\bea
\lb{q_ineq}
\tilde{B}_2\leq q\sqrt{P_1P_2}, \qquad \tilde{B}_3\leq q\sqrt{P_1\tilde{P}_3}, \qquad \tilde{B}_4\leq q\sqrt{P_2\tilde{P}_3} .
\eea
The lemma then follows for $q$ sufficiently large to satisfy all derived conditions for the positivity of the
current $K^{S_0}(\psi)$ in $\cB$.\\
\end{proof}

\subsubsection{Propagation up to the hypersurface $\gamma$}
\begin{prop}
\lb{to_gamma0}
Let $\psi$ be as in Theorem \ref{anfang}, and $Y^k$ as in \eqref{angular_comm} with \eqref{Yk},  \eqref{YkK}, \eqref{YkE} and all $k \in \left\{0,1,2\right\}$. Then, for all $v_*> 2\alpha$ 
\bea
\lb{sieben0}
\int\limits_ {\left\lbrace  v_* \leq v \leq 2v_* \right\rbrace } J_{\mu}^{{S_0}}(Y^k\psi) n^{\mu}_{\gamma} \dV_{\gamma}
\leq
C{v_*}^{-1-2\delta}, 
\eea
on the hypersurface $\gamma$, with $C$ depending on $C_{0}$ of Theorem \ref{anfang} and $D_{0}(u_{\diamond}, 1)$ of Proposition \ref{initialdataprop}, where $u_{\diamond}$ is defined by $r^{\star}_{red}=r^{\star}(u_{\diamond},1)$.
\end{prop}
{\em Remark.} $n^{\mu}_{\gamma}$ denotes the normal vector on the hypersurface $\gamma$ which is a levelset \mbox{$\gamma=\left\{H(u,v)=0\right\}$} of the function $H(u,v)$ as explained in Section \ref{gamma_curve}. Note that for big values of $v$ the current $J_{\mu}^{{S_0}}(Y^k\psi) n^{\mu}_{\gamma}$ approaches $J_{\mu}^{{S_0}}(Y^k\psi) n^{\mu}_{r}$ since $\partial_v H$ approaches $1$, see equation \eqref{ableit}.

\begin{proof}
Let $v_* > 2\alpha$, such that $\gamma$ is spacelike for $v > v_*$, cf.~ Section \ref{gamma_curve}.
Then, the proof of \eqref{sieben0} with $k=0$ follows from using the result of Proposition \ref{r_{red}} and Lemma \ref{K_S_0} in the divergence theorem in region $\cR_{III}$, see Figure \ref{region3_neu}.
{\begin{figure}[ht]
\centering
\includegraphics[width=0.65\textwidth]{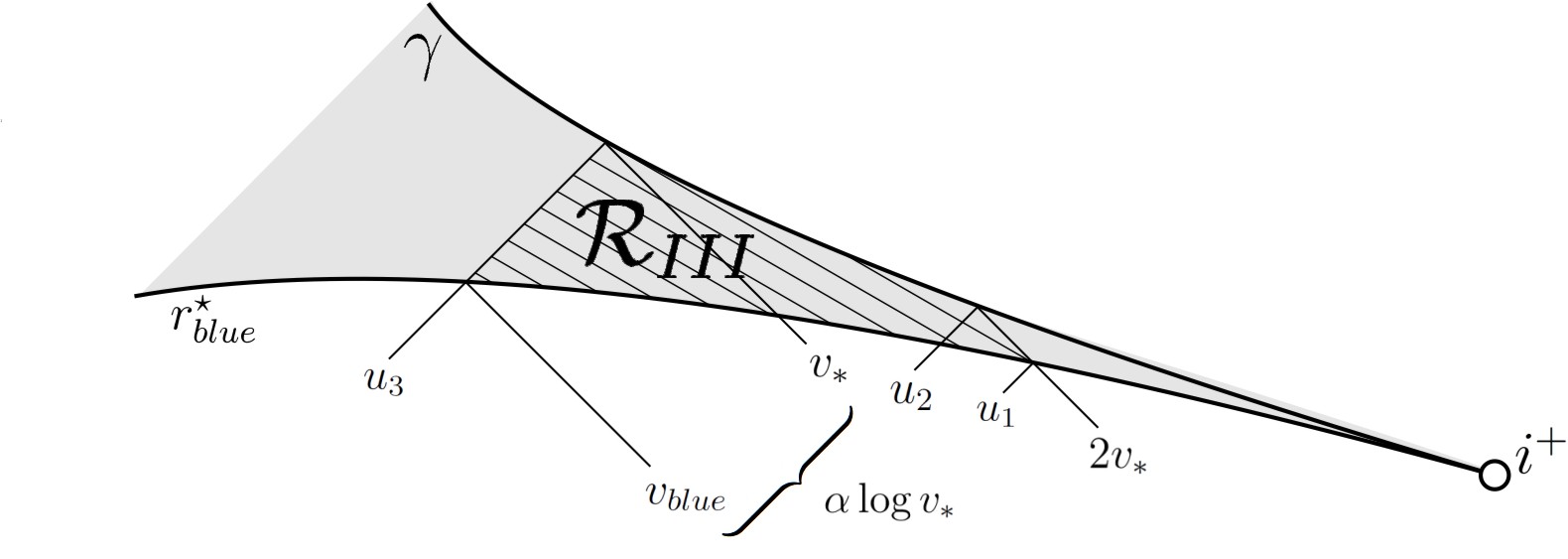}
\caption[Text der im Bilderverzeichnis auftaucht]{Logarithmic distance of hypersurface $r^{\star}=r^{\star}_{blue}$ and hypersurface $\gamma$ depicted in a Penrose diagram.}
\label{region3_neu}\end{figure}}

Now we need to show that a higher order analog of Lemma \ref{K_S_0} holds. The following lemma will allow us to neglect both error and bulk terms when carrying out the divergence theorem again.
\begin{lem}
\label{K_S_0_E}
Let $\psi$ be a solution of the wave equation \eqref{wave_psi} on a subextremal Kerr background $(\cM,g)$, with mass $M$ and angular momentum per unit mass $a$, and $M>|a|\neq 0$, and further $Y^k$ is as in \eqref{angular_comm} with \eqref{Yk},  \eqref{YkK}, \eqref{YkE} and all $k \in \left\{0,1,2\right\}$. Then, for the vector fields $S_0$ as in \eqref{S0field} and a suitable choice of $q$, 
\bea
\lb{lemma_B_gamma_past}
K^{S_0}(Y \psi) +\cE^{S_0}(Y \psi) +K^{S_0}(\psi) &\geq& 0, 
\eea
and
\bea
\lb{lemma_B_gamma_past2}
K^{S_0}(Y^2 \psi) +\cE^{S_0}(Y^2 \psi) +K^{S_0}(Y \psi) +\cE^{S_0}(Y \psi) +K^{S_0}(\psi)&\geq& 0, 
\eea
in $\cB$.
\end{lem}
\begin{proof}  
To prove statement \eqref{lemma_B_gamma_past}
see Appendix \ref{error_app} \ref{error} and notice that we can express the error term by
\bea
\lb{error_structure_phi}
\cE^{S_0}(Y^{\theta_A}\partial_{\theta_A} \psi)\stackrel{\eqref{angular_comm}}{=}\cE^{S_0}(Y\psi)&\stackrel{\eqref{error_structure2}}{=}&\left[ {E}_1 (\partial_u  \psi) +{E}_2 (\partial_v \psi) +{E}_3 (\partial_{\tilde{\phi}}  \psi) +{E}_4 (\partial_{\theta^{\star}}  \psi)\right.\nonumber\\
&& + {F}_1 (\partial_u \partial_{\tilde{\phi}}\psi)+ {F}_2 (\partial_v \partial_{\tilde{\phi}}\psi) + {F}_3 (\partial_{\theta^{\star}} \partial_{\tilde{\phi}} \psi)+{F}_4 (\partial_{\tilde{\phi}}^2 \psi)\nonumber\\
&& \left.+  {G}_1 (\partial_u \partial_{\theta^{\star}} \psi)+  {G}_2 (\partial_v \partial_{\theta^{\star}} \psi) + {G}_3 (\partial_{\theta^{\star}}^2\psi)
\right]\nonumber\\
&& \times f^q\left(\partial_u(Y \psi)+\partial_v(Y\psi)+b^{\theta_B}\partial_{{\theta_B}}(Y\psi)\right), 
\eea
where each coefficient is bounded. In the proof of Lemma \ref{K_S_0} we have seen, that the first three rows of \eqref{KS0_even_simpler} can absorb all remaining terms by choosing $q$ big enough. Similarly, the first and the last term of statement \eqref{lemma_B_gamma_past} can now absorb the error term by using the Cauchy--Schwarz inequality and the right choice of $q$.
This can be seen from writing out the bulk term $K^{S_0}(Y \psi)$ as
\bea
\lb{KS0_simpler_higher}
K^{S_0}(Y\psi)&=&\quad q P_1(\partial_u (Y\psi))^2\nonumber\\
&&+q P_2(\partial_v (Y \psi))^2\nonumber\\
&&+q {P}_3\left[(\partial_{\tilde{\phi}}(Y\psi))^2+ (\partial_{\theta^{\star}} (Y\psi))^2\right]\nonumber\\
&&+\tilde{B}_2(\partial_u (Y\psi)\partial_v (Y\psi))\nonumber\\
&&+\tilde{B}_3(\partial_u (Y \psi)\partial_{\theta_B} (Y\psi))\nonumber\\
&&+\tilde{B}_4(\partial_v (Y\psi)\partial_{\theta_B} (Y\psi)).
\eea
As before, the control follows from the first three terms which are positive and multiplied by $q$ which we can choose as large as needed. Recall, that $Y^{\theta_A}$, defined in \eqref{angular_comm}, does not depend on $u,v$. The ${G}_1$ term of \eqref{error_structure_phi} can now be absorbed in the first term of $K^{S_0}(Y\psi)$, the $G_2$ term in the second and the $G_3$ term can be absorbed in the third term, if we choose the value of our parameter $q$ big enough. And similarly the last three terms, that multiply all and come from the vector field multiplier, can also be absorbed in the first three terms of the bulk, respectively, given that we choose $q$ big enough. The same applies to the $F$ terms, which can alternatively also be controlled by the lower order terms combined with commutation of the Killing vector $\partial_{\phi}$. For an example to see how to absorb the terms see the proof of Lemma \ref{K_S_0}. The above described use of the Cauchy--Schwarz inequality proves \eqref{lemma_B_gamma_past}. 
Further, for the second commutation, to prove statement \eqref{lemma_B_gamma_past2}, we can write the error term as follows:
\bea
\lb{error_structure_second}
\cE^{S_0}(Y^{\theta_A}\partial_{\theta_A} (Y^{\theta_B}\partial_{\theta_B} \psi))&\stackrel{\eqref{Y2}}{=}&\cE^{S_0}(Y^2 \psi)\nonumber\\
&\stackrel{\eqref{error_structure3}}{=}&\left[ {E}_1 (\partial_u \psi) +{E}_2 (\partial_v \psi) +{E}_3 (\partial_{\tilde{\phi}} \psi) +{E}_4 (\partial_{\theta^{\star}} \psi)\right.\nonumber\\
&& + {F}_1 (\partial_u \partial_{\tilde{\phi}} \psi) + {F}_2 (\partial_v \partial_{\tilde{\phi}} \psi) +{F}_3 (\partial_{\theta^{\star}} \partial_{\tilde{\phi}} \psi)+{F}_4 (\partial_{\tilde{\phi}}^2 \psi)\nonumber\\
&& +  {G}_1 (\partial_u \partial_{\theta^{\star}} \psi) + {G}_2 (\partial_v \partial_{\theta^{\star}} \psi) +{G}_3 (\partial_{\theta^{\star}}^2\psi) \nonumber\\
&& +  {H}_1 (\partial_u \partial_{\theta^{\star}} \partial_{\tilde{\phi}} \psi) + {H}_2 (\partial_v \partial_{\theta^{\star}} \partial_{\tilde{\phi}} \psi) +{H}_3 (\partial_{\theta^{\star}}^2 \partial_{\tilde{\phi}}\psi) \nonumber\\
&&  + {I}_1 (\partial_u \partial_{\tilde{\phi}}^2 \psi) + {I}_2 (\partial_v \partial_{\tilde{\phi}}^2 \psi)+{I}_3 (\partial_{\theta^{\star}} \partial_{\tilde{\phi}}^2 \psi)+{I}_4 (\partial_{\tilde{\phi}}^3 \psi)\nonumber\\
&& \left.+  {J}_1 (\partial_u \partial_{\theta^{\star}}^2 \psi) +  {J}_2 (\partial_v \partial_{\theta^{\star}}^2 \psi)+ {J}_3 (\partial_{\theta^{\star}}^3\psi)
\right]\nonumber\\
&& \times f^q\left(\partial_u(Y^2\psi)+\partial_v(Y^2\psi)+b^{\theta_C}\partial_{{\theta_C}}(Y^2\psi)\right). 
\eea
Again we will control all remaining terms by the good terms of $K^{S_0}(Y^2\psi)$, which are the first three rows of
\bea
\lb{KS0_simpler_higher_2}
K^{S_0}(Y^2\psi)&=&\quad q P_1(\partial_u (Y^2\psi))^2\nonumber\\
&&+q P_2(\partial_v (Y^2 \psi))^2\nonumber\\
&&+q {P}_3\left[(\partial_{\tilde{\phi}}(Y^2\psi))^2+ (\partial_{\theta^{\star}} (Y^2\psi))^2\right]\nonumber\\
&&+\tilde{B}_2(\partial_u (Y^2\psi)\partial_v (Y^2 \psi))\nonumber\\
&&+\tilde{B}_3(\partial_u (Y^2 \psi)\partial_{{\theta_B}} (Y^2 \psi))\nonumber\\
&&+\tilde{B}_4(\partial_v (Y^2 \psi)\partial_{{\theta_B}} (Y^2 \psi)).
\eea
We can see now that the $J_1$ term of \eqref{error_structure_second} gets controlled by the first row and $J_2$ by the second row. $J_3$ gets controlled by the third row of the above expression. The three remaining terms multiplying everything get controlled by the first three rows together. The $H$ and $I$ also get controlled in the same manner or alternatively by lower orders combined with commutation by the Killing vector $\partial_{\phi}$. All other terms get controlled by the lower order $K^{S_0}(Y\psi)$ and $K^{S_0}(\psi)$ terms.
Analogously to before, this proves statement \eqref{lemma_B_gamma_past2} of Lemma \eqref{K_S_0_E}.
\end{proof}
To prove Proposition \ref{to_gamma0} for $k=1$ we now use the divergence theorem for the currents depending on $\psi$ added to the divergence theorem for the currents depending on $Y\psi$. This together with statement \eqref{lemma_B_gamma_past} then leads to the proof of Proposition \ref{to_gamma0} for $k=1$. Similarly, 
applying the divergence theorem again for currents depending on $Y^2\psi$ and adding this to the previous inequalities, then proves Proposition \ref{to_gamma0} for $k=2$.
\end{proof}
We could go on proving Lemma \ref{K_S_0_E} and Proposition \ref{to_gamma0} for even higher orders, but in order to prove pointwise boundedness for $\psi$ the second order of derivatives is sufficient.


\subsubsection{Weighted dyadic propagation to the hypersurface $\gamma$} 
Recall, that we eventually require a weighted energy estimate. We therefore state the following.
\begin{cor}
\lb{cor2}
Let $\psi$ be as in Theorem \ref{anfang}, and $Y^k$ as in \eqref{angular_comm} with \eqref{Yk},  \eqref{YkK}, \eqref{YkE} and all $k \in \left\{0,1,2\right\}$. Then, for all $v_*>2\alpha$, see Section \ref{gamma_curve},
\bea
\lb{Nstargammaesti}
\int\limits_ {\left\lbrace  v_* \leq v <\infty \right\rbrace }v^p J_{\mu}^{{S_0}}(Y^k\psi) n^{\mu}_{\gamma} \dV_{\gamma}
&\leq& Cv_*^{-1-2\delta+p}, 
\eea
on the hypersurface $\gamma$, with $C$ depending on $C_{0}$ of Theorem \ref{anfang} and $D_{0}(u_{\diamond}, 1)$ of Proposition \ref{initialdataprop}, where $u_{\diamond}$ is defined by $r^{\star}_{red}=r^{\star}(u_{\diamond},1)$, $p$ as in \eqref{waspist}.
\end{cor}
\begin{proof}
Let $v_*>2\alpha$, as before in Proposition \ref{to_gamma0}.
Then, the Corollary follows by weighting \eqref{sieben0} with $v_*^p$ and summing dyadically. For more details see \cite{anne_thesis}.\\
\end{proof}
Further, we can state the following.
\begin{cor}
\lb{cor7.1}
Let $\psi$ be as in Theorem \ref{anfang}, and $Y^k$ as in \eqref{angular_comm} with \eqref{Yk},  \eqref{YkK}, \eqref{YkE} and all $k \in \left\{0,1,2\right\}$. Further, $\gamma$ as in \eqref{gammadefine} and $r^{\star}_{blue}$ sufficiently large. Then, for all $v_*>2\alpha$  and for all $\tilde{u} \in [u_{blue}(v_*), u_{\gamma}(v_*))$
\bea
\int\limits_ {\left\lbrace  v_{blue}(\tilde{u}) \leq v \leq v_{\gamma}(\tilde{u}) \right\rbrace } v^pJ^{S_0}_\mu(Y^k \psi)n^\mu_{u=\tilde{u}}dVol_{u=\tilde{u}} \le C v_*^{-1-2\delta+p},
\eea
with $C$ depending on $C_{0}$ of Theorem \ref{anfang} and $D_{0}(u_{\diamond}, 1)$ of Proposition \ref{initialdataprop}, where $u_{\diamond}$ is defined by $r^{\star}_{red}=r^{\star}(u_{\diamond},1)$ and $v_{\gamma}(\tilde{u})$, $v_{blue}(\tilde{u})$ as in \eqref{notation_neu}.
\end{cor}
\begin{proof}
The proof is similar to the proof of Corollary \ref{cor6.1} by considering the divergence theorem for a triangular region \mbox{$J^-(x)\cap \cB$} with \mbox{$x=(\tilde{u}, v_{\gamma}(\tilde{u}))$}, $x \in J^-(\gamma)$ and using the results of the proof of Proposition \ref{to_gamma0}.
\end{proof}

\subsection{Propagation through $\cB$ from the hypersurface $\gamma$ to $\cC\cH^+$ in the neighbourhood of $i^+$}
\lb{blue_future}
In the following section we want to close the proof of boundedness for weighted energies up to $\cC\cH^+$ inside the characteristic rectangle $\Xi$ by considering the region $J^+(\gamma)\cap \cB$.
We will use the weighted vector field multiplier
\bea
\lb{Sfield}
S=|u|^p\partial_u+v^p\partial_v+v^pb^{\tilde{\phi}}\partial_{{\tilde{\phi}}},
\eea
for an arbitrary $p$ satisfying \eqref{waspist}.
From this, using \eqref{Sfield} in \eqref{Kplug_edd-f} we obtain 
\bea
\lb{KS}
K^S&=& 
-\left[\frac{p}2\left[|u|^{p-1}+v^{p-1}\right] +|u|^p\frac{\partial_{u} \Omega}{\Omega} +v^p\frac{\partial_{v} \Omega}{\Omega}+|u|^p\frac14 \frac{\partial_{u}\left(  L^2\sin^2\theta\right)}{ L^2\sin^2\theta}\right.\nonumber\\
&& \qquad\left.+v^p\frac14 \frac{\partial_{v}\left(  L^2\sin^2\theta\right)}{ L^2\sin^2\theta}\right]|\nabb \psi|^2\nonumber\\
&&+\frac14\left[ |u|^p\frac{\partial_{u}\left(  L^2\sin^2\theta\right)}{ L^2\sin^2\theta}+v^p\frac{\partial_{v}\left(  L^2\sin^2\theta\right)}{ L^2\sin^2\theta}\right]\frac{1}{\Omega^2}(\partial_u \psi)(\partial_v \psi+b^{\tilde{\phi}}\partial_{\tilde{\phi}}\psi)\nonumber\\
&&+\left[-\frac{b^{\tilde{\phi}}}{2\Omega^2}pv^{p-1} -\frac{b^{\tilde{\phi}}}{2\Omega^2}p|u|^{p-1}-\frac{b^{\tilde{\phi}}}{\Omega^2}\left[\frac{\partial_u \Omega}{\Omega}|u|^p+\frac{\partial_v \Omega}{\Omega}v^p\right] 
+\frac{\partial_u b^{\tilde{\phi}} |u|^p}{2\Omega^2}\right](\partial_u \psi\partial_{\tilde{\phi}} \psi)\nonumber\\
&&+\left[-\frac{v^p\partial_ub^{\tilde{\phi}}}{2\Omega^2}\right](\partial_v \psi\partial_{\tilde{\phi}} \psi)\nonumber\\
&&+\left[-\frac{v^p}{4\Omega^2}\left[b^{\theta_C}\partial_ub^{\theta_D}+b^{\theta_D}\partial_ub^{\theta_C}\right]\right.\nonumber\\
&& \qquad\left.+\frac12(\gin^{-1})^{{\theta_C}{\theta_A}}(\gin^{-1})^{{\theta_D}{\theta_B}}\partial{\eta}\gin_{{\theta_A}{\theta_B}}S^{\eta} \right](\partial_{{\theta_C}} \psi\partial_{{\theta_D}} \psi),
\eea
where we have used \eqref{det_regel}.
Note that $\partial_u b^{\tilde{\phi}}$ decays as shown in Section \ref{bounded} and consequently we will be able to absorb these terms in the following analysis.

For the $J$-currents we obtain
\bea
\lb{JplugnvS}
J^S_{\mu}n^{\mu}_{v=const}&=&\frac1{2\Omega^2}\left[\quad|u|^p(\partial_u \psi)^2+ \Omega^2 v^p |\nabb \psi|^2\right]\\
\lb{JplugnuS}
J^S_{\mu}n^{\mu}_{u=const}
&=&\frac1{2\Omega^2}\left[v^p(\partial_v\psi+b^{\tilde{\phi}}\partial_{{\tilde{\phi}}}\psi)^2
+|u|^p{\Omega^2}|\nabb \psi|^2\right].
\eea

\subsubsection{Pointwise estimates on $\Omega^2$ in $J^+(\gamma)$}
\lb{finiteness}
Before we can bound the bulk term $K^S$ we will first need to derive pointwise estimates on $\Omega^2$ in $J^+(\gamma)$. We note that together with the results of Section \ref{bounded}, this will imply that the
spacetime volume to the future of the hypersurface $\gamma$ is finite, $\operatorname{Vol}(J^+(\gamma))<C$.
Let us state
\bea
\lb{omega1}
2\frac{\partial_v \Omega}{\Omega}= \frac{\partial_v \Omega^2}{\Omega^2}\stackrel{\eqref{R_delta_sigma}}{=}\left[\frac{\partial_v\Delta}{\Delta}-\frac{\partial_v (R^2)}{R^2}\right]\leq -\beta,
\eea
where $\beta$ is a positive constant, introduced in \eqref{lowerboundv}, for $r^{\star}\leq r^{\star}_{blue}$.
We can now derive a future decay bound along a constant $u$ hypersurface for the function $\Omega^2(u,v, \theta^{\star})$ for $(u, v) \in \cB$.
Let \mbox{$x=(u_{fix}, v_{fix})$, $x \in \cB$}, 
then, from \eqref{omega1} it follows from integration in $v$, that
\bea
\lb{lambda_comp}
\left.\log\left({\Omega^2(u_{fix},{v}, \theta^{\star})}\right)\right|^{\bar{v}}_{v_{fix}}
&\leq&-\beta\left[\bar{v}-{v_{fix}}\right],
\eea
which leads to
\bea
\lb{omegaufix}
\Omega^2(\bar{u},\bar{v}, \theta^{\star}){\leq} \Omega^2(\bar{u}, v_{fix}, \theta^{\star}) e^{-\beta[\bar{v}-v_{fix}]}, \quad \mbox{for all $(\bar{u}, v_{fix}) \in \cB$ and $\bar{v}>v_{fix}$}.
\eea
Analogously, we obtain
\bea
\lb{omegavfix}
\Omega^2(\bar{u},v_{fix}, \theta^{\star}){\leq} \Omega^2(u_{fix}, v_{fix}, \theta^{\star}) e^{-\beta[\bar{u}-u_{fix}]}, \quad \mbox{for all $(\bar{u},v_{fix}) \in J^+(x)$},
\eea
and plugging \eqref{omegaufix} into \eqref{omegavfix} it yields
\bea
\lb{omegafix}
\Omega^2(\bar{u},\bar{v}, \theta^{\star}){\leq} \Omega^2(u_{fix}, v_{fix}, \theta^{\star}) e^{-\beta[\bar{u}-u_{fix}+\bar{v}-v_{fix}]}, \quad \mbox{for all $(\bar{u},\bar{v})  \in J^+(x)$}.
\eea
Equation \eqref{omegaufix} together with \eqref{gamma} lead to the relation 
\bea
\lb{lambdalog}
&&\Omega^2(\bar{u}, \bar{v}, \theta^{\star}) \leq \Omega^2(\bar{u}, v_{blue}(\bar{u}), \theta^{\star}) e^{-\beta\alpha \log v_{\gamma}(\bar{u})}=\Omega^2(\bar{u}, v_{blue}(\bar{u}), \theta^{\star}) {v_{\gamma}(\bar{u})}^{-\beta \alpha},\nonumber\\
&&\mbox{for $(\bar{u}, \bar{v}) \in \gamma$},
\eea
constituting poinwise decay for $\Omega^2(u,v, \theta^{\star})$ on the hypersurface $\gamma$.
For $J^+(\gamma)$, using \eqref{omegaufix} and \eqref{lambdalog} we further get
\bea
\lb{spacevolumedecay}
\Omega^2(\bar{u}, \bar{v}, \theta^{\star}) \leq C{v_{\gamma}(\bar{u})}^{-\beta \alpha} e^{-\beta\left[\bar{v}-v_{\gamma}(\bar{u})\right]}, \quad \mbox{for $(\bar{u}, \bar{v}) \in J^+(\gamma)$},
\eea
where we have used $\Omega^2(\bar{u}, v_{blue}(\bar{u}), \theta^{\star})\leq C$.
Moreover, we may think of a parameter $\bar{v}$ which determines the associated $u$ value via intersection with $\gamma$. We denote this value by $u_{\gamma}(\bar{v})$ which was introduced in \eqref{notation_neu}, cf.~ Figure \ref{integralbild3} b). 

Further, by \eqref{R_delta_sigma} we can also state 
\bea
\lb{spacevolumedecay_u}
\Omega^2(\bar{u}, \bar{v}, \theta^{\star}) \leq C {|u_{\gamma}(\bar{v})|}^{-\beta \alpha} e^{\beta\left[u_{\gamma}(\bar{v})-\bar{u}\right]} \quad \mbox{for $(\bar{u}, \bar{v}) \in J^+(\gamma)$}.
\eea
Recall the choice \eqref{alpha} of $\alpha$. 

From the fact that 
$|u_{\gamma}(\bar{v})| \sim \bar{v}$ together with the extra exponential factor of \eqref{spacevolumedecay_u},  we see that pointwise decay of $\Omega^2$ to the future of $\gamma$ follows. This is an important result, since in Section \ref{bounded} we show that most of the quantities showing up in our bulk and error terms are proportional to $\Omega^2$ and can therefore be controlled or absorbed, respectively.
Note further, that the above implies that the
spacetime volume to the future of the hypersurface $\gamma$ is finite, 
\bea
\lb{finitevol}
\operatorname{Vol}(J^+(\gamma))<C.
\eea

\subsubsection{Bounding the bulk terms $K^S$ and the error terms $\cE^S$}
\lb{bounding_bulk_S}
First of all, recall the bulk term $K^S$ stated in \eqref{KS} and let us rewrite it in the following simplified way. 
\bea
\lb{KS2}
K^S&=& \left[|u|^p\left|\frac{\partial_{u} \Omega}{\Omega}\right| +v^p\left|\frac{\partial_{v} \Omega}{\Omega}\right|+N\right]|\nabb \psi|^2\nonumber\\
&&+\frac14\left[ |u|^p\frac{\partial_{u}\left( L^2\sin^2\theta\right)}{L^2\sin^2\theta}+v^p\frac{\partial_{v}\left( L^2\sin^2\theta\right)}{L^2\sin^2\theta}\right]\frac{1}{\Omega^2}(\partial_u \psi)(\partial_v \psi+b^{\tilde{\phi}}\partial_{\tilde{\phi}}\psi)\nonumber\\
&&+\left[S_1
+\frac{\partial_u b^{\tilde{\phi}} |u|^p}{2\Omega^2}\right](\partial_u \psi\partial_{\tilde{\phi}} \psi)
+\left[-\frac{v^p\partial_ub^{\tilde{\phi}}}{2\Omega^2}\right](\partial_v \psi\partial_{\tilde{\phi}} \psi).\nonumber\\
&&+S_2(\partial_{\theta_C} \psi\partial_{{\theta_D}} \psi).
\eea
The term $N$ is negative, but due to the weights $|u|^p$ and $v^p$ the first two terms are dominating and the factor multiplying $|\nabb \psi|^2$ is positive. Further, the terms $S_1$ and $S_2$ are small close to $\cC\cH^+$.
This is obtained by using expressions \eqref{b_phi} to \eqref{glp} in \eqref{KS}. Since these terms can be absorbed by using the Cauchy--Schwarz inequality we conclude, that only $\frac{|u|^p\partial_ub^{\tilde{\phi}}}{2\Omega^2}$ remains as a relevant term multiplying  $(\partial_u \psi \partial_{\tilde{\phi}} \psi)$ together with an analog statement for $(\partial_v \psi \partial_{\tilde{\phi}} \psi)$.
Further, by the Cauchy--Schwarz inequality we obtain
\bea
\lb{vp_cauchy}
-\frac{v^p\partial_ub^{\tilde{\phi}}}{2\Omega^2}(\partial_v \psi \partial_{\tilde{\phi}} \psi)
&\leq& \frac{v^p\sqrt{\partial_u b^{\tilde{\phi}}}}{4\Omega^2}\left[(\partial_v \psi)^2 +{\partial_u b^{\tilde{\phi}}}(\partial_{\tilde{\phi}} \psi)^2\right]\nonumber\\
&\leq& \frac{v^p}{4\Omega^2}\sqrt{\partial_u b^{\tilde{\phi}}}(\partial_v \psi)^2 +C \frac{v^p}{4}\sqrt{\partial_u b^{\tilde{\phi}}}(\partial_{\tilde{\phi}} \psi)^2, 
\eea
and the analog expression for the $(\partial_u \psi \partial_{\tilde{\phi}} \psi)$ term
in expression \eqref{KS}. Note, that we have used \eqref{b_bound} and \eqref{R_delta_sigma} in the second step of \eqref{vp_cauchy}. By using \eqref{b_phi} again, we define
\bea
\lb{KStilde}
|\tilde{K}^S(\psi)|&=&C\left|\frac{\partial_{u}( L^2\sin^2\theta)}{L^2\sin^2\theta}+\frac{\partial_{v}( L^2\sin^2\theta)}{L^2\sin^2\theta}\frac{v^p}{|u|^p}+\sqrt{\partial_u b^{\tilde{\phi}}}\right|\frac{1}{2\Omega^2}|u|^p(\partial_u \psi)^2\nonumber\\
&&+C\left|\frac{\partial_{v}( L^2\sin^2\theta)}{L^2\sin^2\theta}+\frac{\partial_{u}( L^2\sin^2\theta)}{L^2\sin^2\theta}\frac{|u|^p}{v^p}+\sqrt{\partial_u b^{\tilde{\phi}}}\right|\nonumber\\
&& \qquad \times\frac{1}{2\Omega^2}v^p(\partial_v \psi+b^{\tilde{\phi}}\partial_{\tilde{\phi}}\psi)^2,
\eea
where $C$ is such that
\bea
\lb{ksrelation}
-{K}^S(\psi)\leq |\tilde{K}^S(\psi)|
\eea
is satisfied.
Now we can prove the following lemma.
\begin{lem}
\label{K_spher}
Let $\psi$ be an arbitrary function. Then, for all $v_*>2\alpha$ 
and all $\hat{v}>v_*$,
the integral over region \mbox{$\cR_{IV}=J^+(\gamma)\cap J^-(x)$} with \mbox{$x=(u_{\gamma}(v_*), \hat{v})$}, $x \in \cB$, cf.~ Figure \ref{RN_mit_u_v}, of the current $\tilde{K}^{S}$, defined by \eqref{KStilde}, can be estimated by
\ben
\lb{Prop6.1}
\int\limits_{\cR_{IV}} |\tilde{K}^{S}(\psi)| \dV &\leq& \delta_1 \sup_{u_{\gamma}(\hat{v})\leq \bar{u}\leq u_{\gamma}(v_*)}\int\limits_{\left\lbrace  v_{\gamma}(\bar{u}) \leq v \leq \hat{v}\right\rbrace }  J_{\mu}^{S}(\psi) n^{\mu}_{u=\bar{u}}\dV_{u=\bar{u}}\nonumber\\
&&+ \delta_2 \sup_{v_*\leq \bar{v} \leq \hat{v}}\int\limits_{\left\lbrace  u_{\gamma}(\hat{v}) \leq u \leq u_{\gamma}(\bar{v}) \right\rbrace }
J_{\mu}^{S}(\psi) n^{\mu}_{v=\bar{v}} \dV_{v=\bar{v}},
\een
where $\delta_1$ and $\delta_2$ are positive constants, with $\delta_1\rightarrow 0$ and $\delta_2\rightarrow 0$ as $v_*\rightarrow \infty$.
\end{lem}
{\begin{figure}[ht]
\centering
\includegraphics[width=0.45\textwidth]{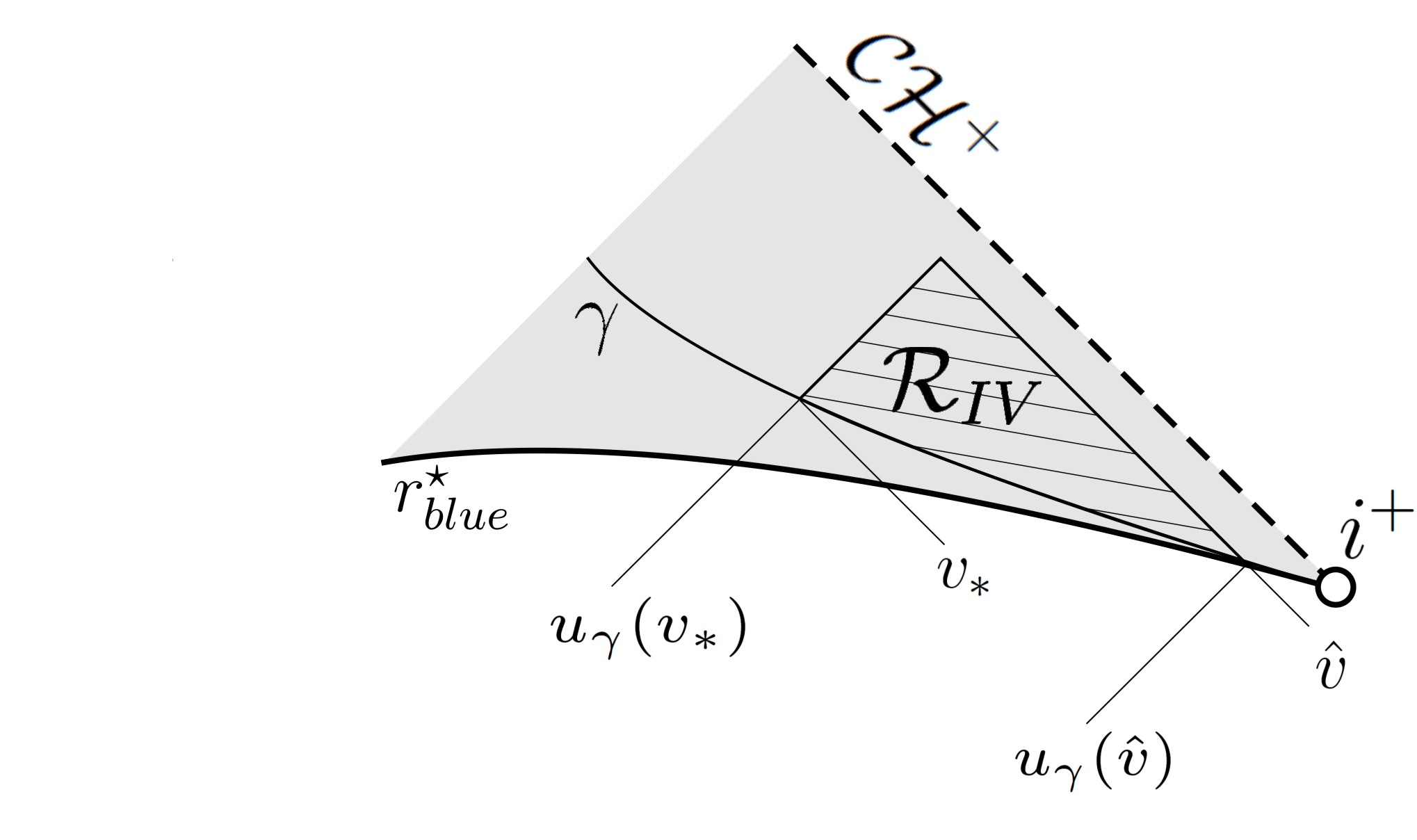}
\caption[]{Blueshift region of the $(u,v)$-plane of Kerr spacetime from the hypersurface $\gamma$ onwards.}
\label{RN_mit_u_v}\end{figure}}
\begin{proof} 
Taking the integral over the spacetime region, using expression \eqref{KStilde} yields
\bea
\lb{knboundhier}
&\int\limits_ {\cR_{IV}}&  |\tilde{K}^{S}(\psi)| \dV \nonumber\\
&\leq& \int\limits^{u_{\gamma}(v_*)}_{u_{\gamma}(\hat{v})}\int\limits_{\left\lbrace  v_{\gamma}(\bar{u}) \leq v \leq \hat{v}\right\rbrace }C\left|\frac{\partial_{{\bar{v}}}( L^2\sin^2\theta)}{L^2 \sin^2 \theta}+\frac{\partial_{{\bar{u}}}( L^2\sin^2\theta)}{L^2 \sin^2 \theta}\frac{|{\bar{u}}|^p}{{\bar{v}}^p}+\sqrt{\partial_{\bar{u}} b^{\tilde{\phi}}}\right|
J_{\mu}^{S}(\psi) n^{\mu}_{u=\bar{u}}\dV_{u=\bar{u}} \md \bar{u}\nonumber \\
&&+ \int\limits_{{v_*}}^{\hat{v}}\int\limits_{\left\lbrace  u_{\gamma}(\bar{v}) \leq u \leq u_{\gamma}(v_*)\right\rbrace } C\left|\frac{\partial_{{\bar{u}}}( L^2\sin^2\theta)}{L^2 \sin^2 \theta}+\frac{\partial_{{\bar{v}}}( L^2\sin^2\theta)}{L^2 \sin^2 \theta}\frac{{\bar{v}}^p}{|{\bar{u}}|^p}+\sqrt{\partial_{\bar{u}} b^{\tilde{\phi}}}\right|
J_{\mu}^{S}(\psi) n^{\mu}_{v=\bar{v}}\dV_{v=\bar{v}}\md \bar{v},\nonumber\\
\eea
with $u_{\gamma}(v)$ and $v_{\gamma}(\bar{u})$ in the integration limits as defined in \eqref{notation_neu}.

Further it follows that
\bea
\lb{131}
&&\int\limits_ {\cR_{IV}}  |\tilde{K}^{S}(\psi)| \dV \nonumber\\
&&\leq C\int\limits^{u_{\gamma}(v_*)}_{u_{\gamma}(\hat{v})} \sup_{v_{\gamma}(\bar{u})\leq \bar{v} \leq \hat{v}}
\left|\frac{\partial_{{\bar{v}}}( L^2\sin^2\theta)}{L^2 \sin^2 \theta}+\frac{\partial_{\bar{u}}( L^2\sin^2\theta)}{L^2 \sin^2 \theta}\frac{|{\bar{u}}|^p}{{\bar{v}}^p}+\sqrt{\partial_{\bar{u}} b^{\tilde{\phi}}}\right|
\md \bar{u} \nonumber\\
&&\qquad  \qquad \times \sup_{u_{\gamma}(\hat{v})\leq \bar{u}\leq u_{\gamma}(v_*)}\int\limits_{\left\lbrace  v_* \leq v \leq \hat{v}\right\rbrace }J_{\mu}^{S}(\psi) n^{\mu}_{u=\bar{u}}\dV_{u=\bar{u}}
\nonumber\\
&&+  C\int\limits^{\hat{v}}_{v_*}\sup_{u_{\gamma}(\bar{v}) \leq \bar{u}\leq u_{\gamma}(v_*)}\left|\frac{\partial_{{\bar{u}}}( L^2\sin^2\theta)}{L^2 \sin^2 \theta}+\frac{\partial_{{\bar{v}}}( L^2\sin^2\theta)}{L^2 \sin^2 \theta}\frac{{\bar{v}}^p}{|{\bar{u}}|^p}+\sqrt{\partial_{\bar{u}} b^{\tilde{\phi}}}\right|
\md \bar{v}\nonumber\\
&&\qquad  \qquad \times \sup_{v_*\leq \bar{v} \leq \hat{v}}\int\limits_{\left\lbrace  u_{\gamma}(\bar{v}) \leq u \leq u_{\gamma}(v_*)\right\rbrace } J_{\mu}^{S}(\psi) n^{\mu}_{v=\bar{v}}\dV_{v=\bar{v}}.
\eea

It remains to show finiteness and smallness of the expressions \mbox{$\int\limits^{u_{\gamma}(v_*)}_{u_{\gamma}(\hat{v})} \sup_{v_{\gamma}(\bar{u})\leq \bar{v} \leq \hat{v}}
\left|\frac{\partial_{{\bar{v}}}( L^2\sin^2\theta)}{L^2 \sin^2 \theta}+\frac{\partial_{{\bar{u}}}( L^2\sin^2\theta)}{L^2 \sin^2 \theta}\frac{|{\bar{u}}|^p}{{\bar{v}}^p}+\sqrt{\partial_{\bar{u}} b^{\tilde{\phi}}}\right|
\md \bar{u}$} and 
\mbox{$\int\limits^{\hat{v}}_{v_*}\sup_{u_{\gamma}(\bar{v}) \leq \bar{u}\leq u_{\gamma}(v_*)}\left|\frac{\partial_{{\bar{u}}}( L^2\sin^2\theta)}{L^2 \sin^2 \theta}+\frac{\partial_{{\bar{v}}}( L^2\sin^2\theta)}{L^2 \sin^2 \theta}\frac{{\bar{v}}^p}{|{\bar{u}}|^p}+\sqrt{\partial_{\bar{u}} b^{\tilde{\phi}}}\right|
\md \bar{v}$}.
Note relation \eqref{spacevolumedecay_u} of Section \ref{finiteness} for region $\cR_{IV}$. Further, with relations \eqref{b_bound} and \eqref{L_bounded} we can write
\bea
 &&\int\limits^{u_{\gamma}(v_*)}_{u_{\gamma}(\hat{v})} \sup_{v_{\gamma}(\bar{u})\leq \bar{v} \leq \hat{v}}
\left|\frac{\partial_{{\bar{v}}}( L^2\sin^2\theta)}{L^2 \sin^2 \theta}+\frac{\partial_{{\bar{u}}}( L^2\sin^2\theta)}{L^2 \sin^2 \theta}\frac{|{\bar{u}}|^p}{{\bar{v}}^p}+\sqrt{\partial_{\bar{u}} b^{\tilde{\phi}}}\right|
\md \bar{u}\nonumber\\
&\leq& C \int\limits_{u_{\gamma}(\hat{v})}^{u_{\gamma}(v_*)}\sup_{v_{\gamma}(\bar{u})\leq \bar{v} \leq \hat{v}}\left[{|u_{\gamma}(\bar{v})|}^{-{\beta \alpha}} e^{{\beta}\left[u_{\gamma}(\bar{v})-\bar{u}\right]}\left( 1 + \frac{|{\bar{u}}|^p}{{\bar{v}}^p}\right)+{|u_{\gamma}(\bar{v})|}^{-\frac{\beta \alpha}{2}} e^{\frac{\beta}{2}\left[u_{\gamma}(\bar{v})-\bar{u}\right]}\right] \md \bar{u}\nonumber\\
&\leq& \tilde{C} \int\limits_{u_{\gamma}(\hat{v})}^{u_{\gamma}(v_*)}{|\bar{u}|}^{-\frac{\beta \alpha}{2}}  \md \bar{u}
\leq \frac{\tilde{{C}}}{|-{\beta \alpha}+p+1|}\left[{|\bar{u}|}^{-\frac{\beta \alpha}{2}+1}\right]_{u_{\gamma}(\hat{v})}^{u_{\gamma}(v_*)}
\leq \delta_1, 
\eea
where $\delta_1\rightarrow 0$ for $|u_{\gamma}(v_*)|\rightarrow -\infty$ and thus for $v_*\rightarrow \infty$. Note that we have here used \eqref{alpha} in the last step.

For finiteness of the second term in \eqref{131} we follow the same strategy and use \eqref{spacevolumedecay} to obtain
\bea
 &&\int\limits^{\hat{v}}_{v_*}\sup_{u_{\gamma}(\bar{v}) \leq \bar{u}\leq u_{\gamma}(v_*)}\left|\frac{\partial_{{\bar{u}}}( L^2\sin^2\theta)}{L^2 \sin^2 \theta}+\frac{\partial_{{\bar{v}}}( L^2\sin^2\theta)}{L^2 \sin^2 \theta}\frac{{\bar{v}}^p}{|{\bar{u}}|^p}+\sqrt{\partial_{\bar{u}} b^{\tilde{\phi}}}\right|
\md \bar{v}\nonumber\\
&\leq& C \int\limits^{\hat{v}}_{v_*}\sup_{u_{\gamma}(\bar{v}) \leq \bar{u}\leq u_{\gamma}(v_*)}
\left[\bar{v}^{-{\beta \alpha}} + \frac{{\bar{v}}^{-{\beta \alpha}+p}}{|{\bar{u}}|^p}+\bar{v}^{-\frac{\beta \alpha}{2}}\right]
\md \bar{v}\nonumber\\
&\leq& \tilde{C}\int\limits^{\hat{v}}_{v_*}\left[\bar{v}^{-{\beta \alpha}} + \frac{{\bar{v}}^{-{\beta \alpha}+p}}{|u_{\gamma}(v_*)|^p}+\bar{v}^{-\frac{\beta \alpha}{2}}\right]
\md \bar{v}\nonumber\\
&\leq& \tilde{\tilde{C}}\left[\frac{{|\bar{v}|}^{-{\beta \alpha}+p+1}}{|-{\beta \alpha}+p+1|}+ \frac{{|\bar{v}|}^{-\frac{\beta \alpha}{2}+1}}{|-\frac{\beta \alpha}{2}+1|}\right]^{\hat{v}}_{v_*}
\leq \delta_2,
\eea
where $\delta_2 \rightarrow 0$ for $v_* \rightarrow \infty$. We demanded \eqref{alpha} of Section \ref{gamma_curve} in order to obtain the last step. Further, we use the argument explained after \eqref{vertausch}, where the bound \eqref{boundedL} was used, allowing us to pull the $(L\sin \theta)$-factor of the volume form inside the integral. From that the statement of Lemma \ref{K_spher} follows. 
\end{proof}
We can now propagate the weighted energy further.
\begin{prop}
\label{kastle}
Let $\psi$ be as in Theorem \ref{anfang} and $p$ as in \eqref{waspist}, and $Y^k$ as in \eqref{angular_comm} with \eqref{Yk},  \eqref{YkK}, \eqref{YkE} and all $k \in \left\{0,1,2\right\}$. Then, for $u_{\schere}$ sufficiently big, for all $v_*\geq v_{\gamma}(u_{\schere})$ and $\hat{v}>v_*$ 
\bea
\lb{propgamma}
\int\limits_ {\left\lbrace u_{\gamma}(\hat{v})\leq u \leq u_{\gamma}(v_*)\right\rbrace }  J_{\mu}^{S}(Y^k\psi) n^{\mu}_{v=\hat{v}} \dV_{v=\hat{v}}
&+&\int\limits_ {\left\lbrace  v_* \leq v \leq \hat{v}\right\rbrace }  J_{\mu}^{S}(Y^k\psi) n^{\mu}_{u=u_{\gamma}(v_*)} \dV_{u=u_{\gamma}(v_*)}\nonumber\\
&& \qquad  \qquad\qquad \leq C {v_*}^{-1-2\delta+p},  
\eea
where $C$ is a positive constant depending on $C_{0}$ of Theorem \ref{anfang} and $D_{0}(u_{\diamond}, 1)$ of Proposition \ref{initialdataprop}, where $u_{\diamond}$ is defined by $r_{red}=r(u_{\diamond},1)$.
\end{prop}
{\em Remark.} Refer to \eqref{notation_neu} for the definition of $u_{\gamma}(v)$ and see Figure \ref{integralbild3} b) and Figure \ref{RN_mit_u_v} for further clarification.
\begin{proof}
The statement for $k=0$ follows, from using the divergence theorem, the result of Proposition \ref{to_gamma0} and Lemma \ref{K_spher}. 
We have also used that $\delta_1\rightarrow 0$, $\delta_2\rightarrow 0$ as $v_*\rightarrow \infty$. Moreover, we need to choose $u_{\schere}$ sufficiently close to $-\infty$, such that for $v_*>v_{\gamma}(u_{\schere})$, say 
\bea
 \delta_1, \delta_2 \leq \frac12
\eea
holds.
For more details see Porposition 4.11 of \cite{anne}.
Further, we need to show boundedness after commutation with $Y$. 
Let us now show, that we can keep definition \eqref{KStilde} for $\tilde{K}^{S}(Y^k\psi)$, which is to say
\bea
\lb{KStildek}
&&|\tilde{K}^S(Y^k\psi)|\nonumber\\
&=&C\left|\frac{\partial_{u}( L^2\sin^2\theta)}{L^2 \sin^2 \theta}+\frac{\partial_{v}( L^2\sin^2\theta)}{L^2 \sin^2 \theta}\frac{v^p}{|u|^p}+\sqrt{\partial_u b^{\tilde{\phi}}}\right|\frac{1}{2\Omega^2}|u|^p(\partial_u Y^k\psi)^2\nonumber\\
&+&C\left|\frac{\partial_{v}( L^2\sin^2\theta)}{L^2 \sin^2 \theta}+\frac{\partial_{u}( L^2\sin^2\theta)}{L^2 \sin^2 \theta}\frac{|u|^p}{v^p}+\sqrt{\partial_u b^{\tilde{\phi}}}\right|\frac{1}{2\Omega^2}v^p(\partial_v Y^k\psi+b^{\tilde{\phi}}\partial_{\tilde{\phi}}Y^k\psi)^2,
\eea
but with C such that 
\bea
\lb{tildeineq1}
-[{K}^{S}(Y\psi) +\cE^{S}(Y\psi)+{K}^{S}(\psi)]  &\leq&
|\tilde{K}^{S}(Y\psi)|+|\tilde{K}^{S}(\psi)|, 
\eea 
and
\bea
\lb{tildeineq2}
-[{K}^{S}(Y^2\psi) +\cE^{S}(Y^2\psi)+{K}^{S}(Y\psi) +\cE^{S}(Y\psi)+{K}^{S}(\psi)]  \nonumber\\
\qquad \leq |\tilde{K}^{S}(Y^2\psi)|+|\tilde{K}^{S}(Y\psi)|+|\tilde{K}^{S}(\psi)|. 
\eea
Now we can state the following lemma.
\begin{lem}
\label{K_spher2}
Let $\psi$ be an arbitrary function. Then, for all $v_*>2\alpha$ 
and all $\hat{v}>v_*$, and $Y^k$ as in \eqref{angular_comm} with \eqref{Yk},  \eqref{YkK}, \eqref{YkE} and all $k \in \left\{0,1,2\right\}$, the integral over region \mbox{$\cR_{IV}=J^+(\gamma)\cap J^-(x)$} with \mbox{$x=(u_{\gamma}(v_*), \hat{v})$}, $x \in \cB$, cf.~ Figure \ref{RN_mit_u_v}, of the current $\tilde{K}^{S}$, defined by \eqref{KStildek}, can be estimated by
\bea
\lb{S_lemma1}
\int\limits_{\cR_{IV}} |\tilde{K}^{S}(Y^k\psi)| 
&\leq& \delta_1 \sup_{u_{\gamma}(\hat{v})\leq \bar{u}\leq u_{\gamma}(v_*)}\int\limits_{\left\lbrace  v_{\gamma}(\bar{u}) \leq v \leq \hat{v}\right\rbrace }  J_{\mu}^{S}(Y^k\psi) n^{\mu}_{u=\bar{u}}\dV_{u=\bar{u}}\nonumber\\
&&+ \delta_2 \sup_{v_*\leq \bar{v} \leq \hat{v}}\int\limits_{\left\lbrace  u_{\gamma}(\hat{v}) \leq u \leq u_{\gamma}(\bar{v}) \right\rbrace }
J_{\mu}^{S}(Y^k\psi) n^{\mu}_{v=\bar{v}} \dV_{v=\bar{v}},\nonumber\\
\eea
where $\delta_1$ and $\delta_2$ are positive constants, with $\delta_1\rightarrow 0$ and $\delta_2\rightarrow 0$ as $v_*\rightarrow \infty$.
\end{lem}
\begin{proof}
We have already shown \eqref{S_lemma1} for $k=0$ in Lemma \ref{K_spher}. 

In order to see, that Lemma \ref{K_spher2} is also true for $k=1$, we first need to consider \eqref{pi-Y-termevv} to \eqref{pi-Y-termeAB} which define the second derivative terms as can be seen from \eqref{comm}. Further, recall that we need to apply \eqref{uv_term} for the $(\partial_u \partial_v \psi)$ term. With this and by recalling that $\partial_{\theta^{\star}} b^{\tilde{\phi}}$ was also bounded by $\Omega^2$, see Section \ref{doublenull}, we can see that all terms involved are bounded. By using the Cauchy--Schwarz inequality it is possible to distribute a weight of $\Omega$ as we have already done in \eqref{vp_cauchy}. Once we have done this, we see that we obtain \eqref{tildeineq1} for big enough $C$. The remaining part of the proof for $k=1$ is then analogous to the proof of Lemma \ref{K_spher}.
We now have to use this result to prove a $k=1$ version of Proposition \ref{kastle} analogously as shown above for $k=0$. This is done by using the divergence theorem for first and second order terms and leads to control over all second order terms (as well as third order terms containing $\tilde{\phi}$). Looking at equation \eqref{comm} again, we now have to bring special attention to the first and second term of the right hand side. Again we see, that we only obtain terms which can be dominated by $\tilde{K}^S(Y^2 \psi)$. This proves statement \eqref{S_lemma1} for $k=2$.
\end{proof}
Hereafter, we again use the divergence theorem up to third order and the above lemma to prove the $k=2$ version of Proposition \ref{kastle}.
\end{proof}

\subsubsection{Higher order weighted energy boundedness up to $\cC\cH^+$ close to $i^+$}
In the previous Sections \ref{red_region} to \ref{bounding_bulk_S} we have proven energy estimates for each region with specific properties separately. Putting all results together we can state the following proposition. 
\begin{prop}
\label{gesamtesti}
Let $\psi$ be as in Theorem \ref{anfang} and $p$ as in \eqref{waspist}, and $Y^k$ as in \eqref{angular_comm} with \eqref{Yk},  \eqref{YkK}, \eqref{YkE} and all $k \in \left\{0,1,2\right\}$.
Then, for $u_{\schere}$ sufficiently close to $-\infty$, for all $v_*>1$, $\hat{v}>v_*$ and $\tilde{u} \in (-\infty, u_{\schere})$.
\bea
\lb{propgamma2}
\int\limits_ {\left\lbrace  v_* \leq v \leq \hat{v}\right\rbrace }  J_{\mu}^{S}(Y^k\psi) n^{\mu}_{u=\tilde{u}} \dV_{u=\tilde{u}}\leq {C} {v_*}^{-1-2\delta+p},
\eea
where ${C}$ is a positive constant depending on $C_0$ of Theorem \ref{anfang} and $D_{0}(u_{\diamond},1)$ of Proposition \ref{initialdataprop}, where $u_{\diamond}$ is defined by $r_{red}=r(u_{\diamond},1)$.
\end{prop}
\begin{proof}
First of all, we partition the integral of the statement into a sum of integrals of the different regions. That is to say
\bea
\int\limits_ {\left\lbrace  v_* \leq v \leq \hat{v}\right\rbrace }  J_{\mu}^{S}(Y^k\psi) n^{\mu}_{u=\tilde{u}} \dV_{u=\tilde{u}}&=&\int\limits_ {{\left\lbrace  v_* \leq v \leq \hat{v}\right\rbrace } \cap \cR} J_{\mu}^{S}(Y^k\psi) n^{\mu}_{u=\tilde{u}} \dV_{u=\tilde{u}}\nonumber\\
&&+\int\limits_ {{\left\lbrace  v_* \leq v \leq \hat{v}\right\rbrace } \cap \cN} J_{\mu}^{S}(Y^k\psi) n^{\mu}_{u=\tilde{u}} \dV_{u=\tilde{u}}\nonumber\\
&&+\int\limits_ {{\left\lbrace  v_* \leq v \leq \hat{v}\right\rbrace } \cap J^-(\gamma)\cap\cB} J_{\mu}^{S}(Y^k\psi) n^{\mu}_{u=\tilde{u}} \dV_{u=\tilde{u}}\nonumber\\
&&+\int\limits_ {{\left\lbrace  v_* \leq v \leq \hat{v}\right\rbrace } \cap J^+(\gamma)\cap\cB} J_{\mu}^{S}(Y^k\psi) n^{\mu}_{u=\tilde{u}} \dV_{u=\tilde{u}}.
\eea
For the integral in $\cR$ and the integral in $\cN$ we use Corollaries \ref{cor5.2} and \ref{cor6.1}. (Note that the former has to be summed and weighted resulting in the loss of $1+p$ polynomial powers.)
Further, for the integral in region $J^-(\gamma)\cap\cB$ we apply Corollary \ref{cor7.1} and for the integral in region $J^+(\gamma)\cap\cB$ we use Proposition \ref{kastle}. 
Putting all this together, we arrive at the conclusion of Proposition \ref{gesamtesti}.
\end{proof}
We can now state a version of Theorem \ref{energythm} inside the characteristic rectangle $\Xi$.
\begin{thm}
\lb{Xithm}
Let $\phi$ be as in Theorem \ref{anfang} and $p$ as in \eqref{waspist}, and $Y^k$ as in \eqref{angular_comm} with \eqref{Yk},  \eqref{YkK}, \eqref{YkE} and all $k \in \left\{0,1,2\right\}$. Then, for $u_{\schere}$ sufficiently close to $-\infty$, for all \mbox{$v_{fix}\geq1$}, and $\tilde{u} \in (-\infty, u_{\schere})$,
\bea
\lb{array1_cor}
&&\int\limits^{\infty}_{v_{fix}} \int\limits_{\bbS^2_{u,v}}\left[v^p (\partial_v (Y^k\psi)+b^{\tilde{\phi}}\partial_{\tilde{\phi}} (Y^k\psi))^2(\tilde{u}, v, \theta^{\star}, \tilde{\phi}) +\Omega^2|\nabb (Y^k\psi)|^2(\tilde{u}, v, \theta^{\star}, \tilde{\phi}) \right]L\md \sigma_{\mathbb S^2}\md v\nonumber\\
&&\qquad \qquad \leq {E_{k}},
\eea
where $C$ is a positive constant dependent on $C_{0}$ of Theorem \ref{anfang} and $D_{0}(u_{\diamond},1)$ of Proposition \ref{initialdataprop}, where $u_{\diamond}$ is defined by $r_{red}=r(u_{\diamond},1)$.
\end{thm}
\begin{proof}
The proof follows from comparing the weights in Proposition \ref{gesamtesti}.\\
\end{proof}


\subsection{Pointwise estimates from higher order energies}
\lb{nineten}
\subsubsection{The Sobolev inequality on spheres}
\lb{sobolevsec}
Recall that we had introduced the vector fields $Y_i$, $i=1,2,3$, in Section \ref{angular}. Further, note that from \eqref{Yk} it follows that
\bea
\lb{leonotation}
 \left(Y^k \psi\right)^2&=&\sum\limits_{i_1=1}^3 \cdot \cdot\cdot \sum\limits_{i_k=1}^3\left(Y_{i_1}\cdot \cdot \cdot (Y_{i_k} \psi)\right)^2, \\
\left(\partial_v (Y^k \psi)\right)^2&=&\sum\limits_{i_1=1}^3 \cdot \cdot\cdot \sum\limits_{i_k=1}^3 \left(\partial_v (Y_{i_1}\cdot \cdot \cdot (Y_{i_k} \psi))\right)^2, \\
\left|\nabb (Y^k \psi)\right|^2&=&\sum\limits_{i_1=1}^3 \cdot \cdot\cdot \sum\limits_{i_k=1}^3 \left|\nabb (Y_{i_1}\cdot \cdot \cdot (Y_{i_k} \psi))\right|^2, \quad
\eea
with $i_j=1,2$ or $3$.
By Sobolev embedding on the standard spheres we have in this notation 
\bea
\lb{sobo_embed}
\sup_{\left\{\theta^{\star},\tilde{\phi}\right\} \in \bbS^2_{u,v}}|\psi(u,v,\theta^{\star},\tilde{\phi})|^2\leq \tilde{C} \sum_{k=0}^{2} \int\limits_{\bbS^2_{u,v}} \left(Y^k \psi\right)^2(u,v,\theta^{\star},\tilde{\phi})L \sin\theta\md \theta^{\star} \md \tilde{\phi}, 
\eea
which means that we can derive a pointwise estimate from an estimate of the integrals on the spheres. In order to obtain these integrals we will need the previously derived higher order weighted energy estimates.

\subsubsection{Pointwise boundedness in the neighbourhood of $i^+$}
\lb{uni_bounded}
We will now discuss the derivation of pointwise boundedness from the weighted energy estimates in the characteristic rectangle $\Xi$, stated in Theorem \ref{dashier}. \\
By the fundamental theorem of calculus it follows for all $v_*>1$, $\hat{v} >v_*$ and $u \in (-\infty, u_{\schere})$ that
\bea
\psi(u,\hat{v}, \theta^{\star}, \tilde{\phi})&=& \int\limits_{v_*}^{\hat{v}} \left(\partial_v \psi\right)(u, v, \theta^{\star}, \tilde{\phi}) \md v +\psi(u, v_*, \theta, \tilde{\phi})\nonumber\\
&\leq& \int\limits_{v_*}^{\hat{v}} (\partial_v \psi)(u, v, \theta^{\star}, \tilde{\phi})v^{\frac{p}{2}}v^{-\frac{p}{2}}\md v+\psi(u, v_*, \theta^{\star}, \tilde{\phi})\nonumber\\
&\leq& \left(\int\limits_{v_*}^{\hat{v}} v^p(\partial_v \psi)^2(u, v, \theta^{\star}, \tilde{\phi})\md v\right)^{\frac12}\left(\int\limits_{v_*}^{\hat{v}} v^{-{p}}\md v\right)^{\frac12}+\psi(u, v_*, \theta^{\star}, \tilde{\phi}),
\eea
where we have used the Cauchy--Schwarz inequality in the last step.

Squaring the entire expression, using Cauchy--Schwarz again and integrating over ${\mathbb S}^2_{u,v}$ we obtain 
\bea
\int\limits_{\bbS^2_{u,v}} \psi^2(u,\hat{v})L\md \sigma_{\mathbb S^2}
&\leq& \tilde{C}\left[\int\limits_{v_*}^{\hat{v}} \int\limits_{\bbS^2_{u,v}} v^p(\partial_v \psi)^2(u,v)L(u,v)\md \sigma_{\mathbb S^2}\md v\int\limits_{v_*}^{\hat{v}} v^{-{p}}\md v
\right.\nonumber\\
&& \left. \qquad \qquad +\int\limits_{\bbS^2_{u,v}} \psi^2(u,v_*)L\md \sigma_{\mathbb S^2}\right],
\eea
with $p$ as in \eqref{waspist} and $L\sin \theta$ pulled into the integral by boundedness \eqref{boundedL}, such that the first term on the right hand side is controlled by the flux that we have derived in Theorem \ref{Xithm} for $k=0$.
Further, in the notation of Section \ref{sobolevsec}, we can state
\bea
\lb{fundcauchy}
\int\limits_{\bbS^2_{u,v}} (Y^k \psi) ^2(u,\hat{v})\md \sigma_{\mathbb S^2_{u,v}}&\leq&\tilde{C}\left[ E_{k}\int\limits_{v_*}^{\hat{v}} \int\limits_{\bbS^2_{u,v}} v^{-p}\md \sigma_{\mathbb S^2_{u,v}}\md v +\int\limits_{\bbS^2_{u,v}} (Y^k \psi) ^2(u,v_*)\md \sigma_{\mathbb S^2_{u,v}}\right]\nonumber\\
&\leq&\tilde{C}\left[\tilde{\tilde{C}} E_{k}+\int\limits_{\bbS^2_{u,v}} (Y^k \psi) ^2(u,v_*)\md \sigma_{\mathbb S^2_{u,v}}\right],
\eea 
for all $k \in \left\{0,1,2\right\}$. As before, the right hand side of \eqref{fundcauchy} with $v_*=1$ is estimated by Theorem \ref{Xithm}.
Adding all equations up, we derive pointwise boundedness according to \eqref{sobo_embed}
\bea
\lb{supr}
\sup_{\left\{\theta^{\star},\tilde{\phi}\right\}\in\bbS^2_{u,v}}|\psi(u,\hat{v},\theta^{\star},\tilde{\phi})|^2&\leq& \tilde{C} \left[\int\limits_{\bbS^2_{u,v}} ( \psi)^2(u,\hat{v})\md \sigma_{\mathbb S^2_{u,v}} +\int\limits_{\bbS^2_{u,v}} (Y \psi) ^2(u,\hat{v})\md \sigma_{\mathbb S^2_{u,v}}\right.\nonumber\\
&& \left. \qquad \qquad+\int\limits_{\bbS^2_{u,v}} (Y^2 \psi) ^2(u,\hat{v})\md \sigma_{\mathbb S^2_{u,v}} \right],\nonumber\\
\lb{supr2}
&\leq&\tilde{C}\left[\tilde{\tilde{C}}\left( E_{0}+  E_{1}+ E_{2}\right)+D_{0}(u_{\diamond}, 1)+D_{1}(u_{\diamond}, 1)+D_{2}(u_{\diamond}, 1)\right] \nonumber\\
&\leq& C,
\eea
with $C$ depending on the initial data on $\Sigma$. 
We therefore arrive at the statement given in Theorem \ref{dashier}. 
\begin{trivlist}
\item[\hskip \labelsep ]\qed\end{trivlist}

\subsection{Energy along the future boundaries of $\cR_{V}$} 
\lb{region5_proof}
For the {\it right} side of the two-ended spacetime, it remains to show boundedness in regions $\cR_{V}$ and $\cR_{VI}$, see Figure \ref{alle_bnr} a).
Let $u_{\diamond}>u_{\schere}$ and $v_*\geq v_{\gamma}(u_{\schere})$. Define \mbox{$\cR_{V}=\left\{u_{\schere}\leq u \leq u_{\diamond}\right\}\cap\left\{ v_*\leq v\leq \hat{v}\right\}$}, cf.~ Figure \ref{region5}, and note that \mbox{$\cR_{V} \subset \cB$}. 
{\begin{figure}[ht]
\centering
\includegraphics[width=0.25\textwidth]{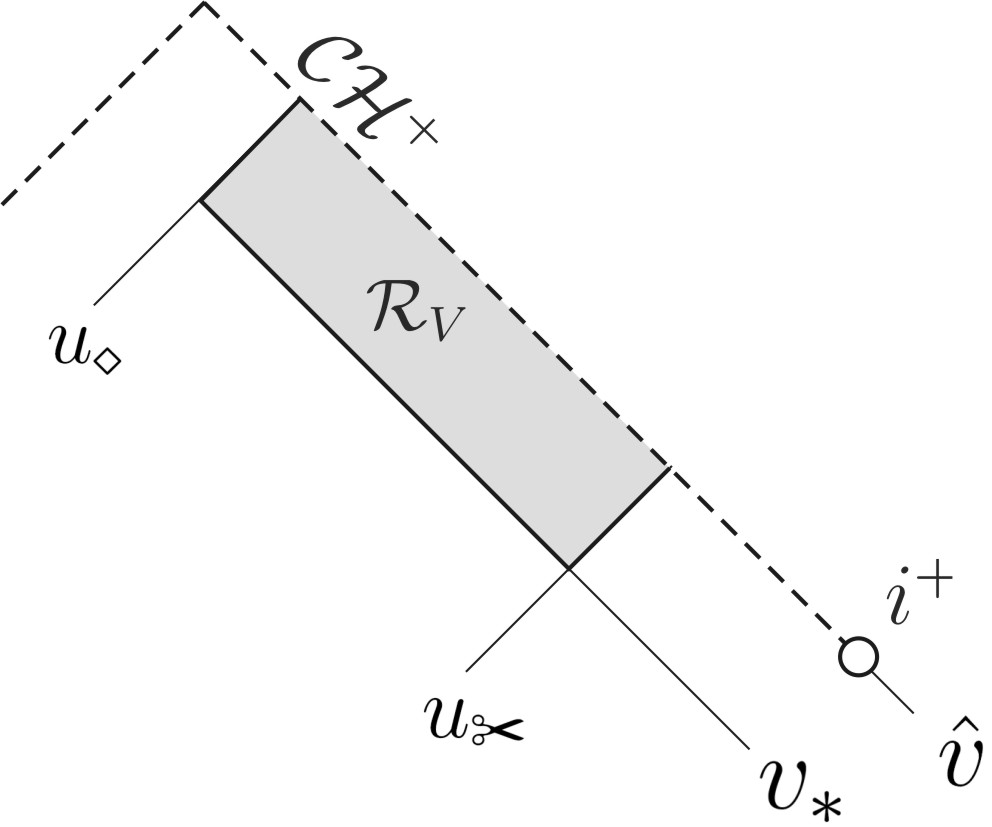}
\caption[]{$(u,v)$-diagram depicting region $\cR_{V}$.}
\label{region5}\end{figure}}
We will apply the vector field
\bea
\lb{V_feld}
W=v^p (\partial_v+b^{\tilde{\phi}}\partial_{\tilde{\phi}})+\partial_u
\eea
as a multiplier.
The bulk can be calculated as
\bea
\lb{KW}
K^W(\psi)&=& 
-\left[\frac{p}2v^{p-1} +\frac{\partial_{u} \Omega}{\Omega} +v^p\frac{\partial_{v} \Omega}{\Omega} +\frac14\frac{\partial_{u}\left( L^2\sin^2\theta\right)}{L^2\sin^2\theta}+v^p\frac14\frac{\partial_{v}\left( L^2\sin^2\theta\right)}{L^2\sin^2\theta}\right]|\nabb \psi|^2\nonumber\\
&&+\frac14\left[ \frac{\partial_{u}\left( L^2\sin^2\theta\right)}{L^2\sin^2\theta}+v^p\frac{\partial_{v}\left( L^2\sin^2\theta\right)}{L^2\sin^2\theta}\right]\frac{1}{\Omega^2}(\partial_u \psi)(\partial_v \psi+b^{\tilde{\phi}}\partial_{\tilde{\phi}}\psi)\nonumber\\
&&+\left[-\frac{pv^{p-1}b^{\tilde{\phi}}}{2\Omega^2}-\frac{b^{\tilde{\phi}}}{\Omega^2}\left[\frac{\partial_u \Omega}{\Omega}+\frac{\partial_v \Omega}{\Omega}v^p\right]+\frac{\partial_u b^{\tilde{\phi}}}{2\Omega^2}\right](\partial_u \psi\partial_{\tilde{\phi}} \psi)\nonumber\\
&&+\left[-\frac{v^p\partial_ub^{\tilde{\phi}}}{2\Omega^2}\right](\partial_v \psi\partial_{\tilde{\phi}} \psi)\nonumber\\
&&+\left[-\frac{v^p}{4\Omega^2}\left[b^{\theta_A}\partial_ub^{\theta_B}+b^{\theta_B}\partial_ub^{\theta_A}\right]+\frac12(\gin^{-1})^{\theta_A\theta_C}(\gin^{-1})^{\theta_B\theta_D}\partial{\eta}\gin_{CD}X^{\eta} \right]\nonumber\\
&& \quad \times(\partial_{\theta_A} \psi\partial_{\theta_B} \psi),
\eea
and for the $J$-currents we obtain
\bea
\lb{JplugnvS}
J^W_{\mu}n^{\mu}_{v=const}&=&\frac1{2\Omega^2}\left[(\partial_u \psi)^2+ \Omega^2 v^p |\nabb \psi|^2\right]\\
\lb{JplugnuS}
J^W_{\mu}n^{\mu}_{u=const}
&=&\frac1{2\Omega^2}\left[v^p(\partial_v\psi+b^{\tilde{\phi}}\partial_{{\tilde{\phi}}}\psi)^2
+{\Omega^2}|\nabb \psi|^2\right].
\eea
Again, note that the terms in \eqref{KW} multiplying $|\nabb \psi|^2$, are all positive. Moreover, the dominating term \mbox{$v^p\frac{\partial_{v} \Omega}{\Omega}$} is big enough, so that factors of the last three rows of \eqref{KW} can be absorbed by using the Cauchy--Schwarz inequality. 
We define
\bea
\lb{KWtilde}
|\tilde{K}^W(\psi)|&=&\quad C\left|\frac{\partial_{u}( L^2\sin^2\theta)}{L^2 \sin^2 \theta}+\frac{\partial_{v}( L^2\sin^2\theta)}{L^2 \sin^2 \theta}{v^p}+\sqrt{\partial_u b^{\tilde{\phi}}}\right|\frac{1}{2\Omega^2}(\partial_u \psi)^2\nonumber\\
&&+C\left|\frac{\partial_{v}( L^2\sin^2\theta)}{L^2 \sin^2 \theta}+\frac{\partial_{u}( L^2\sin^2\theta)}{L^2 \sin^2 \theta}v^{-p}+\sqrt{\partial_u b^{\tilde{\phi}}}\right|\nonumber\\
&&\qquad \times\frac{1}{2\Omega^2}v^p(\partial_v \psi+b^{\tilde{\phi}}\partial_{\tilde{\phi}}\psi)^2.
\eea
Recall
relation \eqref{vp_cauchy} to see that we obtain
\bea
\lb{167}
-{K}^W(\psi)\leq |\tilde{K}^W(\psi)|,
\eea
for a suitable $C$.
We can now prove the following lemma.
\begin{lem}
\label{K_spherV}
Let $\psi$ be an arbitrary function. Then, 
for all \mbox{$v_*\geq v_{\gamma}(u_{\schere})$}, 
$\hat{v} >v_*$, 
for \mbox{$u_{\diamond}\geq u_2>u_1\geq u_{\schere}$ and $\epsilon\geq u_2-u_1>0$}
\bea
\lb{Propdelta}
\int\limits_{\cR_{V_1}} |\tilde{K}^{W}(\psi)| \dV &\leq& \delta_1 \sup_{u_{1}\leq \bar{u}\leq u_{2}}\int\limits_ {\left\lbrace  v_*\leq v \leq \hat{v}\right\rbrace } J_{\mu}^{W}(\psi) n^{\mu}_{u=\bar{u}}\dV_{u=\bar{u}}\nonumber \\
&&+ \delta_2 \sup_{v_*\leq \bar{v} \leq \hat{v}}\int\limits_ {\left\lbrace  u_1\leq u \leq u_2\right\rbrace } J_{\mu}^{W}(\psi) n^{\mu}_{v=\bar{v}} \dV_{v=\bar{v}},
\eea
where $\cR_{V_1}=\left\{u_1\leq {u} \leq u_2\right\}\cap \cR_{V}$ and $\delta_1$, $\delta_2$ 
are positive constants, depending only on $v_*$ and $\epsilon$ such that $\delta_1\rightarrow 0$ for $\epsilon\rightarrow 0$ and $\delta_2\rightarrow 0$ as ${v}_*\rightarrow \infty$.
\end{lem}
\begin{proof}
Taking the integral over the spacetime region, using expression \eqref{KWtilde}, yields
\bea
\lb{knboundhierW}
\int\limits_ {\cR_{V}}  |\tilde{K}^{W}(\psi)| \dV 
&\leq& \quad\int\limits^{u_{\gamma}(v_*)}_{u_{\gamma}(\hat{v})}\int\limits_{\left\lbrace  v_{\gamma}(\bar{u}) \leq v \leq \hat{v}\right\rbrace }C\left|\frac{\partial_{{\bar{v}}}( L^2\sin^2\theta)}{L^2 \sin^2 \theta}+\frac{\partial_{{\bar{u}}}( L^2\sin^2\theta)}{L^2 \sin^2 \theta}{\bar{v}}^{-p}+\sqrt{\partial_{\bar{u}} b^{\tilde{\phi}}}\right|\nonumber \\
&& \qquad \qquad \qquad \qquad \qquad\times J_{\mu}^{W}(\psi) n^{\mu}_{u=\bar{u}}\dV_{u=\bar{u}} \md \bar{u}\nonumber \\
&&+ \int\limits_{{v_*}}^{\hat{v}}\int\limits_{\left\lbrace  u_{\gamma}(\bar{v}) \leq u \leq u_{\gamma}(v_*)\right\rbrace } C\left|\frac{\partial_{{\bar{u}}}( L^2\sin^2\theta)}{L^2 \sin^2 \theta}+\frac{\partial_{{\bar{v}}}( L^2\sin^2\theta)}{L^2 \sin^2 \theta}{{\bar{v}}^p}+\sqrt{\partial_{\bar{u}} b^{\tilde{\phi}}}\right|\nonumber \\
&& \qquad \qquad \qquad \qquad \qquad\times
J_{\mu}^{W}(\psi) n^{\mu}_{v=\bar{v}}\dV_{v=\bar{v}}\md \bar{v},
\eea
with $u_{\gamma}(v)$ and $v_{\gamma}(\bar{u})$ in the integration limits as defined in \eqref{notation_neu}.

Further, it follows that
\bea
\lb{131W}
\int\limits_ {\cR_{V}}  |\tilde{K}^{W}(\psi)| \dV&\leq& C\int\limits^{u_{\gamma}(v_*)}_{u_{\gamma}(\hat{v})} \sup_{v_{\gamma}(\bar{u})\leq \bar{v} \leq \hat{v}}
\left|\frac{\partial_{{\bar{v}}}( L^2\sin^2\theta)}{L^2 \sin^2 \theta}+\frac{\partial_{{\bar{u}}}( L^2\sin^2\theta)}{L^2 \sin^2 \theta}{\bar{v}}^{-p}+\sqrt{\partial_{\bar{u}} b^{\tilde{\phi}}}\right|
\md \bar{u} \nonumber\\
&&\qquad  \qquad \times \sup_{u_{\gamma}(\hat{v})\leq \bar{u}\leq u_{\gamma}(v_*)}\int\limits_{\left\lbrace  v_* \leq v \leq \hat{v}\right\rbrace }J_{\mu}^{S}(\psi) n^{\mu}_{u=\bar{u}}\dV_{u=\bar{u}}
\nonumber\\
&&+  C\int\limits^{\hat{v}}_{v_*}\sup_{u_{\gamma}(\bar{v}) \leq \bar{u}\leq u_{\gamma}(v_*)}\left|\frac{\partial_{{\bar{u}}}( L^2\sin^2\theta)}{L^2 \sin^2 \theta}+\frac{\partial_{{\bar{v}}}( L^2\sin^2\theta)}{L^2 \sin^2 \theta}{{\bar{v}}^p}+\sqrt{\partial_{\bar{u}} b^{\tilde{\phi}}}\right|
\md \bar{v}\nonumber\\
&&\qquad  \qquad \times \sup_{v_*\leq \bar{v} \leq \hat{v}}\int\limits_{\left\lbrace  u_{\gamma}(\bar{v}) \leq u \leq u_{\gamma}(v_*)\right\rbrace } J_{\mu}^{S}(\psi) n^{\mu}_{v=\bar{v}}\dV_{v=\bar{v}}.\nonumber\\
&&
\eea
It remains to show finiteness and smallness of the expressions \mbox{$\int\limits^{u_{\gamma}(v_*)}_{u_{\gamma}(\hat{v})} \sup_{v_{\gamma}(\bar{u})\leq \bar{v} \leq \hat{v}}
\left|\frac{\partial_{{\bar{v}}}( L^2\sin^2\theta)}{L^2 \sin^2 \theta}+\frac{\partial_{\bar{u}}( L^2\sin^2\theta)}{L^2 \sin^2 \theta}{\bar{v}}^{-p}+\sqrt{\partial_{\bar{u}} b^{\tilde{\phi}}}\right|
\md \bar{u}$} and 
\mbox{$\int\limits^{\hat{v}}_{v_*}\sup_{u_{\gamma}(\bar{v}) \leq \bar{u}\leq u_{\gamma}(v_*)}\left|\frac{\partial_{{\bar{u}}}( L^2\sin^2\theta)}{L^2 \sin^2 \theta}+\frac{\partial_{\bar{v}}( L^2\sin^2\theta)}{L^2 \sin^2 \theta}{{\bar{v}}^p}+\sqrt{\partial_{\bar{u}} b^{\tilde{\phi}}}\right|
\md \bar{v}$}.
We can now use the relations \eqref{b_bound} and \eqref{L_bounded} in region $\cR_{V}$. 
As we have already explained in Section \ref{bounding_bulk_S}, the term $\sqrt{\partial_u b^{\tilde{\phi}}}$ is the slower decaying term, which dictated the condition \eqref{alpha}.
Recall the properties of the hypersurface $\gamma$ shown in Section \ref{gamma_curve}. Since $v_*> v_{\gamma}(u_{\schere})$, \eqref{spacevolumedecay_u} implies that
\bea
\lb{zukunftu}
\Omega^2(\bar{u}, \bar{v}, \theta^{\star})\leq C \Omega^2(u_{\schere}, v_*, \theta^{\star}), \quad \mbox{for any $(\bar{u}, \bar{v}) \in J^+(x)$, with \mbox{$x=(u_{\schere}, v_*)$}, $x \in \cB$,}
\eea
so that we obtain
\bea
&&\int\limits^{u_{\gamma}(v_*)}_{u_{\gamma}(\hat{v})} \sup_{v_{\gamma}(\bar{u})\leq \bar{v} \leq \hat{v}}
\left|\frac{\partial_{{\bar{v}}}( L^2\sin^2\theta)}{L^2 \sin^2 \theta}+\frac{\partial_{\bar{u}}( L^2\sin^2\theta)}{L^2 \sin^2 \theta}{\bar{v}}^{-p}+\sqrt{\partial_{\bar{u}} b^{\tilde{\phi}}}\right|
\md \bar{u}\nonumber\\
&\stackrel{\eqref{zukunftu}}{\leq}& 
C \int\limits^{u_{2}}_{u_1} \sup_{v_*\leq \bar{v} \leq \hat{v}}
\left[\Omega(u_{\schere}, v_*, \theta^{\star})^2(1+ {{\bar{v}}^{-p}})+|\Omega(u_{\schere}, v_*, \theta^{\star})|\right]\md \bar{u}\nonumber\\
&\leq& \tilde{C} \int\limits^{u_{2}}_{u_1}
\left[|u_{\schere}|^{-{\beta\alpha}}+ |u_{\schere}|^{-{\beta\alpha}}{{v_*}^{-p}}+|u_{\schere}|^{-\frac{\beta\alpha}{2}}\right]\md \bar{u}\nonumber\\
&\leq& \tilde{\tilde{C}}\left|{u_{2}}-{u_1}\right| 
\leq \delta_1,
\eea
and moreover $\delta_1 \rightarrow 0$ for $\epsilon \rightarrow 0$.

Further, in Section \ref{finiteness} we derived that similarly
\bea
\lb{zukunft}
\Omega^2(\bar{u}, \bar{v}, \theta^{\star})\leq C \bar{v}^{-\beta\alpha}, \quad \mbox{for any $(\bar{u}, \bar{v}) \in J^+(x)$, with \mbox{$x=(u_{\schere}, v_*)$}, $x \in \cB$,}
\eea
where $v_*> v_{\gamma}(u_{\schere})$.
\bea
&&\int\limits^{\hat{v}}_{v_*}\sup_{u_{\gamma}(\bar{v}) \leq \bar{u}\leq u_{\gamma}(v_*)}\left|\frac{\partial_{{\bar{u}}}( L^2\sin^2\theta)}{L^2 \sin^2 \theta}+\frac{\partial_{\bar{v}}( L^2\sin^2\theta)}{L^2 \sin^2 \theta}{{\bar{v}}^p}+\sqrt{\partial_{\bar{u}} b^{\tilde{\phi}}}\right|
\md \bar{v}\nonumber\\
&\stackrel{\eqref{zukunft}}{\leq}& 
\tilde{C}\int\limits^{\hat{v}}_{v_*}\left[\bar{v}^{-{\beta \alpha}} + {{\bar{v}}^{-{\beta \alpha}+p}}+\bar{v}^{-\frac{\beta \alpha}{2}}\right]
\md \bar{v}\nonumber\\
&\leq& \tilde{\tilde{C}}\left[\frac{{|\bar{v}|}^{-{\beta \alpha}+p+1}}{|-{\beta \alpha}+p+1|}+ \frac{{|\bar{v}|}^{-\frac{\beta \alpha}{2}+1}}{|-\frac{\beta \alpha}{2}+1|}\right]^{\hat{v}}_{v_*}
\leq \delta_2,
\eea
where $\delta_2 \rightarrow 0$ for $v_*\rightarrow \infty$. Thus the conclusion of Lemma \ref{K_spherV} is obtained.
\end{proof}

With the above lemma we can prove the following proposition.
\begin{prop}
\lb{sequenceregion}
Let $\psi$ be as in Theorem \ref{anfang} and $p$ as in \eqref{waspist}, and $Y^k$ as in \eqref{angular_comm} with \eqref{Yk},  \eqref{YkK}, \eqref{YkE} and all $k \in \left\{0,1,2\right\}$. For all \mbox{$v_*>v_{\gamma}(u_{\schere})$} sufficiently large, 
\mbox{$\hat{v} \in (v_*, \infty)$}, 
for \mbox{$u_{\diamond}\geq u_2>u_1\geq u_{\schere}$} and $\epsilon\geq u_2-u_1>0$. Then for $\epsilon$ sufficiently small, the
following is true. If 
\bea
\lb{istart2}
\int\limits_ {\left\lbrace  v_* \leq v \leq \hat{v}\right\rbrace } J_{\mu}^{W}(Y^k\psi) n^{\mu}_{u=u_{1}} \dV_{u=u_{1}}&\leq& {{\tilde{C}}_1},
\eea
then
\bea
\lb{idiamond}
&\int\limits_ {\left\lbrace  v_* \leq v \leq \hat{v}\right\rbrace }& J_{\mu}^{W}(Y^k\psi) n^{\mu}_{u=u_{2}} \dV_{u_{\diamond}}
+\int\limits_ {\left\lbrace  u_{1} \leq u \leq u_{2}\right\rbrace } J_{\mu}^{W}(Y^k\psi) n^{\mu}_{v=\hat{v}} \dV_{v=\hat{v}} \nonumber\\
&\leq& {{\tilde{C}}_{2}}(\tilde{C}_1, u_\diamond, v_*), 
\eea
where ${{\tilde{C}}_{2}}$ depends on $\tilde{C}_1$, ${C_{0}}$ of Theorem \ref{anfang} and ${D_{0}(u_{\diamond}, v_*)}$ of Proposition \ref{initialdataprop}. 
\end{prop}
{\em Remark.} Note already, that the hypothesis \eqref{istart2} is implied by the conclusion of Proposition \ref{gesamtesti} for $u_1=u_{\schere}$. 
\begin{proof}
The statement for $k=0$ follows by the divergence theorem, Theorem \ref{Xithm}, equation \eqref{167} and Lemma \ref{K_spherV}. For more details see the proof of Proposition 4.16 in \cite{anne}. Further, for higher order derivatives we require the following definition:
\bea
\lb{KWtildek}
|\tilde{K}^W(Y^k\psi)|&=&C\left|\frac{\partial_{u}( L^2\sin^2\theta)}{L^2 \sin^2 \theta}+\frac{\partial_{v}( L^2\sin^2\theta)}{L^2 \sin^2 \theta}{v^p}+\sqrt{\partial_u b^{\tilde{\phi}}}\right|\frac{1}{2\Omega^2}(\partial_u Y^k\psi)^2\nonumber\\
&&+C\left|\frac{\partial_{v}( L^2\sin^2\theta)}{L^2 \sin^2 \theta}+\frac{\partial_{u}( L^2\sin^2\theta)}{L^2 \sin^2 \theta}v^{-p}+\sqrt{\partial_u b^{\tilde{\phi}}}\right|\nonumber\\
&& \quad \times\frac{1}{2\Omega^2}v^p(\partial_v Y^k\psi+b^{\tilde{\phi}}\partial_{\tilde{\phi}}Y^k\psi)^2,
\eea
but with C such that 
\bea
\lb{tildeineqW1}
-[{K}^{W}(Y\psi) +\cE^{W}(Y\psi)+{K}^{W}(\psi)]  &\leq&
|\tilde{K}^{W}(Y\psi)|+|\tilde{K}^{W}(\psi)|, 
\eea 
and
\bea
\lb{tildeineqW2}
-[{K}^{W}(Y^2\psi) +\cE^{W}(Y^2\psi)+{K}^{W}(Y\psi) +\cE^{W}(Y\psi)+{K}^{W}(\psi)]  \nonumber\\
\qquad \leq |\tilde{K}^{W}(Y^2\psi)|+|\tilde{K}^{W}(Y\psi)|+|\tilde{K}^{W}(\psi)|. 
\eea
Now we can state the following 
lemma.
\begin{lem}
\label{K_spherV2}
Let $\psi$ be an arbitrary function, and $Y^k$ as in \eqref{angular_comm} with \eqref{Yk},  \eqref{YkK}, \eqref{YkE} and all $k \in \left\{0,1,2\right\}$. Then, 
for all \mbox{$v_*\geq v_{\gamma}(u_{\schere})$}, 
$\hat{v} >v_*$, 
for $u_{\diamond}\geq u_2>u_1\geq u_{\schere}$ and $\epsilon\geq u_2-u_1>0$
\bea
\lb{Propdelta2}
\int\limits_{\cR_{V_1}} |\tilde{K}^{W}(Y^k\psi)| \dV 
&\leq& \delta_1 \sup_{u_{1}\leq \bar{u}\leq u_{2}}\int\limits_ {\left\lbrace  v_*\leq v \leq \hat{v}\right\rbrace } J_{\mu}^{W}(Y^k\psi) n^{\mu}_{u=\bar{u}}\dV_{u=\bar{u}}\nonumber \\
&&+ \delta_2 \sup_{v_*\leq \bar{v} \leq \hat{v}}\int\limits_ {\left\lbrace  u_1\leq u \leq u_2\right\rbrace } J_{\mu}^{W}(Y^k\psi) n^{\mu}_{v=\bar{v}} \dV_{v=\bar{v}},
\eea
where $\cR_{V_1}=\left\{u_1\leq {u} \leq u_2\right\}\cap \cR_{V}$ and $\delta_1$, $\delta_2$ 
are positive constants, depending only on $v_*$ and $\epsilon$ such that $\delta_1\rightarrow 0$ for $\epsilon\rightarrow 0$ and $\delta_2\rightarrow 0$ as ${v}_*\rightarrow \infty$.
\end{lem}
\begin{proof}
We have proven statement \eqref{Propdelta2} for $k=0$ in Lemma \ref{K_spherV}. Like in the proof of Lemma \ref{K_spher2}, we can see, that by using the definition of $\tilde{K}^{W}(Y^k\psi)$, given in equation \eqref{KWtildek}, we can also compensate for all appearing error terms if $C$ is chosen big enough.
This is done by considering \eqref{pi-Y-termevv} to \eqref{pi-Y-termeAB} for the second derivative terms and applying \eqref{uv_term} for the $(\partial_u \partial_v \psi)$-term. By using the Cauchy--Schwarz inequality for all bounded terms, it is possible to distribute a weight of $\Omega$ as we have already done in \eqref{vp_cauchy}. By using this strategy and by adding up the first and second order terms, we obtain
exactly \eqref{tildeineqW1},
and the remaining part of the proof for $k=1$ is analogous to the proof of Lemma \ref{K_spherV}. 
We have therefore shown \eqref{Propdelta2} for $k=1$. We now use this result in the divergence theorem to prove a $k=1$ version of Proposition \ref{sequenceregion} analogously as shown above for $k=0$. This gives us control over all second order terms, as well as third order terms containing $\tilde{\phi}$. Looking at equation \eqref{comm} again, we now have to bring special attention to the first and second term of the right hand side. Again we see, that we only obtain terms which can be dominated by $\tilde{K}^W(Y^2 \psi)$. This proves statement \eqref{Propdelta2} for $k=2$.
\end{proof}
Using the divergence theorem again as well as Lemma \ref{K_spherV2}, finally proves Proposition \ref{sequenceregion} for $k \in \left\{0,1,2\right\}$.
\end{proof}
We are now ready to make a statement for the entire region $\cR_V$. 
\begin{prop}
\lb{rechtes}
Let $\psi$ be as in Theorem \ref{anfang} and $p$ as in \eqref{waspist}, and $Y^k$ as in \eqref{angular_comm} with \eqref{Yk},  \eqref{YkK}, \eqref{YkE} and all $k \in \left\{0,1,2\right\}$.
Then,
for all \mbox{$v_*> v_{\gamma}(u_{\schere})$} sufficiently large, 
$\hat{v}>v_*$, 
and \mbox{$u_{\diamond}>\hat{u}>u_{\schere}$},
\bea
\lb{diamond}
\int\limits_ {\left\lbrace  u_{\schere} \leq u \leq u_{\diamond}\right\rbrace } J_{\mu}^{W}(Y^k\psi) n^{\mu}_{v=\hat{v}} \dV_{v=\hat{v}}
&+&\int\limits_ {\left\lbrace  v_* \leq v \leq \hat{v}\right\rbrace } J_{\mu}^{W}(Y^k\psi) n^{\mu}_{u=\hat{u}} \dV_{u=\hat{u}}\nonumber\\
&& \qquad \qquad \qquad \qquad\leq {C(u_{\diamond}, v_*) }, 
\eea
where $C$ depends on ${C_{0}}$ of Theorem \ref{anfang} and ${D_{0}(u_{\diamond}, v_*)}$ of Proposition \ref{initialdataprop}. 
\end{prop}
\begin{proof}
Let $\epsilon$ be as in Proposition \ref{sequenceregion}.
We choose a sequence \mbox{$u_{i+1}-u_i\leq\epsilon$} and $i=\left\{1,2,..,n\right\}$ such that $u_1=u_{\schere}$ and $u_n=\hat{u}$.
Denote \mbox{$\cR_{V_i}=\left\{u_{i}\leq u \leq u_{i+1}\right\}\cap\left\{ v_*\leq v\leq \hat{v}\right\}$}, cf.~ Figure \ref{region5_sequence}.
Iterating the conclusion of Proposition \ref{sequenceregion} from $u_1$ up to $u_n$ then completes the proof.
Note that $n$ depends only on the smallness condition on $\epsilon$ from Proposition \ref{sequenceregion},
since $n\lesssim  \frac{u_{\diamond}-u_{\schere}}{\epsilon}$.
{\begin{figure}[ht]
\centering
\includegraphics[width=0.2\textwidth]{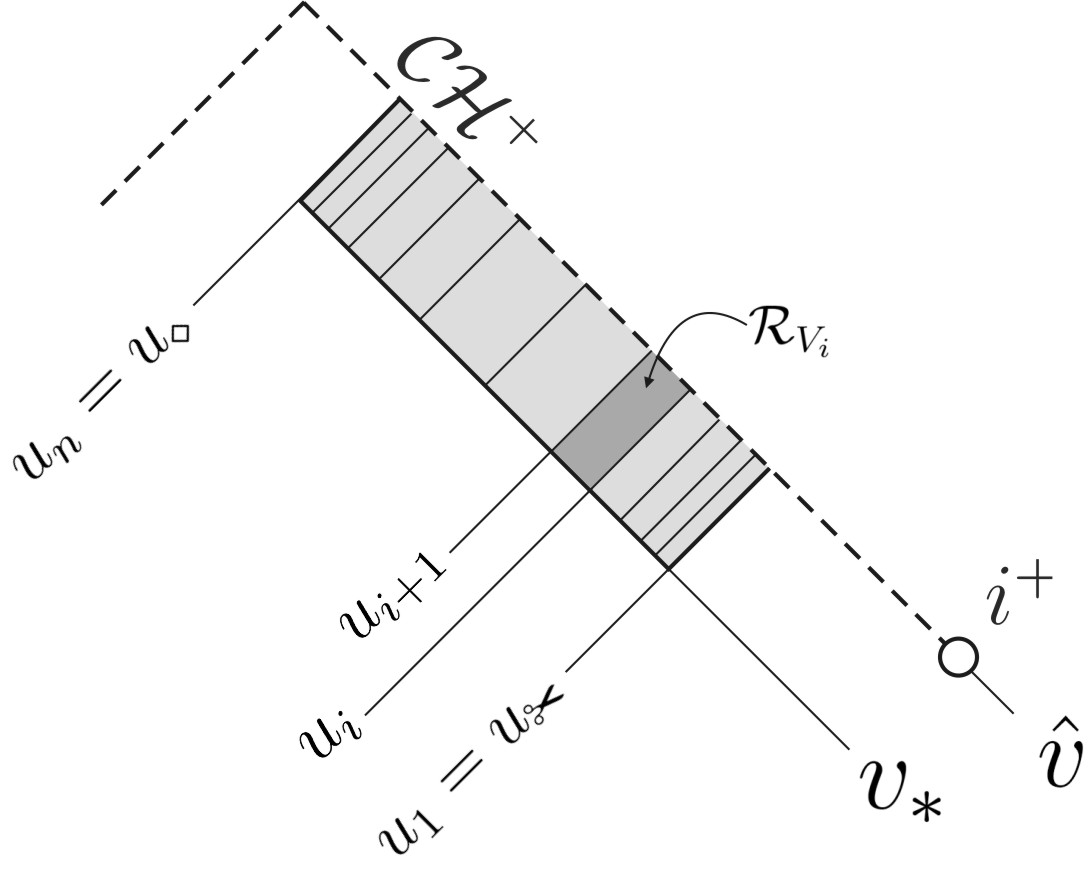}
\caption[]{Penrose diagram depicting regions $\cR_{V_i}$.}
\label{region5_sequence}\end{figure}}
\end{proof}

\subsection{Propagating the energy estimate up to the bifurcation sphere}
\lb{bifurcate}
All results stated so far, correspond to the neighbourhood of $i^+$ at $u=-\infty$, which we call the {\it right} side of the spacetime. Analogously, we can derive all statements for the {\it left} side, in the neighbourhood of $i^+$ at $v=-\infty$. In order to do this, we substitute\footnote{Doing this substitution we keep the expression of the metric \eqref{kerrmetric} the same.} the parameters
\bea
\lb{uv}
u &\leftrightarrow& v,
\eea
and the vector fields
\bea
\lb{paruv}
\partial_u &\leftrightarrow& \partial_v+b^{\tilde {\phi}}\partial_{\tilde {\phi}},
\eea
and repeat all derivations, starting from {\it left} side analog statements of Theorem \ref{anfang} and Proposition \ref{initialdataprop}.
In this section we will use both results from the {\it right} and {\it left} side on $\cC\cH^+$.
Fix $u_{\diamond}=v_{\diamond}$, such that moreover Proposition \ref{rechtes} holds with $v_{\diamond}=v_*$,
and such that its {\it left} side analog holds with $u_{\diamond}=u_*$.
We will consider a region $\cR_{VI}=\left\{u_{\diamond}\leq u \leq \hat{u}, v_{\diamond}\leq v \leq \hat{v}\right\}$, with $\hat{u} \in (u_{\diamond}, \infty)$ and $\hat{v} \in (v_{\diamond}, \infty)$, cf.~ Figure \ref{diamant}.
{\begin{figure}[ht]
\centering
\includegraphics[width=0.4\textwidth]{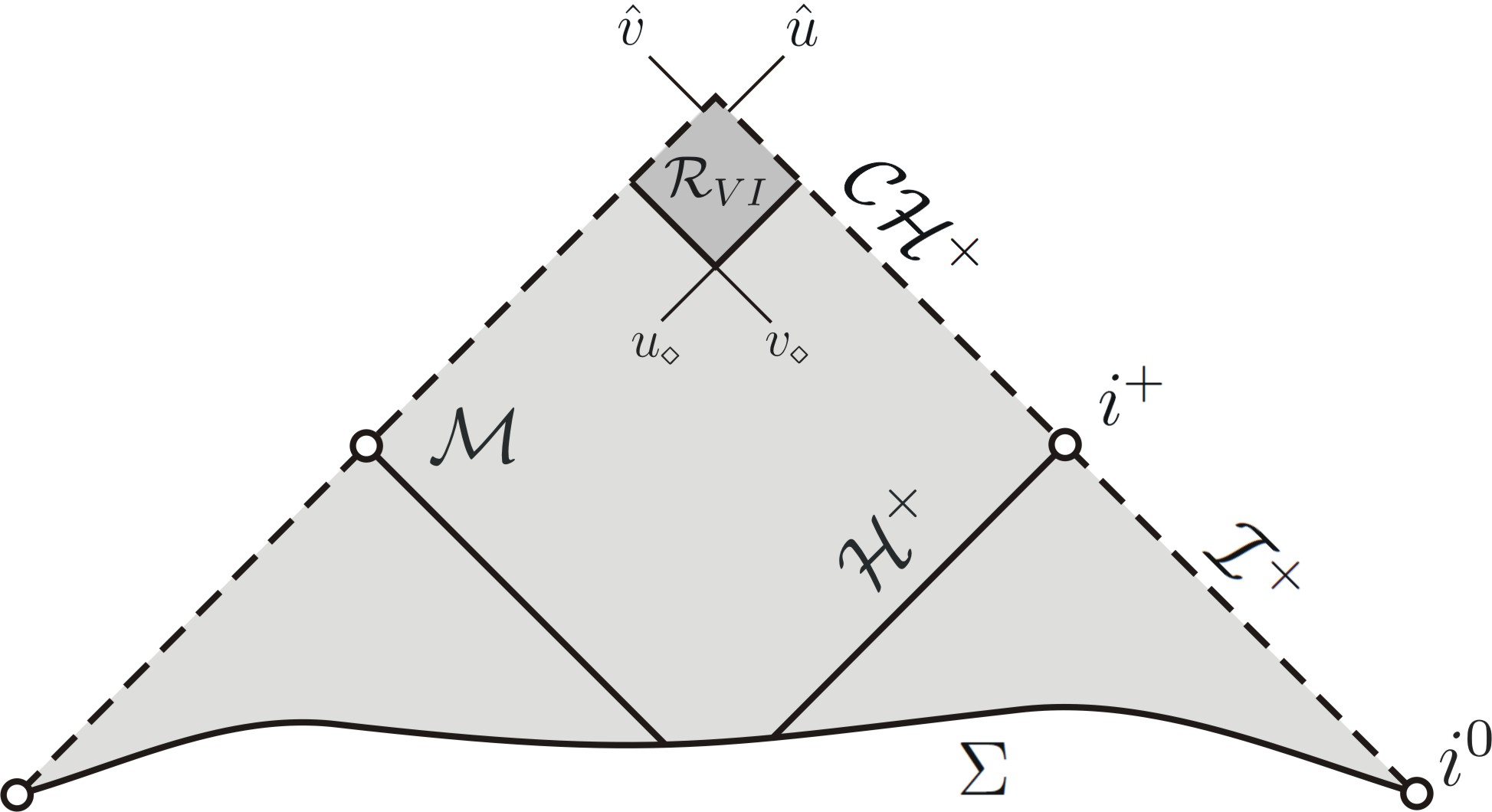}
\caption[]{Representation of region $\cR_{VI}$ in the $(u,v)$-plane.}
\label{diamant}\end{figure}}
Recall, that in Section \ref{blue_future} we have defined the weighted vector field\footnote{Since $u$ is always positive in the remaining region under consideration $\cR_{VI}$, we have omitted the absolute value in the $u$-weight.}
\bea
\lb{{S2}}
{S}=v^p(\partial_v+b^{\tilde{\phi}}\partial_{\tilde{\phi}})+u^p\partial_u,
\eea
which we are going to use again, to obtain an energy estimate up to the bifurcate two-sphere.
Recall $K^{S}$ given in \eqref{KS},
from which we see, that we have positive dominating terms multiplying $|\nabb \psi|^2$, since $\cR_{VI}$ is located in the blueshift region.
We further defined $\tilde{K}^{S}$ in \eqref{KStilde}
and \eqref{ksrelation} which will be useful to state the following proposition.
\begin{lem}
\label{K_spher_ende}
Let $\psi$ be an arbitrary function. Then, for all $(u_{\diamond},v_{\diamond})\in J^+(\gamma)\cap\cB$ and all $\hat{u}>u_{\diamond}$, all $\hat{v}>v_{\diamond}$,
the integral over \mbox{$\cR_{VI}$}, cf.~ Figure \ref{diamant} of the current $\tilde{K}^{S}$, defined by \eqref{KStilde}, can be estimated by
\bea
\lb{deltaende}
\int\limits_{\cR_{VI}} |\tilde{K}^{S}(\psi)| \dV &\leq& \delta_1 \sup_{u_{\diamond}\leq \bar{u}\leq \hat{u}}\int\limits_ {\left\lbrace  v_{\diamond} \leq v \leq \hat{v}\right\rbrace }  J_{\mu}^{S}(\psi) n^{\mu}_{u=\bar{u}}\dV_{u=\bar{u}}\nonumber\\
&&  + \delta_2 \sup_{v_{\diamond}\leq \bar{v} \leq \hat{v}}\int\limits_ {\left\lbrace  u_{\diamond} \leq u \leq \hat{u}\right\rbrace }J_{\mu}^{S}(\psi) n^{\mu}_{v=\bar{v}} \dV_{v=\bar{v}},
\eea
where $\delta_1$ and $\delta_2$ are positive constants, with $\delta_1\rightarrow 0$ as $u_{\diamond}\rightarrow \infty$ and $\delta_2\rightarrow 0$ as $v_{\diamond}\rightarrow \infty$.
\end{lem}
\begin{proof}
The proof is similar to the proof of Lemma \ref{K_spher} of Section \ref{bounding_bulk_S} and Lemma \ref{K_spherV} of Section \ref{region5_proof}.
We still need to show finiteness and smallness of \mbox{$\int\limits^{u_{\gamma}(v_*)}_{u_{\gamma}(\hat{v})} \sup_{v_{\gamma}(\bar{u})\leq \bar{v} \leq \hat{v}}
\left|\frac{\partial_{{\bar{v}}}( L^2\sin^2\theta)}{L^2 \sin^2 \theta}+\frac{\partial_{{\bar{u}}}( L^2\sin^2\theta)}{L^2 \sin^2 \theta}\frac{|{\bar{u}}|^p}{{\bar{v}}^p}+\sqrt{\partial_{\bar{u}} b^{\tilde{\phi}}}\right|
\md \bar{u}$} and 
\mbox{$\int\limits^{\hat{v}}_{v_*}\sup_{u_{\gamma}(\bar{v}) \leq \bar{u}\leq u_{\gamma}(v_*)}\left|\frac{\partial_{{\bar{u}}}( L^2\sin^2\theta)}{L^2 \sin^2 \theta}+\frac{\partial_{{\bar{v}}}( L^2\sin^2\theta)}{L^2 \sin^2 \theta}\frac{{\bar{v}}^p}{|{\bar{u}}|^p}+\sqrt{\partial_{\bar{u}} b^{\tilde{\phi}}}\right|
\md \bar{v}$}.
In Section \ref{finiteness} we derived \eqref{omegafix} which we will now use
for all $\bar{u},\bar{v} \in J^+(u_{\diamond}, v_{\diamond})$
Therefore, we can write
\bea
&&\int\limits^{u_{\gamma}(v_*)}_{u_{\gamma}(\hat{v})} \sup_{v_{\gamma}(\bar{u})\leq \bar{v} \leq \hat{v}}
\left|\frac{\partial_{{\bar{v}}}( L^2\sin^2\theta)}{L^2 \sin^2 \theta}+\frac{\partial_{{\bar{u}}}( L^2\sin^2\theta)}{L^2 \sin^2 \theta}\frac{|{\bar{u}}|^p}{{\bar{v}}^p}+\sqrt{\partial_{\bar{u}} b^{\tilde{\phi}}}\right|
\md \bar{u}\nonumber\\
&\leq& C\int\limits_{u_{\diamond}}^{\hat{u}} \sup_{v_{\diamond}\leq \bar{v} \leq \hat{v}}\left[\Omega^2(u_{\diamond}, v_{\diamond}, \theta^{\star})e^{-\beta\left[\bar{v}-v_{\diamond}+\bar{u}-u_{\diamond}\right]}\left( 1 + \frac{{\bar{u}}^p}{{\bar{v}}^p}\right)+
|\Omega(u_{\diamond}, v_{\diamond}, \theta^{\star})|e^{-\frac{\beta}{2}\left[\bar{v}-v_{\diamond}+\bar{u}-u_{\diamond}\right]}
\right] \md \bar{u},\nonumber\\
&\leq& \tilde{C}\int\limits_{u_{\diamond}}^{\hat{u}}\left[ \Omega^2(u_{\diamond}, v_{\diamond}, \theta^{\star})e^{-\beta\left[\bar{u}-u_{\diamond}\right]}
\left( 1 + \frac{{\bar{u}}^p}{v_{\diamond}^p}\right)
+|\Omega(u_{\diamond}, v_{\diamond}, \theta^{\star})|e^{-\frac{\beta}{2}\left[\bar{u}-u_{\diamond}\right]}
\right] \md \bar{u}
\leq \delta_1, 
\eea
where $\delta_1\rightarrow 0$ as $u_{\diamond}= v_{\diamond}\rightarrow \infty$ (since $\Omega^2(u_{\diamond}, v_{\diamond}, \theta^{\star})\rightarrow 0$, cf.~ \eqref{R_delta_sigma}). 
Similarly, for finiteness of the second term we obtain
\bea
 &&\int\limits^{\hat{v}}_{v_*}\sup_{u_{\gamma}(\bar{v}) \leq \bar{u}\leq u_{\gamma}(v_*)}\left|\frac{\partial_{{\bar{u}}}( L^2\sin^2\theta)}{L^2 \sin^2 \theta}+\frac{\partial_{{\bar{v}}}( L^2\sin^2\theta)}{L^2 \sin^2 \theta}\frac{{\bar{v}}^p}{|{\bar{u}}|^p}+\sqrt{\partial_{\bar{u}} b^{\tilde{\phi}}}\right|
\md \bar{v}\nonumber\\
&\leq& C\int\limits_{v_{\diamond}}^{\hat{v}} \sup_{u_{\diamond}\leq \bar{u} \leq \hat{u}}\left[\Omega^2(u_{\diamond}, v_{\diamond}, \theta^{\star})e^{-\beta\left[\bar{v}-v_{\diamond}+\bar{u}-u_{\diamond}\right]}\left( 1 + \frac{{\bar{v}}^p}{{\bar{u}}^p}\right)+
|\Omega(u_{\diamond}, v_{\diamond}, \theta^{\star})|e^{-\frac{\beta}{2}\left[\bar{v}-v_{\diamond}+\bar{u}-u_{\diamond}\right]}
\right] \md \bar{v},\nonumber\\
&\leq& \left[ {\tilde{C}}\int\limits_{v_{\diamond}}^{\hat{v}} \Omega^2(u_{\diamond}, v_{\diamond}, \theta^{\star})e^{-\beta\left[\bar{v}-v_{\diamond}\right]}
\left( 1 + \frac{{\bar{v}}^p}{u_{\diamond}^p}\right)
+|\Omega(u_{\diamond}, v_{\diamond}, \theta^{\star})|e^{-\frac{\beta}{2}\left[\bar{v}-v_{\diamond}\right]}
\right] \md \bar{v}
\leq \delta_2, 
\eea
where $\delta_2\rightarrow 0$ as $u_{\diamond}= v_{\diamond}\rightarrow \infty$. 
Thus we obtain the statement of Lemma \ref{K_spher} by fixing $u_{\diamond}=v_{\diamond}$ sufficiently large and by again using the bound \eqref{boundedL} as we had done it in \eqref{vertausch}.
\end{proof}

\begin{prop}
\lb{bifu}
Let $\psi$ be as in Theorem \ref{anfang}, and $Y^k$ as in \eqref{angular_comm} with \eqref{Yk},  \eqref{YkK}, \eqref{YkE} and all $k \in \left\{0,1,2\right\}$.
Then, for $u_{\diamond}=v_{\diamond}$ sufficiently close to $\infty$ and $\hat{u} >u_{\diamond}$, $\hat{v} >v_{\diamond}$
\bea
\int\limits_ {\left\lbrace  v_{\diamond} \leq v \leq \hat{v}\right\rbrace } J_{\mu}^{S}(Y^k\psi) n^{\mu}_{u=\tilde{u}} \dV_{u=\tilde{u}}+
\int\limits_ {\left\lbrace  u_{\diamond} \leq u \leq \hat{u}\right\rbrace } J_{\mu}^{S}(Y^k\psi) n^{\mu}_{v=\tilde{v}} \dV_{v=\tilde{v}}&\leq&  C(u_{\diamond}, v_{\diamond}),
\eea
where $C$ depends on ${C_{0}}$ of Theorems \ref{anfang} and ${D_{0}(u_{\diamond}, v_{\diamond})}$ of Propositions \ref{initialdataprop} (and their {\it left} side analogs, which we obtain from interchanging relations \eqref{uv} and \eqref{paruv}). 
\end{prop}
\begin{proof}
The proof for $k=0$ follows from applying the divergence theorem for the current $J_{\mu}^{S}(\psi)$ in the region $\cR_{VI}$. The past boundary terms are estimated by Proposition \ref{rechtes} and its {\it left} side analog. Note that the weights of $J_{\mu}^{S}(\psi)$ are comparable to the weights of $J_{\mu}^{W}(\psi)$ for fixed $u_{\diamond}$, and similarly the weights of $J_{\mu}^{S}(\psi)$ are comparable to the weights of the {\it left} analog of $J_{\mu}^{W}(\psi)$ for fixed $v_{\diamond}$. The bulk term is absorbed by Lemma \ref{K_spher_ende}.
 
Recall \eqref{KStildek} to \eqref{tildeineq2} and note that analogously to Lemma \ref{K_spher2}, we can now state the following lemma.
\begin{lem}
\label{K_spher2_end}
Let $\psi$ be an arbitrary function, and $Y^k$ as in \eqref{angular_comm} with \eqref{Yk},  \eqref{YkK}, \eqref{YkE} and all $k \in \left\{0,1,2\right\}$. Then, for all $v_*>2\alpha$ 
and all $\hat{v}>v_*$,
the integral over region \mbox{$\cR_{IV}=J^+(\gamma)\cap J^-(x)$} with \mbox{$x=(u_{\gamma}(v_*), \hat{v})$}, $x \in \cB$, cf.~ Figure \ref{RN_mit_u_v}, of the current $\tilde{K}^{S}$, defined by \eqref{KStildek}, can be estimated by
\bea
\lb{S_lemma_ende} 
\int\limits_{\cR_{VI}} |\tilde{K}^{S}(Y^k \psi)| \dV 
&\leq&\delta_1 \sup_{u_{\gamma}(\hat{v})\leq \bar{u}\leq u_{\gamma}(v_*)}\int\limits_{\left\lbrace  v_{\gamma}(\bar{u}) \leq v \leq \hat{v}\right\rbrace }  J_{\mu}^{S}(Y^k\psi) n^{\mu}_{u=\bar{u}}\dV_{u=\bar{u}}\nonumber\\
&&+ \delta_2 \sup_{v_*\leq \bar{v} \leq \hat{v}}\int\limits_{\left\lbrace  u_{\gamma}(\hat{v}) \leq u \leq u_{\gamma}(\bar{v}) \right\rbrace }
J_{\mu}^{S}(Y^k\psi) n^{\mu}_{v=\bar{v}} \dV_{v=\bar{v}},\nonumber\\
\eea
where $\delta_1$ and $\delta_2$ are positive constants, with $\delta_1\rightarrow 0$ and $\delta_2\rightarrow 0$ as $v_*\rightarrow \infty$.
\end{lem}
\begin{proof}
The proof follows from Lemma \ref{K_spher_ende} and the arguments given in the proof of Lemma \ref{K_spher2}.\\
\end{proof}
By using this result and the divergence theorem up to third order, we then have proven Proposition \ref{bifu} for \mbox{$k \in \left\{0,1,2\right\}$}.
\end{proof}

Now that we have shown boundedness for the different subregions of the interior we can state the following theorem for the entire interior region. 
\begin{thm}
\lb{letzteprop}
Let $\psi$ be as in Theorem \ref{anfang}, and $Y^k$ as in \eqref{angular_comm} with \eqref{Yk},  \eqref{YkK}, \eqref{YkE} and all $k \in \left\{0,1,2\right\}$.
Then 
\bea
\int\limits_{\bbS^2_{u,v}}\int\limits^{\infty}_{v_{fix}}&&\left[ (|v|+1)^p (\partial_v (Y^k\psi)+ b^{\tilde{\phi}} \partial_{\tilde{\phi}} (Y^k\psi))^2(u_{fix}, v, \theta^{\star}, \tilde{\phi}) \right.\nonumber\\
&&\left.\qquad+\Omega^2|\nabb (Y^k\psi)|^2(u_{fix}, v, \theta^{\star}, \tilde{\phi}) \right]L\md v\md \sigma_{\mathbb S^2}\leq E_k, \\
&& \mbox{for  $v_{fix} \geq v_{\schere}$, $u_{fix} > -\infty$ }, \nonumber\\
\int\limits_{\bbS^2_{u,v}}\int\limits^{\infty}_{u_{fix}}&& \left[(|u|+1)^p (\partial_u (Y^k\psi))^2 (u, v_{fix}, \theta^{\star}, \tilde{\phi})\right.\nonumber\\
&&\left.\qquad+\Omega^2|\nabb (Y^k\psi)|^2(u, v_{fix}, \theta^{\star}, \tilde{\phi}) \right]L\md u\md \sigma_{\mathbb S^2} \leq E_k,\\ 
&&\mbox{for  $u_{fix} \geq u_{\schere}$, $v_{fix} > -\infty$}, \nonumber
\eea
where $p$ is as in \eqref{waspist} and 
$E_k$ depends on ${C_{0}}$ of Theorem \ref{anfang} and ${D_{0}(u_{\diamond}, v_{\diamond})}$ of Proposition \ref{initialdataprop} (and their {\it left} side analogs) where $u_{\diamond}=v_{\diamond}$ is as in Proposition \ref{bifu}. 
\end{thm}
\begin{proof}
This follows by examining the weights in Proposition \ref{rechtes} (and its {\it left} side analog) and \ref{bifu} together with Theorem \ref{Xithm} (and its {\it left} side analog). 
\end{proof}

\section{Pointwise boundedness and continuity}
\lb{globalenergy}
In order to derive pointwise estimates up to and including $\cC\cH^+$ we need weighted higher order energy estimates up to $\cC\cH^+$, as we have derived in Theorem \ref{letzteprop}. Recall that the crux of our proof was to obtain estimates which actually reach up to $v=\infty$, which we have shown in our analysis of the characteristic rectangle $\Xi$. 
The proof of the pointwise boundedness is now completely analogous to the proof in Section \ref{uni_bounded}, just that we also have to integrate in $u$ direction as follows.
We estimate
\bea
\lb{uglobalfundcauchy}
\int\limits_{\bbS^2_{u,v}} (Y^k \psi) ^2(\hat{u}, v)L \md \sigma_{\mathbb S^2}
&\leq& \tilde{C}\left[\int\limits_{\bbS^2_{u,v}}\left(\int\limits_{u_*}^{\hat{u}} (|u|+1)^p(\partial_u Y^k\psi)^2(u,v)\md u\int\limits_{u_*}^{\hat{u}} (|u|+1)^{-{p}}\md u\right)L\md \sigma_{\mathbb S^2}\right.\nonumber\\
&&\left.\qquad +\int\limits_{\bbS^2_{u,v}}  (Y^k \psi) ^2(u_*,v)L\md \sigma_{\mathbb S^2}\right],\nonumber \\
&\leq&\tilde{C}\left[\tilde{\tilde{C}} E_{k}+\int\limits_{\bbS^2_{u,v}} \psi^2(u_*,v)L\md \sigma_{\mathbb S^2}\right],
\eea 
where $u_*\geq u_{\schere}$, $\hat{u} \in (u_*, \infty)$ and $v \in (1, \infty)$ and $k \in \bbN^0$. 

By using the result \eqref{supr2} in \eqref{uglobalfundcauchy} we derive pointwise boundedness according to \eqref{sobo_embed} 
\bea
\lb{usupr}
&\sup_{\left\{\theta^{\star},\tilde{\phi}\right\}\in\bbS^2_{u,v}}&|\psi(\hat{u},{v},\theta^{\star},\tilde{\phi})|^2 \nonumber\\
&\leq& \tilde{C} \left[\int\limits_{\bbS^2_{u,v}} ( \psi) ^2(\hat{u},{v})L\md \sigma_{\mathbb S^2} +\int\limits_{\bbS^2_{u,v}} (Y \psi) ^2(\hat{u},{v})L\md \sigma_{\mathbb S^2} +\int\limits_{\bbS^2_{u,v}} (Y^2 \psi) ^2(\hat{u},{v})L\md \sigma_{\mathbb S^2} \right],\nonumber\\
&\leq&\tilde{C}\left[\tilde{\tilde{C}}\left( E_{0}+  E_{1}+ E_{2}\right)+C )\right]
\leq C,
\eea
with $C$ depending on the initial data. 

Inequalities \eqref{usupr} and \eqref{supr2} give the desired \eqref{maineq} for all $v\geq1$.
Using the interchanges \eqref{uv} and \eqref{paruv},
equation \eqref{maineq} follows likewise for all $u\geq 1$. 
The remaining subset of the interior has compact closure in spacetime for which \eqref{maineq} thus follows by Cauchy stability. We have thus shown \eqref{maineq} globally in the interior.

The continuity statement of Theorem \ref{main}
follows easily by estimating \mbox{$|\psi(u,v,\theta^{\star},\tilde{\phi})- \psi(\tilde{u},{v},\theta^{\star},\tilde{\phi})|$} via the fundamental theorem of calculus and Sobolev embedding\footnote{We estimate as in \eqref{uglobalfundcauchy}
and \eqref{usupr} but exploiting that $\int_{\tilde{u}}^{u} (|u|+1)^{-p}\to 0$ as $\tilde{u}, u\to \infty$.}, and similarly for $v, \theta^{\star}$ and $\tilde{\phi}$ in the role of $u$. For a more detailed discussion of the continuity argument, see \cite{anne_thesis}.
\begin{trivlist}
\item[\hskip \labelsep ]\qed\end{trivlist}

\section{Outlook}
\lb{outlook}
In the following we would like to give an overview of recent related results, as well as open problems. For further results see the outlook of \cite{anne} as well as the more detailed review given in part I of \cite{anne_thesis}.

\subsection{Scalar wave equation without symmetry in black hole \textit{interiors}}

Apart from the ``poor man's'' linear stability results obtained in this and the previous paper \cite{anne}, note the related decay results by Hintz \cite{peter1} at the Cauchy horizon on Reissner--Nordstr\"om and slowly rotating Kerr backgrounds. Similar results along the Cauchy horizon on asymptotically de Sitter backgrounds by Hintz and Vasy are given in \cite{peter2}. 

This suggested singular behavior was analyzed further in \cite{luk_oh} by Luk and Oh.
They were able to prove certain ``poor man's'' linear instability properties along $\cC\cH^+$ for scalar waves on fixed subextremal Reissner--Nordstr\"om backgrounds. In particular, they showed that due to the blueshift effect generic smooth and compactly supported initial data on a Cauchy hypersurface indeed give rise to solutions with infinite non-degenerate\footnote{By non-degeneracy we mean that the multiplier is constructed such that it does not become null on the hypersurfaces of interest. } energy
near the Cauchy horizon in the interior of the black hole.  They proved that
the solution generically does not belong to $W^{1,2}_{\mbox{\tiny{loc}}}$.
A recent result of Gleeson showed that the singularity is in fact even stronger, so that we can state the following theorem.
\begin{thm}
\lb{RN_instab}
(``Poor man's'' linear Reissner--Nordstr\"om instability without symmetries, Luk and Oh, \cite{luk_oh} and Gleeson, \cite{eavan}.)
Generic smooth and compactly supported initial
data to the wave equation on fixed subextremal Reissner--Nordstr\"om background on a two-ended asymptotically flat Cauchy surface $\Sigma$
give rise to solutions that are not in
$W^{1,p}_{\mbox{\tiny{loc}}}$, for all $1<p<2$, for the subextremal range $0 < \log{\frac{ r_+}{r_-}}< 1$;
\footnote{Although Reissner--Nordstr\"om spacetimes with subextremal range approaching the extremal range are expected to show a more stable behavior on $\cC\cH^+$ in comparison to black holes with a small charge and mass ratio, it remains open to show that solutions are not in $W^{1,p}_{\mbox{\tiny{loc}}}$, for all $1<p<2$, for the subextremal range $1\leq \log{\frac{ r_+}{r_-}}$.}
and not in
$W^{1,2}_{\mbox{\tiny{loc}}}$, for $p\geq 2$ for the entire subextremal range,
in a neighborhood of any point on the future Cauchy horizon
$\cC\cH^+$.\footnote{ A function $\psi$ belonging to the Sobolev space $W^{1,2}_{\mbox{\tiny{loc}}}$ would have the properties that locally $\psi$ and all of its first weak derivatives exist and are square integrable. Taking \eqref{wave_psi} seriously as a model for the full Einstein field equations and therefore considering $\psi$ as an agent for the metric two tensor $g$, the result of Luk and Oh suggests that in the full theory the Christoffel symbols would fail to be square integrable. For an introduction to Sobolev spaces refer for example to \cite{evans}, \cite{taylor}.}
\end{thm}
A version of the Strong Cosmic Censorship Conjecture demanding that solutions should be inextendible in $W^{1,p}_{\mbox{\tiny{loc}}}$ beyond $\cC\cH^+$ could be interpreted as generically the Lorentzian manifold cannot be extended even in a weak sense such that the Einstein equations are still satisfied. This was already pointed out in \cite{christo_form} by Christodoulou. 
Note however, that adding a cosmological constant can change the behavior, as stated below in a rough version of the original theorem.
\begin{thm}
\lb{mainThmdS} (``Poor man's'' linear Reissner--Nordstr\"om-de Sitter stability without symmetries, Costa and A.F., \cite{joao_ich}.)
For solutions of the wave equation on fixed subextremal Reissner--Nordstr\"om--de Sitter spacetime, satisfying
\bea
\lb{Price}
\int_{v_1}^{v_2}\int_{\bbS^2}\left[ \left(\partial_v\psi\right)^2
+\left|{\nabb}\psi\right|^2\right] \md v \md \sigma_{\bbS^2}
&\lesssim& e^{-2p v_1}\;,
\eea
for all $0\leq v_1\leq v_2$. The energy flux can be shown to be bounded, even in the transverse directions if
\begin{equation}
\label{mainCond}
2\min\{p,\kappa_+\}>\kappa_-\;.
\end{equation}
\end{thm}

Analogous to the Reissner--Nordstr\"om instability results, Luk and Sbierski could show the following on fixed Kerr backgrounds.
\begin{thm}
\lb{jo_jan}
(``Poor man's'' linear Kerr instability without symmetries, Luk and Sbierski, \cite{luk_jan}).
Let $\psi$ be a solution of \eqref{wave_psi}. If the energy of $\psi$ along $\cH^+$ obeys some polynomial upper and lower bounds, then the non-degenerate energy on any spacelike hypersurface intersecting the Cauchy horizon transversally is infinite.
\end{thm}
Further, Dafermos and Shlapentokh-Rothmann, \cite{m_yakov} obtained scattering maps for axisymmetric scalar waves on fixed Kerr backgrounds in the {\it exterior} as well as the {\it interior} with conclusion concordant with the above statements. For precise statements and further scattering results see \cite{m_yakov, m_rueck, jp} and references therein. For results concerning extremal black holes see \cite{stef1,stef2, stef3} by Aretakis and \cite{dejan1, dejan2} by Gajic.

\subsection{Linear perturbations without symmetry on \textit{exterior} backgrounds}

Going from ``poor man's'' linear stability and instability results to full linear considerations the following theorem was proven for Schwarzschild backgrounds.
\begin{thm}
\lb{schw_stab}
(Full linear stability of Schwarzschild, Dafermos et al., \cite{m_linstab}). Solutions for the linearisation
of the Einstein equations around Schwarzschild arising from regular admissible\footnote{For Schwarzschild spacetime this means that the Cauchy hypersurface does not intersect the white hole of the solution.}
data remain bounded in the exterior and decay (with respect to a hyperboloidal
foliation) to a linearised Kerr solution.
\end{thm}
A full linear stability result for slowly rotating Kerr-de Sitter spacetime was obtained by Hintz and Vasy, see Theorem 1.2 in \cite{peter3}.

\subsection{Non-linear perturbations without symmetry on \textit{exterior} backgrounds}

Considering the Einstein-Maxwell-scalar field system, Luk and Oh were able to show that $L^2$-averaged (inverse) polynomial lower bounds for the derivative of the scalar field hold along each of the horizons, arising from a generic set of Cauchy initial data, see \cite{luk_oh2}. Using this result as data for the analysis in the interior, they were able to prove Theorem \ref{luk_oh_sccc}, stated in the next section, which is related to the Strong Cosmic Censorship Conjecture.

On  the level of full non-linear Einstein field equations, Hintz and Vasy were recently able to derive stability results in Kerr-de Sitter spacetime. The following informal version of Theorem 1.4 of \cite{peter3} was stated in their work.

\begin{thm}
(Stability of the Kerr-de Sitter family for small $a$, Hintz and Vasy \cite{peter3}).
Suppose $(h,k)$ are smooth initial data on $\Sigma_0$, satisfying the constraint equations,
which are close to the data $(h_{b_0}, k_{b_0})$
of a Schwarzschild-de Sitter spacetime in a
high regularity norm. Then there exist a solution
$g$ of $R_{\mu \nu}+\lambda g_{\mu \nu}=0$
attaining these initial
data at $\Sigma_0$, and black hole parameters
$b$ which are close to
$b_0$, so that 
\ben
g-g_b=\cO(e^{-\alpha t_*})
\een
for a constant $\alpha>0$ independent of the initial data; that is, $g$ decays exponentially fast to the Kerr-de Sitter metric $g_b$. Moreover, $g$ and $b$ are quantitatively controlled by $(h,k)$.
\end{thm}

\subsection{Non-linear perturbations without symmetry on \textit{interior} backgrounds}

As mentioned in the last section, in combination with the exterior results \cite{luk_oh2}, Luk and Oh proved the following Theorem.
\begin{thm}
\lb{luk_oh_sccc} ($C^2$ stability result for Einstein-Maxwell-(real)-scalar field system, Luk and Oh, \cite{luk_oh1}). The $C^2$-formulation of the Strong Cosmic Censorship Conjecture
for the Einstein-Maxwell-(real)-scalar field system in spherical symmetry with 2-ended asymptotically
at initial data on $R\times \bbS_2$ is true.
\end{thm}

A first non-linear step toward investigation of null singularities in vacuum spacetime was achieved by Luk.
In \cite{luk3} he managed to construct examples of local patches of vacuum spacetimes with weak null singular boundary.
In subsequent work of Dafermos and Luk \cite{m_luk} it was shown that fully non-linear perturbations of Kerr spacetime do not lead to formation of a strong spacelike singularity. In fact $C^0$ stability holds up to $\cC\cH^+$. Putting together insights from preceding investigations for the wave equation on fixed background as well as coupled scalar field analyzes, see \cite{m_price, m_stab, m_interior, anne}, they were able to prove the following Theorem as already stated in \cite{m_survey}.
\begin{thm}
\lb{kerr_cauchy}
(Global stability of the Kerr Cauchy horizon, Dafermos and Luk, \cite{m_luk}). Consider characteristic
initial data for \eqref{EF2} on a bifurcate null hypersurface $\cH^+ \cup \cH^+$, where $\cH^{\pm}$
have future-affine complete null generators and their induced geometry dynamically
approaches that of the event horizon of Kerr with \mbox{$0 < |a| < M$} at a sufficiently
fast polynomial rate. Then the maximal development can
be extended beyond a bifurcate Cauchy horizon $\cC\cH^+$ as a Lorentzian manifold with
$C^0$ metric. All finitely-living observers pass into the extension.
\end{thm}

\section*{Acknowledgements}
I wish to express my gratitude towards Mihalis Dafermos for inspiring this problem and for comments on the manuscript. 
Further, I would like to thank Pedro Gir\~ao and Jos\'e Nat\'ario for sharing their insights that helped me completing this work. I also benefited from discussions with Dejan Gajic, Philipp Kuhn, George Moschidis and Sashideep Gutti.   
This work was partially supported by FCT/Portugal through UID/MAT/04459/2013, grant (GPSEinstein) PTDC/MAT-ANA/1275/2014 and SFRH/BPD/115959/2016.

\begin{appendix}

\section{The $K$-current}
\lb{Kcurrents}
In order to compute all scalar currents according to \eqref{K}
in $(u,v)$ coordinates we first derive the components of the deformation tensor which is given by
\bea
\lb{deftensor}
(\pi^X)^{\mu \nu}=\frac12(g^{\mu \lambda}\partial_{\lambda}X^{\nu}+g^{\nu \sigma}\partial_{\sigma}X^{\mu}+g^{\mu \lambda}g^{\nu \sigma}g_{\lambda \sigma, \delta}X^{\delta}),
\eea
where $X$ is an arbitrary vector field, \mbox{$X=X^u \partial_u+X^v \partial_v+X^{\theta_C}\partial_{\theta_C}$}, with $X^u$, $X^v$ and $X^{\theta_C}$ depending on $u,v,\theta^{\star}, \tilde{\phi}$.
From this we obtain
\bea
\lb{pi-terme}
(\pi^X)^{v v}&=&-\frac{1}{2\Omega^2}\partial_u X^v,\\
(\pi^X)^{u u}&=&-\frac{1}{2\Omega^2}\partial_v X^u,\\
(\pi^X)^{u v}&=&-\frac{1}{4\Omega^2}\partial_uX^{\theta_C}-\frac{b^{\theta_C}}{4\Omega^2}\partial_uX^v +\frac12 (\gin^{-1})^{{\theta_C}{\theta_D}}\partial_{\theta_D} X^v ,\\
(\pi^X)^{u {\theta_C}}&=&-\frac{b^{\theta_D}}{4\Omega^2}\partial_{{\theta_D}}X^{\theta_C}-\frac{1}{4\Omega^2}\partial_v X^{\theta_C}-\frac{b^{\theta_C}}{4\Omega^2}\partial_uX^u+\frac12 (\gin^{-1})^{{\theta_C}{\theta_D}}\partial_{\theta_D} X^u \nonumber\\
&&-\frac{b^{\theta_C}}{2\Omega^2}\frac{\partial_{\eta} \Omega }{\Omega}X^{\eta} 
+\frac{1}{4\Omega^2}(\gin^{-1})^{{\theta_C}{\theta_D}}\partial_{\eta}b_{\theta_D} X^{\eta}\nonumber\\
&&-\frac{b^{\theta_D}}{4\Omega^2}(\gin^{-1})^{{\theta_C}{\theta_A}}\partial_{\eta}\gin_{{\theta_D}{\theta_A}} X^{\eta},\\
(\pi^X)^{{\theta_C} {\theta_D}}&=&-\frac{b^{\theta_C}}{4\Omega^2}\partial_U X^{\theta_D}+ \frac12(\gin^{-1})^{{\theta_C}{\theta_A}}\partial_{{\theta_A}}X^{\theta_D}\nonumber\\
&&-\frac{b^{\theta_D}}{4\Omega^2}\partial_u X^{\theta_C}+ \frac12(\gin^{-1})^{{\theta_D}{\theta_A}}\partial_{{\theta_A}}X^{\theta_C}\nonumber\\
&&+\frac12 (\gin^{-1})^{{\theta_C}{\theta_A}}(\gin^{-1})^{{\theta_D}{\theta_B}}\partial_{\eta} \gin_{{\theta_A}{\theta_B}} X^{\eta},
\eea
with $A,B,C,D=1,2$, $\theta_1=\theta^{\star}$, $\theta_2=\tilde{\phi}$ ,$\zeta=u,v$ and $\eta=u,v,\theta^{\star}, \tilde{\phi}$.
From \eqref{energymomentum} we calculate the components of the energy momentum tensor in $(u,v)$ coordinates as
\bea
T^{KG}_{v v}&=&(\partial_v \psi)^2+\frac{1}{2}|b|^2\left(\frac{1}{\Omega^2}(\partial_u\psi)( \partial_v \psi+{b^{\theta_C}}\partial_{{\theta_C}} \psi)-|\nabb \psi|^2\right),\\
T^{KG}_{u u}&=&(\partial_u \psi)^2,\\
T^{KG}_{u v}&=&T^{KG}_{v u}=\Omega^2|\nabb \psi|^2-{b^{\theta_C}}(\partial_u\psi \partial_{{\theta_C}} \psi),\\
T^{KG}_{v {\theta_C}}&=&(\partial_v \psi\partial_{{\theta_C}} \psi)-\frac{b_{\theta_C}}{2}\left(\frac{1}{\Omega^2}(\partial_u\psi)( \partial_v \psi+{b^{\theta_D}}\partial_{{\theta_D}} \psi)-|\nabb \psi|^2\right),\\
T^{KG}_{u {\theta_C}}&=&(\partial_u \psi\partial_{{\theta_C}} \psi),\\
T^{KG}_{{\theta_C} {\theta_D}}&=&(\partial_{{\theta_C}} \psi\partial_{{\theta_D}} \psi)+\frac{1}{2}\gin_{{\theta_C}{\theta_D}}\left(\frac{1}{\Omega^2}(\partial_u\psi)( \partial_v \psi+{b^{\theta_A}} \partial_{{\theta_A}} \psi)-|\nabb \psi|^2
\right),
\eea
with $b^{\theta_C}b_{\theta_C}=|b|^2$ and $\gin_{{\theta_C}{\theta_D}}b^{\theta_D}=b_{\theta_C}$.

Multiplying the components according to \eqref{K} 
in Eddington--Finkelstein-like coordinates 
we get 
\bea
\lb{Kplug_edd-f}
K^X&=& -\left[\frac{1}{2\Omega^2}\partial_u X^v \right](\partial_v \psi)^2\nonumber\\
&&-\left[\frac{1}{2\Omega^2}\partial_v X^u \right](\partial_u \psi)^2\nonumber\\
&&+\left[-\frac{\partial_{\eta} X^{\eta}}{2} -\frac{\partial_{\eta} \Omega}{\Omega} X^{\eta}-\frac14 (\gin^{-1})^{{\theta_C}{\theta_D}}\partial{\eta}\gin_{{\theta_C}{\theta_D}}X^{\eta}+b^{\tilde{\phi}}\partial_{\tilde{\phi}} X^v \right]|\nabb \psi|^2\nonumber\\
&&+\left[\frac{\partial_{\theta_C} X^{\theta_C}}{2}+ \frac14 (\gin^{-1})^{{\theta_C}{\theta_D}}\partial{\eta}\gin_{{\theta_C}{\theta_D}}X^{\eta}-b^{\tilde{\phi}}\partial_{\tilde{\phi}} X^v\right]\frac{1}{\Omega^2}(\partial_u \psi)(\partial_v \psi+b^{\tilde{\phi}}\partial_{\tilde{\phi}}\psi)\nonumber\\
&&+\left[-\frac{b^{\tilde{\phi}}}{2\Omega^2}\partial_{{\tilde{\phi}}}X^{\theta_C}-\frac{1}{2\Omega^2}\partial_v X^{\theta_C}-\frac{b^{\theta_C}}{2\Omega^2}\partial_u X^u+ (\gin^{-1})^{{\theta_C}{\theta_D}}\partial_{\theta_D} X^u\right.\nonumber \\
&&\left.\qquad -\frac{b^{\theta_C}}{\Omega^2}\frac{\partial_{\eta} \Omega }{\Omega}X^{\eta} 
+\frac{\delta^{\theta_C}_{\; \; \tilde{\phi}}}{2\Omega^2}\partial_{\eta}b^{\tilde{\phi}} X^{\eta}
\right](\partial_u \psi\partial_{{\theta_C}} \psi)\nonumber\\
&&+\left[-\frac{1}{2\Omega^2}\partial_uX^{{\theta_C}}-\frac{b^{\theta_C}}{2\Omega^2}\partial_uX^{v}+ (\gin^{-1})^{{\theta_C}{\theta_D}}\partial_{\theta_D} X^v \right](\partial_v \psi\partial_{{\theta_C}} \psi)\nonumber\\
&&+\left[-\frac{b^{\theta_C}}{4\Omega^2}\partial_uX^{\theta_D}-\frac{b^{\theta_D}}{4\Omega^2}\partial_uX^{\theta_C}+\frac12 (\gin^{-1})^{{\theta_C}{\theta_A}}\partial_{\theta_A} X^{\theta_D}+\frac12 (\gin^{-1})^{{\theta_D}{\theta_A}}\partial_{\theta_A} X^{\theta_C}\right.\nonumber\\
&&\left. \quad+ \frac12(\gin^{-1})^{{\theta_C}{\theta_A}}(\gin^{-1})^{{\theta_D}{\theta_B}}\partial{\eta}\gin_{{\theta_A}{\theta_B}}X^{\eta} \right](\partial_{{\theta_C}} \psi\partial_{{\theta_D}} \psi),
\eea
also recall \eqref{det_regel}.

\section{The $J$-currents and normal vectors}
\lb{Jcurrents}
For the $J$-currents, according to equation \eqref{J}, we obtain
\bea
\lb{Jplugnv}
J^X_{\mu}n^{\mu}_{v=const}&=& \frac{1}{2\Omega^2}\left[X^u(\partial_u \psi)^2+[X^{\theta_C}-X^vb^{\theta_C}](\partial_u\psi \partial_{{\theta_C}} \psi)+ \Omega^2 X^v|\nabb \psi|^2\right],\\
\lb{Jplugnu}
J^X_{\mu}n^{\mu}_{u=const}&=&\frac{1}{2\Omega^2}\left[X^v(\partial_v\psi)^2
+\left[X^{\theta_C}+ X^vb^{\theta_C}\right](\partial_v\psi \partial_{{\theta_C}} \psi)+{b^{\theta_C}}X^{\theta_D}(\partial_{{\theta_C}}\psi \partial_{{\theta_D}} \psi)\right.\nonumber\\
&&\left. +X^u{\Omega^2}|\nabb \psi|^2\right],\\
\lb{Jplugnr}
J^X_{\mu}n^{\mu}_{r^{\star}=const}&=&\frac{1}{\sqrt{2\Omega^2}}\left[X^v(\partial_v\psi)^2 +  X^u(\partial_u \psi)^2+[X^{\theta_C}-X^vb^{\theta_C}](\partial_u\psi \partial_{{\theta_C}} \psi)
\right.\nonumber\\
&&\left.+\left[X^{\theta_C}+ X^vb^{\theta_C}\right](\partial_v\psi \partial_{{\theta_C}} \psi)+{b^{\theta_C}}X^{\theta_D}(\partial_{{\theta_C}}\psi \partial_{{\theta_D}} \psi)\right. \nonumber\\
&&\left.+\left(X^u+X^v\right){\Omega^2}|\nabb \psi|^2\right].
\eea
The normal vectors $n^{\mu}_{u=const}$ and $n^{\mu}_{v=const}$ where already stated in \eqref{nu} and \eqref{nv}, respectively. Further, we have
\bea
\lb{normal_r}
n^{\mu}_{r^{\star}=const}&=&\frac{1}{\sqrt{2\Omega^2}}(\partial_u+\partial_v+b^{\theta_C}\partial_{\theta_C}). 
\eea
Moreover, note that for large $v$ the current $J^X_{\mu}n^{\mu}_{\gamma}$ approximates $J^X_{\mu}n^{\mu}_{r^{\star}=const}$  and likewise the normal vector $n^{\mu}_{\gamma}$ approximates $n^{\mu}_{r^{\star}=const}$, where $\gamma$ is as defined in Section \ref{gamma_curve}.

\section{The redshift vector field}
\lb{redshift_app}
We construct the redshift vector field by defining
\bea
\lb{N_red}
N=N^u \partial_u+N^v (\partial_v+b^{\tilde{\phi}}\partial_{\tilde{\phi}}),
\eea
where $N^u$ and $N^v$ depend on $(u,v,\theta^{\star})$. We further require that $N$ is timelike so that we have the condition
\bea
\lb{ncon}
g(N,N)<0.
\eea
From the above we obtain, that $N^u$ and $N^v$ have to have the same sign and should be positive, so that the vector field is future directed. 
Moreover, we have already introduced the vector field $T_{\cH^+}$ in the end of Section \ref{ambient}, see \eqref{T_H}. Adjusting to our coordinates, the vector field 
\bea
\lb{T_H2}
T_{\cH^+}=\frac{\partial_v}{2}-\frac{\partial_u}{2}+\omega_+ \partial_{\tilde{\phi}}
\eea
satisfies the condition of being null at $\cH^+$, spacelike in the interior and timelike in the exterior. The constant $\omega_+$ is defined by $\omega_B$, see equation \eqref{omega_B}, evaluated at $r=r_+$.  
We can now choose $N$ such that
\bea
\lb{ntcon}
g(N,T_{\cH^+})|_{\cH^+}=-2.
\eea
Further, we use \eqref{Kplug_edd-f} to calculate
\bea
\lb{KplugN}
K^N&=& -\left[\frac{1}{2\Omega^2}\partial_u N^v \right](\partial_v \psi)^2\nonumber\\
&&-\left[\frac{1}{2\Omega^2}\partial_v N^u \right](\partial_u \psi)^2\nonumber\\
&&+\left[-\frac{\partial_{u} N^{u}}{2}-\frac{\partial_{v} N^{v}}{2} -\frac{\partial_{u} \Omega}{\Omega} N^{u}-\frac{\partial_{v} \Omega}{\Omega} N^{v}\right.\nonumber\\
&& \quad \left.-\frac14 (\gin^{-1})^{{\theta_C}{\theta_D}}\left(\partial_{u}\gin_{{\theta_C}{\theta_D}}N^{u}+\partial_{v}\gin_{{\theta_C}{\theta_D}}N^{v}\right)\right]|\nabb \psi|^2\nonumber\\
&&+\left[\frac14 (\gin^{-1})^{{\theta_C}{\theta_D}}\left(\partial_{u}\gin_{{\theta_C}{\theta_D}}N^{u}+\partial_{v}\gin_{{\theta_C}{\theta_D}}N^{v}\right)\right]\frac{1}{\Omega^2}(\partial_u \psi)(\partial_v \psi+b^{\tilde{\phi}}\partial_{\tilde{\phi}}\psi)\nonumber\\
&&+\left[-\frac{b^{\tilde{\phi}}}{2\Omega^2}\left( \partial_v N^v +\partial_u N^u\right)+ (\gin^{-1})^{{\tilde{\phi}}\theta^{\star}} \partial_{\theta^{\star}}N^u
-\frac{b^{\tilde{\phi}}}{\Omega^2}\left[\frac{\partial_u \Omega}{\Omega}N^u+\frac{\partial_v \Omega}{\Omega}N^v\right] \right.\nonumber\\
&&\left.\qquad+\frac{1}{2\Omega^2}\partial_u b^{\tilde{\phi}} N^u\right](\partial_u \psi\partial_{\tilde{\phi}} \psi)\nonumber\\
&&+\left[ (\gin^{-1})^{\theta^{\star} \theta^{\star}} \partial_{\theta^{\star}}N^u
\right](\partial_u \psi\partial_{\theta^{\star}} \psi)\nonumber\\
&&+\left[-\frac{1}{2\Omega^2} \partial_v(N^v b^{\tilde{\phi}}) -\frac{b^{\tilde{\phi}}}{2\Omega^2}\partial_u N^v+ (\gin^{-1})^{{\tilde{\phi}}\theta^{\star}} \partial_{\theta^{\star}}N^v\right](\partial_v \psi\partial_{\tilde{\phi}} \psi)\nonumber\\
&&+\left[(\gin^{-1})^{\theta^{\star}\theta^{\star}} \partial_{\theta^{\star}}N^v\right](\partial_v \psi\partial_{\theta^{\star}} \psi)\nonumber\\
&&+\left[-\frac{b^{\tilde{\phi}}}{2\Omega^2}\partial_u\left(N^vb^{\tilde{\phi}}\right)+(\gin^{-1})^{{\tilde{\phi}} \theta^{\star}}\partial_{\theta^{\star}}\left(N^vb^{\tilde{\phi}}\right)\right.\nonumber\\
&&\left. \quad+ \frac12(\gin^{-1})^{{\tilde{\phi}}{\theta_C}}(\gin^{-1})^{{\tilde{\phi}}{\theta_D}}\left(\partial_{u}\gin_{{\theta_C}{\theta_D}}N^{u}+\partial_{v}\gin_{{\theta_C}{\theta_D}}N^{v}\right) \right](\partial_{\tilde{\phi}} \psi)^2\nonumber\\
&&+\left[\frac12(\gin^{-1})^{{\theta^{\star}}{\theta_C}}(\gin^{-1})^{{\theta^{\star}}{\theta_D}}\left(\partial_{u}\gin_{{\theta_C}{\theta_D}}N^{u}+\partial_{v}\gin_{{\theta_C}{\theta_D}}N^{v}\right) \right](\partial_{\theta^{\star}} \psi)^2\nonumber\\
&&+\left[(\gin^{-1})^{{\theta^{\star}}{\theta^{\star}}}\partial_{\theta^{\star}}\left(N^vb^{\tilde{\phi}}\right)+ \frac12(\gin^{-1})^{{\theta^{\star}}{\theta_C}}(\gin^{-1})^{{\tilde{\phi}}{\theta_D}}\left(\partial_{u}\gin_{{\theta_C}{\theta_D}}N^{u}+\partial_{v}\gin_{{\theta_C}{\theta_D}}N^{v}\right)  \right]\nonumber\\
&& \qquad \times (\partial_{\theta^{\star}} \psi\partial_{\tilde{\phi}} \psi).
\eea
Recall Paragraph \ref{bounded} which showed that derivatives of the metric coefficients can be bounded by $\Delta$ which is zero at $\cH^+$. The same holds for derivatives of $b^{\tilde{\phi}}$. 
With the choice $N^u, N^v$ positive, $\partial_u N^u$ negative and with large absolute value and $\partial_u N^v$ negative and with big enough absolute value we can prove Proposition \ref{mi} and Lemma \ref{mi_k}. This can be seen from noticing that all dominating terms are rendered positive, remember \eqref{lowerboundu} and applying the Cauchy--Schwarz inequality. 

\section{Commutators}
\lb{commutators}
The following proposition was already proven in \cite{m_lec} and can also be found in \cite{alinhac} Section 6.2. For the sake of completeness will briefly repeat it here.
\begin{prop}
Let $\psi$ be a solution of the scalar wave equation
\bea
\Box_g \psi=f,
\eea
and Y be an arbitrary vector field. Then
\bea
\lb{comm}
\Box_g(Y\psi)=Y(f)+2(\pi^Y)^{\alpha \beta}\nabla_{\alpha}\nabla_{\beta}\psi+2\nabla^{\alpha}(\pi^Y)_{\alpha \mu}\nabla^{\mu}\psi-\nabla_{\mu}(\pi^Y)_{~\alpha}^{\alpha}\nabla^{\mu}\psi.
\eea
\end{prop}
\begin{proof}
First we state
\bea
\lb{weq}
Y(\Box_g \psi)=\cL_Y(g^{\alpha \beta}\nabla_{\alpha}\nabla_{\beta}\psi)&=&-2(\pi^Y)^{\alpha \beta}\nabla_{\alpha}\nabla_{\beta}\psi+g^{\alpha \beta}\cL_Y(\nabla_{\alpha}\nabla_{\beta} \psi)=Y(f),\\
\lb{lie1}
\cL_Y(\nabla_{\alpha}\nabla_{\beta} \psi)-\nabla_{\alpha}\cL_Y\nabla_{\beta} \psi&=&\left[(\nabla_{\beta}(\pi^Y)_{\alpha \mu})-(\nabla_{\mu}(\pi^Y)_{\beta\alpha})+(\nabla_{\alpha}(\pi^Y)_{\mu\beta})\right]\nabla^{\mu}\psi,\\
\lb{lie2}
\cL_Y\nabla_{\beta} \psi&=&\nabla_{Y}\nabla_{\beta} \psi+\nabla_{\beta}Y^{\mu}\nabla_{\mu} \psi=\nabla_{\beta}(Y\psi).
\eea
Now we use equation \eqref{lie2} in \eqref{lie1} and the result of that in \eqref{weq}. It then only remains to solve the equation for $\Box_g(Y\psi)$ to obtain \eqref{comm}.
\end{proof}

\section{Error terms}
\lb{error_app}
\subsection{The general structure of error terms}
\lb{error}
In order to prove pointwise boundedness we need to commute with all angular operators, as explained in Section \ref{angular}, which are unfortunately not all Killing. Therefore, we are interested in the error term 
\bea
\lb{errorexpression}
\cE^V(Y\psi)=\Box_g(Y \psi) V(Y\psi),
\eea
resulting from commutation with the vector field $Y$ as defined in \eqref{angular_comm}, according to \eqref{E}.
Since the commutator $[\Box_g, Y]\psi$ is the defining quantity for the first order error term it will be useful to analyze it further. First of all, recall \eqref{comm} and notice, that for the given vector field multiplier only $(\pi^Y)^{uu},(\pi^Y)^{vv}$ and $(\pi^Y)^{v\theta_C}$ are zero and all other terms will contribute, see \eqref{pi-Y-termevv} of Section \ref{S_error}. 
We can easily deduce that \eqref{comm} of Section \ref{commutators} will leave us with the following terms
\bea
\lb{error_structure}
[\Box_g, Y]\psi&=& \tilde{E}_1 (\partial_u \psi) +\tilde{E}_2 (\partial_v \psi) +\tilde{E}_3 (\partial_{\tilde{\phi}} \psi) +\tilde{E}_4 (\partial_{\theta^{\star}} \psi)\nonumber\\
&& +\tilde{F}_1 (\partial_u \partial_{\tilde{\phi}}\psi)+ \tilde{F}_2 (\partial_v \partial_{\tilde{\phi}}\psi) + \tilde{F}_3 (\partial_{\theta^{\star}} \partial_{\tilde{\phi}} \psi)+\tilde{F}_4 (\partial_{\tilde{\phi}}^2 \psi)\nonumber\\
&& +\tilde{G}_1 (\partial_u \partial_{\theta^{\star}} \psi) +\tilde{G}_2 (\partial_v \partial_{\theta^{\star}} \psi)+ \tilde{G}_3 (\partial_{\theta^{\star}}^2\psi)+ \tilde{G}_4 (\partial_u \partial_v \psi).
\eea
For the last term we can use the wave equation itself. See Appendix \ref{wave_equation}, equation \eqref{uv_term} and notice that we can add all these terms to equation \eqref{error_structure}, to obtain 
\bea
\lb{error_structure2}
[\Box_g, Y]\psi&=& {E}_1 (\partial_u \psi) +{E}_2 (\partial_v \psi) +{E}_3 (\partial_{\tilde{\phi}} \psi) +{E}_4 (\partial_{\theta^{\star}} \psi)\nonumber\\
&& + {F}_1 (\partial_u \partial_{\tilde{\phi}}\psi)+ {F}_2 (\partial_v \partial_{\tilde{\phi}}\psi) + {F}_3 (\partial_{\theta^{\star}} \partial_{\tilde{\phi}} \psi)+{F}_4 (\partial_{\tilde{\phi}}^2 \psi)\nonumber\\
&& +  {G}_1 (\partial_u \partial_{\theta^{\star}} \psi)+  {G}_2 (\partial_v \partial_{\theta^{\star}} \psi) + {G}_3 (\partial_{\theta^{\star}}^2\psi).
\eea
with all coefficients changed. Further, we require
\bea
\lb{errorexpression}
\cE^V(Y^2\psi)=\Box_g(Y^2 \psi) V(Y^2\psi),
\eea
So we can express the second commutation via \eqref{comm} as
\bea
\lb{error_structure3}
[\Box_g, Y^2 \psi]&=&\left[ {E}_1 (\partial_v \psi) +{E}_2 (\partial_u \psi) +{E}_3 (\partial_{\tilde{\phi}} \psi) +{E}_4 (\partial_{\theta^{\star}} \psi)\right.\nonumber\\
&& + {F}_1 (\partial_u \partial_{\tilde{\phi}} \psi) + {F}_2 (\partial_v \partial_{\tilde{\phi}} \psi) +{F}_3 (\partial_{\theta^{\star}} \partial_{\tilde{\phi}} \psi)+{F}_4 (\partial_{\tilde{\phi}}^2 \psi)\nonumber\\
&& +  {G}_1 (\partial_u \partial_{\theta^{\star}} \psi) + {G}_2 (\partial_v \partial_{\theta^{\star}} \psi) +{G}_3 (\partial_{\theta^{\star}}^2\psi) \nonumber\\
&& +  {H}_1 (\partial_u \partial_{\theta^{\star}} \partial_{\tilde{\phi}} \psi) + {H}_2 (\partial_v \partial_{\theta^{\star}} \partial_{\tilde{\phi}} \psi) +{H}_3 (\partial_{\theta^{\star}}^2 \partial_{\tilde{\phi}}\psi) \nonumber\\
&& + {I}_1 (\partial_u \partial_{\tilde{\phi}}^2 \psi) + {I}_2 (\partial_v \partial_{\tilde{\phi}}^2 \psi)+{I}_3 (\partial_{\theta^{\star}} \partial_{\tilde{\phi}}^2 \psi)+{I}_4 (\partial_{\tilde{\phi}}^3 \psi)\nonumber\\
&& \left.+  {J}_1 (\partial_u \partial_{\theta^{\star}}^2 \psi) +  {J}_2 (\partial_v \partial_{\theta^{\star}}^2 \psi)+ {J}_3 (\partial_{\theta^{\star}}^3\psi)
\right]. 
\eea
Please note, that the coefficients $E_1, E_2,..$ etc.~ are not the same as in \eqref{error_structure2}.

\subsection{Relevant terms appearing in the error terms}
\lb{S_error}
Looking at \eqref{comm}, it is evident, that the higher order terms of the error terms are defined by the following:
\bea
\lb{pi-Y-termevv}
(\pi^Y)^{vv}&=&(\pi^Y)^{uu}=(\pi^Y)^{v{\theta_C}}=0,\\
(\pi^Y)^{uv}&=&-\frac{1}{2\Omega^2}\frac{\partial_{\theta^{\star}} \Omega }{\Omega}Y^{\theta^{\star}} ,\\
(\pi^Y)^{u {\theta_C}}&=&-\frac{b^{\tilde{\phi}}}{4\Omega^2}\partial_{\tilde{\phi}}Y^{\theta_C}-\frac{b^{\theta_C}}{2\Omega^2}\frac{\partial_{\theta^{\star}} \Omega }{\Omega}Y^{\theta^{\star}}\nonumber\\
&&+\frac{1}{4\Omega^2}(\gin^{-1})^{{\theta_C}{\theta_D}}\partial_{\theta^{\star}}b_{\theta_D} Y^{\theta^{\star}}-\frac{b^{\theta_D}}{4\Omega^2}(\gin^{-1})^{{\theta_C}{\theta_A}}\partial_{\theta^{\star}}\gin_{{\theta_D}{\theta_A}} Y^{\theta^{\star}},\\
\lb{pi-Y-termeAB}
(\pi^Y)^{{\theta_C} {\theta_D}}&=& \frac12(\gin^{-1})^{{\theta_C}{\theta_A}}\partial_{{\theta_A}}Y^{\theta_D}+ \frac12(\gin^{-1})^{{\theta_C}{\theta_D}}\partial_{{\theta_A}}Y^{\theta_C}\nonumber\\
&&+\frac12 (\gin^{-1})^{{\theta_C}{\theta_A}}(\gin^{-1})^{{\theta_D}{\theta_B}}\partial_{\theta^{\star}} \gin_{{\theta_A}{\theta_B}} Y^{\theta^{\star}},
\eea
where $Y$ is as in Section \ref{angular}.

\section{The wave equation and higher derivatives}
\lb{wave_equation}
The wave on fixed background is given by
\bea 
\lb{box_curved} 
\Box_g \psi=\frac{1}{\sqrt{-g}} \frac{\partial}{\partial
x^{\mu}}\left(g^{\mu \nu}\sqrt{-g}\frac{\partial
\psi}{\partial x^{\nu}}\right), 
\eea 
where $\sqrt{-g}=2\Omega^2L\sin \theta$.
Therefore, 
the wave equation in Eddington--Finkelstein-like coordinates on fixed Kerr background yields
\bea
\lb{wave_eq_eddf}
\Box_g \psi&=&\frac{1}{2\Omega^2}\left[-\frac{\partial_u |L\sin \theta|}{|L\sin\theta|}\partial_v \psi-\frac{\partial_v |L\sin \theta|}{|L\sin\theta|}\partial_u \psi-\frac{\partial_u |b^{\tilde{\phi}}L\sin \theta|}{|L\sin\theta|}\partial_{\tilde{\phi}} \psi\right]\nonumber\\
&&-\frac{1}{\Omega^2}\partial_u\partial_v \psi-\frac{b^{\tilde{\phi}}}{\Omega^2}\partial_u\partial_{\tilde{\phi}} \psi+\Dell\psi\nonumber\\
&&+\frac{2\partial_{\theta^{\star}}\Omega}{\Omega}\left[\frac{R^2}{L^2}\partial_{\theta^{\star}}\psi-\left(\frac{\partial h}{\partial \theta^{\star}}\right)\frac{R^2}{L^2}\partial_{\tilde{\phi}}\psi\right]=0,
\eea
where
\bea 
\lb{dell_curved} 
\Dell \psi=\frac{1}{\sqrt{-\gin}} \frac{\partial}{\partial
x^{{\theta_C}}}\left(\gin^{{\theta_C} {\theta_D}}\sqrt{-\gin}\frac{\partial
\psi}{\partial x^{{\theta_D}}}\right), 
\eea 
with $\sqrt{-\gin}=L\sin \theta$ and $C,D=1,2$ and $\theta_1=\theta^{\star}$, $\theta_2=\tilde{\phi}$.
So we get
\bea 
\lb{dell_expression} 
\Dell \psi&=&\frac{\partial_{\theta^{\star}}|L\sin\theta|}{|L\sin\theta|}\left[\frac{R^2}{L^2}\partial_{\theta^{\star}}\psi-\left(\frac{\partial h}{\partial \theta^{\star}}\right)\frac{R^2}{L^2}\partial_{\tilde{\phi}}\psi\right]\nonumber\\
&&+\partial_{\theta^{\star}}\left(\frac{R^2}{L^2}\right)\partial_{\theta^{\star}}\psi-\partial_{\theta^{\star}}\left[\left(\frac{\partial h}{\partial \theta^{\star}}\right)\frac{R^2}{L^2}\right]\partial_{\tilde{\phi}}\psi\nonumber\\
&&+\frac{R^2}{L^2}\partial_{\theta^{\star}}\partial_{\theta^{\star}}\psi-2\left(\frac{\partial h}{\partial \theta^{\star}}\right)\frac{R^2}{L^2}\partial_{\theta^{\star}}\partial_{\tilde{\phi}}\psi\nonumber\\
&&+\left(\frac{1}{R^2\sin\theta}+\left(\frac{\partial h}{\partial \theta^{\star}}\right)\frac{R^2}{L^2}\right)\partial_{\tilde{\phi}}\partial_{\tilde{\phi}}\psi.
\eea 
Solving \eqref{wave_eq_eddf} to
\bea
\lb{uv_term}
\frac{1}{\Omega^2}\partial_u\partial_v \psi+\frac{b^{\tilde{\phi}}}{\Omega^2}\partial_u\partial_{\tilde{\phi}} \psi&=&\frac{1}{2\Omega^2}\left[-\frac{\partial_u |L\sin \theta|}{|L\sin\theta|}\partial_v \psi-\frac{\partial_v |L\sin \theta|}{|L\sin\theta|}\partial_u \psi-\frac{\partial_u |b^{\tilde{\phi}}L\sin \theta|}{|L\sin\theta|}\partial_{\tilde{\phi}} \psi\right]\nonumber\\
&&+\Dell\psi+\frac{2\partial_{\theta^{\star}}\Omega}{\Omega}\left[\frac{R^2}{L^2}\partial_{\theta^{\star}}\psi-\left(\frac{\partial h}{\partial \theta^{\star}}\right)\frac{R^2}{L^2}\partial_{\tilde{\phi}}\psi\right],
\eea
will enable us to control the mixed term $\partial_u\partial_v \psi$ by using the right hand side of the equation, once it is shown that all coefficients are bounded. This will be needed in order to estimate error terms.

\end{appendix}

 
\end{document}